\documentclass[11pt,a4paper]{article}
\pdfoutput=1
\usepackage{jheppub}
\usepackage{amssymb,amsfonts,amsmath,amsthm}
\usepackage{enumerate}
\usepackage{tikz}
\usepackage{subcaption}
\usepackage{epsfig}
\usepackage{titlesec}
\titlelabel{\thetitle.\!\!\quad}
\let \savenumberline \numberline
\def \numberline#1{\savenumberline{#1.}}
\makeatletter
\def\@fpheader{\relax}
\makeatother

\newcommand{\R}{{\bf R}}

\newcommand{\CF}{{\cal F}}

\newcommand{\CJ}{{\cal J}}

\newcommand{\CN}{{\cal N}}
\newcommand{\CO}{{\cal O}}
\newcommand{\CP}{{\cal P}}
\newcommand{\CT}{{\cal T}}

\newcommand{\bk}{{\bf k}}
\newcommand{\bx}{{\bf x}}

\newcommand{\p}{\partial}

\renewcommand{\tilde}[1]{\widetilde{#1}}
\newcommand{\be}{\begin{equation}}
\newcommand{\ee}{\end{equation}}
\newcommand{\bea}{\begin{eqnarray}}
\newcommand{\eea}{\end{eqnarray}}

\newcommand{\diff}{{\rm Diff}}
\newcommand{\muir}{{\mu^{\ }_{\rm IR}}}
\usepackage{graphicx}
\usepackage{latexsym}

\newtheorem{Thm}{Theorem}
\newtheorem{Cor}{Corollary}
\newtheorem{Def}{Definition}
\newtheorem{Lem}{Lemma}
\newtheorem{Prop}{Proposition}
\newcommand{\nd}{n}
\newcommand{\nf}{F}
\newcommand{\nv}{N}
\newcommand{\nbv}{N(\bullet)}
\newcommand{\nbsv}{N^{\bullet} (\star)}
\newcommand{\nsv}{N(\star)}
\newcommand{\ntv}{N( \times )}
\newcommand{\pd}{P}
\newcommand{\sdim}{D}
\newcommand{\red}{black}
\newcommand{\blue}{black}
\newcommand{\Tred}{T_\text{A}}
\newcommand{\Tblue}{T_\text{B}}
\newcommand{\1}{0.8}
\newcommand{\mini}{1.4}
\newcommand{\nts}{Q}
\title{\ \vspace{1.6cm} \\ Scalar \hspace{-0.05cm}Field \hspace{-0.05cm}Theories \hspace{-0.05cm}with \hspace{-0.05cm}Polynomial \hspace{-0.05cm}Shift \hspace{-0.05cm}Symmetries}
\author{Tom Griffin${}^a$, Kevin T. Grosvenor${}^{b,c}$, 
Petr Ho\v{r}ava${}^{b,c}$ and Ziqi Yan${}^{b,c}$} 
\affiliation{${}^a$Blackett Laboratory, Department of Physics\\
Imperial College, London, SW7 2AZ, UK\medskip\\ 
${}^b$Berkeley Center for Theoretical Physics and Department of Physics\\
University of California, Berkeley, CA, 94720-7300, USA\medskip\\ 
${}^c$Theoretical Physics Group, Lawrence Berkeley National Laboratory\\
Berkeley, CA 94720-8162, USA}
\abstract{We continue our study of naturalness in nonrelativistic QFTs of the 
Lifshitz type, focusing on scalar fields that can play the role of 
Nambu-Goldstone (NG) modes associated with spontaneous symmetry breaking.  Such 
systems allow for an extension of the constant shift symmetry to a shift by a 
polynomial of degree $\pd$ in spatial coordinates.  These ``polynomial shift 
symmetries'' in turn protect the technical naturalness of modes with 
a higher-order dispersion relation, and lead to a refinement of the proposed 
classification of infrared Gaussian fixed points available to describe 
NG modes in nonrelativistic theories.  Generic interactions in such theories 
break the polynomial shift symmetry explicitly to the constant shift.  It is 
thus natural to ask:  Given a Gaussian fixed point with polynomial shift 
symmetry of degree $\pd$, what are the lowest-dimension operators that preserve 
this symmetry, and deform the theory into a self-interacting scalar field 
theory with the shift symmetry of degree $\pd$?  To answer this (essentially 
cohomological) question, we develop a new graph-theoretical technique, and use 
it to prove several classification theorems.  First, in the special case of 
$\pd =1$ (essentially equivalent to Galileons), we reproduce the known Galileon 
$\nv$-point invariants, and find their novel interpretation in terms of graph 
theory, as an equal-weight sum over all labeled trees with $\nv$ 
vertices.  Then we extend the classification to $\pd >1$ and find a whole host 
of new invariants, including those that represent the most relevant 
(or least irrelevant) deformations of the corresponding Gaussian fixed points, 
and we study their uniqueness.}
\begin{document}
\maketitle
\section{Introduction: Landscapes of Naturalness}

Some of the most fundamental questions of modern theoretical physics can 
be formulated as puzzles of naturalness \cite{th}.  Why is the observed 
cosmological constant so small compared to the Planck scale?  Why is the 
observed Higgs mass so small compared to any high particle-physics scale 
(be it the quantum gravity scale, or the scale of grand unification, or some 
other scale of new physics)?  In both cases, the expected quantum corrections 
estimated in the framework of relativistic effective field theory (EFT) 
predict natural values 
at a much higher scale, many orders of magnitude larger than the observed 
ones.  The principle of naturalness is rooted in the time-honored physical 
principles of causality and the hierarchy of energy scales from the 
microscopic to the macroscopic.  It is conceivable that some puzzles of 
naturalness may only have environmental explanations, based on the landscape 
of many vacua in the multiverse.  However, before we give up naturalness as 
our guiding principle, it is important to investigate more systematically the 
``landscape of naturalness'':  To map out the various quantum systems and 
scenarios in which technical naturalness does hold, identifying possible 
surprises and new pieces of the puzzle that might help restore the power of 
naturalness in fundamental physics.

One area in which naturalness has not yet been fully explored is 
nonrelativistic gravity theory \cite{lif,mqc}.  This approach to quantum 
gravity 
has attracted a lot of attention in recent years, largely because of its 
improved quantum behavior at short distances, novel phenomenology at long 
distances \cite{shinji}, its connection to the nonperturbative Causal Dynamical 
Triangulations approach to quantum gravity \cite{ajlh,gnx,carlip}, as well as 
for its applications to holography and the AdS/CFT correspondence of 
nonrelativistic systems \cite{jkone,jktwo,lglh}.  This 
area of research in quantum gravity is still developing rapidly, with new 
surprises already encountered and other ones presumably still awaiting 
discovery.   Mapping out the quantum structure of nonrelativistic gravity 
theories, and in particular investigating the role of naturalness, represents 
an intriguing and largely outstanding challenge.  

Before embarking on a systematic study of the quantum properties of 
nonrelativistic gravity, one 
can probe some of the new conceptual features of quantum field theories (QFTs) 
with 
Lifshitz-type symmetries in simpler systems, without gauge symmetry, dynamical 
gravity and fluctuating spacetime geometry.  In \cite{msb}, we considered one 
of the ubiquitous themes of modern physics:  The phenomenon of spontaneous 
symmetry breaking, in the simplest case of global, continous internal 
symmetries.  According to Goldstone's theorem, spontaneous breaking of such 
symmetries implies 
the existence of a gapless Nambu-Goldstone (NG) mode in the system.  For 
Lorentz invariant systems, the relativistic version of Goldstone's theorem 
is stronger, and we know more:  There is a one-to-one correspondence between 
broken symmetry generators and the NG modes, whose gaplessness implies that 
they all share the same dispersion relation $\omega=ck$. On the other hand, 
nonrelativistic systems are phenomenologically known to exhibit a more complex 
pattern:  Sometimes, the number of NG modes is smaller than the number of 
broken symmetry generators, and sometimes they disperse quadratically instead 
of linearly.  This rich phenomenology opens up the question of a full 
classification of possible NG modes.  A natural and elegant approach to this 
problem has been pursued in \cite{mw}:  In order to classify NG modes, one 
classifies the low-energy EFTs available to control their 
dynamics.  

In the case of systems with nonrelativistic Lifshitz symmetries, this approach 
suggests a classification of NG modes into two categories \cite{mw}: Type A, 
and Type B NG modes.  Each Type A mode is associated with a single broken 
symmetry generator, and each Type B mode with a pair.  Upon closer inspection 
\cite{msb}, it turns out that even such simple examples of Lifshitz-type 
QFTs exhibit rich and surprising features, often contrasting or 
contradicting the intuition developed in relativistic QFTs.  Our analysis of 
naturalness in the patterns of spontaneous symmetry breaking in systems 
with Lifshitz symmetries has revealed a refined hierarchy of the Type A and B 
universality classes of NG dynamics, with rich low-energy phenomenology 
dominated by multicritical NG modes whose dispersion is of higher degree in 
momentum.  These results shed some new, and perhaps surprising, light 
on the concept of naturalness in nonrelativistic quantum field theory. 
However, as usual, the naturalness of the multicritical dispersion relation 
turns out to be protected by a symmetry.  This new kind of symmetry is 
generated by the shifts of the NG fields by a polynomial in spatial coordinates.\footnote{This symmetry is 
a natural generalization of two types of symmetries well-known in the 
literature: The famous constant shift symmetry observed in systems with 
relativistic NG modes, and the shift linear in the spacetime coordinates known 
from the relativistic theory of Galileons \cite{gal}.}

In this paper, we continue our investigation of scalar field theories with 
such polynomial shift symmetries of degree $\pd$.  
The paper is organized into several relatively self-contained blocks.  
In \S{\ref{sec: mng}}, we review the physics background and discuss the 
structure of multicritical symmetry breaking in nonrelativistic systems of 
the Lifshitz type, summarizing and expanding on the findings of 
\cite{msb,cmu}.  We discuss the Goldstone theorem in the nonrelativistic 
regime, and give the refined classification of Nambu-Goldstone modes for 
systems with Lifshitz symmetries  into two hierarchies of multicritical fixed 
points of Type A${}_\nd$ and B${}_{2 \nd}$, 
with $\nd =1,2,\ldots$.  We present the nonrelativistic analog of the 
Coleman-Hohenberg-Mermin-Wagner (CHMW) theorem, and discuss its implications 
for the dynamics of the multicritical NG modes.  Throughout, we stress 
the role played by the polynomial shift symmetry, as an approximate symmetry 
restored at the Gaussian infrared fixed points.  In many general examples of 
multicritical symmetry breaking, the polynomial shift symmetry is broken by 
the self-interactions of the NG modes.  It is then natural to ask:  What if 
we impose the polynomial shift symmetry as an exact symmetry?  What is the 
lowest-dimension operator that can be added to the action while preserving 
the symmetry?  This is the task we address in the remainder of the paper.  

The classification of Lagrangians invariant under the polynomial shift of 
degree $\pd$ up to a total derivative (which we will refer to as 
``$\pd$-invariants'' for short) is essentially a cohomological problem.  
In \S{\ref{sec: Galileons}}, we consider the polynomial-shift invariants 
in the simplest case of linear shifts (i.e., ``1-invariants'').  In order to 
prepare for the general case of $\pd >1$, we develop a novel technique, based 
on graph theory.  Having rephrased the defining relation for the invariants 
into the language of graphs, we can address the classification problem using 
the abstract mathematical machinery of graph theory.  The basic ingredients 
of this technique are explained as needed in \S{\ref{sec: Galileons}} and 
\S{\ref{sec: beyond}}.  However, we relegate all the technicalities of the 
graphical technique into a self-contained Appendix~\ref{AppB} (preceded by 
Appendix~\ref{graph theory}, in which we offer a glossary of the basic terms 
from graph theory).  Appendix~\ref{AppB} is rather mathematical in nature, as 
it contains a systematic exposition of all our definitions, theorems and 
proofs that we found useful in the process of generating the invariants 
discussed in the body of the paper.  The good news is that Appendix~B is not 
required for the understanding of the results presented in 
\S{\ref{sec: Galileons}} and \S{\ref{sec: beyond}}: Once the invariants have 
been found using the techniques in Appendix~\ref{AppB}, their actual 
invariance can be checked by explicit 
calculation (for example, on a computer).  In this sense, the bulk of the 
paper (\S{\ref{sec: Galileons}} and \S{\ref{sec: beyond}}) is also 
self-contained, and can be read independently of the Appendices.  

The $\nv$-point 1-invariants discussed in \S{\ref{sec: Galileons}} are known 
in the literature, where they have been generated in the closely related 
context of the relativistic Galileon theories \cite{gal}.  While it is 
reassuring to see that our graph-theoretical technique easily reproduces these 
known {1-invariants}, the novelty of our results presented in 
\S{\ref{sec: Galileons}} lies elsewhere:  We find 
a surprisingly simple and elegant interpretation of the known $\nv$-point 
1-invariants in the language of graphs.  They are simply given 
by the equal-weight sum over all labeled trees with $\nv$ vertices!  

In \S{\ref{sec: beyond}} we move beyond the 1-invariants, and initiate 
a systematic study of $\pd$-invariants with $\pd >1$, organized in the order 
of their scaling dimension.  We find several series of invariants; some of 
them we prove to be the unique and most relevant (or, more accurately, least 
irrelevant) $\nv$-point $\pd$-invariants, while others represent hierarchies 
of $\pd$-invariants of higher dimensions.  We also show how to construct 
higher $\pd$-invariants from superposing several graphs that represent 
invariants of lower $\pd$.  Appendix~\ref{AppC} contains a brief 
discussion of the connection between our invariant Lagrangians and the 
Chevalley-Eilenberg Lie algebra cohomology theory.  In 
\S{\ref{sec: conclusion}} we present our conclusions.

\section{M\hspace{-0.007cm}u\hspace{-0.007cm}l\hspace{-0.007cm}t\hspace{-0.007cm}i\hspace{-0.007cm}c\hspace{-0.007cm}r\hspace{-0.007cm}i\hspace{-0.007cm}t\hspace{-0.007cm}i\hspace{-0.007cm}c\hspace{-0.007cm}a\hspace{-0.007cm}l \hspace{-0.075cm}N\hspace{-0.007cm}a\hspace{-0.007cm}m\hspace{-0.007cm}b\hspace{-0.007cm}u\hspace{-0.007cm}-\hspace{-0.007cm}G\hspace{-0.007cm}o\hspace{-0.007cm}l\hspace{-0.007cm}d\hspace{-0.007cm}s\hspace{-0.007cm}t\hspace{-0.007cm}o\hspace{-0.007cm}n\hspace{-0.007cm}e \hspace{-0.075cm}B\hspace{-0.007cm}o\hspace{-0.007cm}s\hspace{-0.007cm}o\hspace{-0.007cm}n\hspace{-0.007cm}s \hspace{-0.075cm}a\hspace{-0.007cm}n\hspace{-0.007cm}d \hspace{-0.075cm}P\hspace{-0.007cm}o\hspace{-0.007cm}l\hspace{-0.007cm}y\hspace{-0.007cm}n\hspace{-0.007cm}o\hspace{-0.007cm}m\hspace{-0.007cm}i\hspace{-0.007cm}a\hspace{-0.007cm}l \hspace{-0.075cm}S\hspace{-0.007cm}h\hspace{-0.007cm}i\hspace{-0.007cm}f\hspace{-0.007cm}t \hspace{-0.075cm}S\hspace{-0.007cm}y\hspace{-0.007cm}m\hspace{-0.007cm}m\hspace{-0.007cm}e\hspace{-0.007cm}t\hspace{-0.007cm}r\hspace{-0.007cm}i\hspace{-0.007cm}e\hspace{-0.007cm}s} \label{sec: mng}

Some surprising features of naturalness in the regime of nonrelativistic field 
theories are illustrated by considering one of the classic problems in physics: 
The classification of NG modes associated with possible patterns 
of spontaneous breaking of continuous global internal symmetries.  In this 
section, we summarize and explain the results found in \cite{msb,cmu}, which 
lead to a refinement in the classification of NG modes in systems with 
Lifshitz symmetries, characterized by a multicritical behavior which is 
technically natural, and protected by a symmetry.  

\subsection{Geometry of the Spacetime with Lifshitz Symmetries} 
\label{subset: geometry}

For clarity and simplicity, as in \cite{msb}, in this paper we focus on 
systems on the flat spacetime with Lifshitz spacetime symmetries.  
We define this spacetime to be $M=\R^{\sdim+1}$ with a preferred foliation 
$\CF$ by fixed spatial slices $\R^\sdim$, and equipped with a flat metric.  Such a 
spacetime with the preferred foliation $\CF$ would for example appear as 
a ground-state solution of nonrelativistic gravity \cite{lif} whose gauge 
symmetry is given by the group of foliation-preserving spacetime 
diffeomorphisms,  $\diff(M,\CF)$ (or a nonrelativistic extension thereof 
\cite{genc}).  It is useful to 
parametrize $M$ by coordinates $(t,\bx=\{x^i,i=1,\ldots \sdim\})$, such that 
the leaves of $\CF$ are the leaves of constant $t$, and the metric has the 
canonical form 
\be
\label{camet}
g_{ij}(t,\bx)=\delta_{ij},\quad\CN(t,\bx)=1,\quad \CN_i(t,\bx)=0
\ee
(here $g_{ij}$ is the spatial metric on the leaves of $\CF$, $\CN$ is the lapse 
function, and $\CN_i$ the shift vector).  

The isometries of this spacetime are, by definition, those elements of 
$\diff(M,\CF)$ that preserve this flat metric \cite{aci}.  Explicitly, the 
connected component of this isometry group is generated by infinitesimal 
spatial rotations and spacetime translations, 
\be
\label{lgtr}
\delta t=b,\qquad\delta x^i=\omega^i_{\phantom{i}j} \, x^j+b^i,\qquad\omega_{ij}=-\omega_{ji}.
\ee
At fixed points of the renormalization group, systems with Lifshitz isometries 
develop an additional scaling symmetry, generated by 
\be
\label{scalez}
\delta x^i=\lambda x^i,\qquad \delta t=z\lambda t.
\ee
The dynamical critical exponent $z$ is an important observable associated with 
the fixed point, and characterizes the degree of anisotropy between space and 
time at the fixed point.      

The connected component of the group of isometries of our spacetime $M$ with 
the flat metric (\ref{camet}) is generated by (\ref{lgtr}), and we will refer 
to it as the ``Lifshitz symmetry'' group.%
\footnote{It would be natural to refer to $M$ with the flat metric 
(\ref{camet}) as the ``Lifshitz spacetime''.  Unfortunately, this term already 
has another widely accepted meaning in the holography literature, where it 
denotes the curved spacetime geometry in one dimension higher, whose 
isometries realize the Lifshitz symmetries (\ref{lgtr}) plus the Lifshitz 
scaling symmetry (\ref{scalez}) for some fixed value of $z$ \cite{kachru}.}
The full isometry group of this spacetime has four disconnected components, 
which can be obtained by combining the Lifshitz symmetry group generated by 
(\ref{lgtr}) with two discrete 
symmetries:  The time-reversal symmetry $\CT$, and a discrete symmetry $\CP$ 
that reverses the orientation of space.  In this paper, we shall be interested 
in systems that are invariant under the Lifshitz symmetry group.  Note that 
this mandatory Lifshitz symmetry does not contain either the discrete 
symmetries $\CT$ and $\CP$, or the anisotropic scaling symmetry (\ref{scalez}).

\subsection{Effective Field Theories of Type A and B Nambu-Goldstone Bosons}

We are interested in the patterns of spontaneous symmetry breaking of global 
continuous internal symmetries in the flat spacetime with the Lifshitz 
symmetries, as defined in the previous paragraph.  Our analysis gives an 
example of phenomena that are novel to Goldstone's theorem 
in nonrelativistic settings, and can in principle be generalized to 
nonrelativistic systems with even less symmetry.   

An elegant strategy has been proposed in \cite{mw}:  In order to classify 
Nambu-Goldstone modes, we can classify the corresponding EFTs available to 
describe their low-energy dynamics.  In this EFT approach, we organize the 
terms in the effective action by their increasing dimension.  Such 
dimensions are defined close enough to the infrared fixed point.  However, 
until we identify the infrared fixed point, we don't a priori know the value 
of the dynamical critical exponent, and hence the relative dimension of 
the time and space derivatives -- it is then natural to count the time 
derivatives and spatial derivatives separately.  Consider first the 
``potential terms'' 
in the action, i.e., terms with no time derivatives.  The general statement 
of Goldstone's theorem implies that non-derivative terms will be absent, and 
the spatial rotational symmetry further implies that (for $D > 1$) all derivatives will 
appear in pairs contracted with the flat spatial metric.  Hence, we can write 
the general ``potential term'' in the action as 
\be
\label{seffv}
S_{{\rm eff},\, V}=\int dt\,d\bx\left\{\frac{1}{2}g_{IJ}(\pi)\p_i\pi^I\p_i\pi^J
+\ldots\right\}
\ee
where $g_{IJ}(\pi)$ is the most general metric on the vacuum manifold which 
is compatible with all the global symmetries, and $\ldots$ stand for all 
the terms of higher order in spatial derivatives.  

If the system is also invariant under the primitive version $\CT$ of time 
reversal, defined as the transformation that acts trivially on fields, 
\be
\label{trev}
\CT:\left\{\begin{array}{ccc}
t&\to&-t,\\
\pi^I&\to&\pi^I,\\
\end{array}
\right.
\ee
the time derivatives will similarly have to appear in pairs, and the kinetic 
term will be given by
\be
\label{seffk}
S_{{\rm eff},\, K}=\int dt\,d\bx\left\{\frac{1}{2}h_{IJ}(\pi)\dot\pi^I\dot\pi^J
+\ldots\right\},
\ee
where again $h_{IJ}$ is a general metric on the vacuum manifold compatible 
with all symmetries, but not necessarily equal to the $g_{IJ}$ that appeared 
in (\ref{seffv}); and $\ldots$ are higher-derivative terms.  

However, invariance under $\CT$ is not mandated by the Lifshitz symmetry.  
If it is absent, the Lifshitz symmetries allow a new, more relevant kinetic 
term, 
\be
\label{somega}
\tilde S_{{\rm eff},\, K}=\int dt\,d\bx\left\{\Omega_I(\pi)\dot\pi^I
+\ldots\right\},
\ee
assuming one can define the suitable object $\Omega_I(\pi)$ on the vacuum 
manifold so that all the symmetry requirements are satisfied, and 
$\Omega_I(\pi)\dot\pi^I$ is not a total derivative.  Since $\Omega_I(\pi)$ 
plays the role of the canonical momentum conjugate to $\pi^I$, if such 
$\Omega$-terms are present in the action, they induce a natural canonical 
pairing on an even-dimensional subset of the coordinates on the vacuum 
manifold.

In specific dimensions, new terms in the effective action that are odd under 
spatial parity $\CP$ may exist.  For example, in $\sdim =2$ spatial dimensions, 
we can add new terms to the ``potential'' part of the action, of the form
\be
\tilde S_{{\rm eff},\, V}=\int dt\,d\bx\left\{\frac{1}{2} \, \Omega_{IJ}(\pi) \,
\varepsilon_{ij} \, \p_i\pi^I\p_j\pi^J+\ldots\right\},
\ee
where $\Omega_{IJ}$ is any two-form on the vacuum manifold that respects all 
the symmetries.%
\footnote{For example, if $\Omega_I(\pi)$ suitable for (\ref{somega}) exist, 
one can take $\Omega_{IJ}=\p_{[I}\Omega_{J]}$.}
In the interest of simplicity, we wish to forbid such terms, and will do so 
by imposing the $\CP$ invariance of the action, focusing on the symmetry 
breaking patterns that respect spatial parity.  This condition can of course 
be easily relaxed, without changing our conclusions significantly.  

This structure of low-energy effective theories suggests the following 
classification of NG modes, into two general types:

\begin{itemize}
\item
Type A: One NG mode per broken symmetry generator (not paired by 
$\Omega_I$).  The low-energy dispersion relation is linear, $\omega\propto 
k$.  
\item
Type B: One NG mode per each pair of broken symmetry generators (paired by 
$\Omega_I$).  The low-energy dispersion relation is quadratic,  $\omega\propto 
k^2$.  
\end{itemize}

\noindent In general, Type A and Type B NG modes may coexist in one system.  
Some examples from condensed matter theory can be found in \cite{mw}.

Based on the intuition developed in the context of relativistic quantum field 
theory, one might be tempted to conclude that everything else would be fine 
tuning, as quantum corrections would be likely to generate large terms 
of the form (\ref{seffv}) in the effective action if we attempted to tune 
them to zero.  

\subsection{Naturalness of Slow Nambu-Goldstone Modes}

Our careful study of a number of explicit examples revealed \cite{msb} that 
the naive intuition about fine-tuning summarized in the previous paragraph 
is incorrect.  It turns out that the leading spatial-derivative term in 
(\ref{seffv}) can be naturally small (or even zero), as we illustrated in 
\cite{msb} by explicit calculations of loop corrections in a series of 
examples.  The leading contribution to $S_{{\rm eff},\,V}$ then comes at fourth 
order in spatial derivatives; schematically,  
\be
S_{{\rm eff},\, V}=\int dt\,d\bx\left\{\frac{1}{2}g_{IJ}(\pi)\p^2\pi^I\p^2\pi^J
+\ldots\right\},
\ee
where the ``$\ldots$'' stand for all other terms of order four and higher 
in $\p_i$.  The dispersion relation of this NG mode at the Gaussian infrared 
fixed point is then $\omega\propto k^2$ (or $\omega\propto k^4$), if the 
kinetic term is of Type A as in (\ref{seffk}) (or of Type B as in 
(\ref{somega})).  The reason why this behavior does not require fine tuning 
is simple \cite{msb}:  As we approach this new Gaussian infrared fixed point, 
the theory develops a new enhanced symmetry.  Specifically, the symmetry that 
protects the Type A NG modes with the quadratic dispersion relation is a 
generalization of the constant shift symmetry of conventional NG modes:  
The generators of the new symmetry act by shifting each field component 
$\pi^I$ by a quadratic polynomial in spatial coordinates,
\be
\label{quadsh}
\delta\pi^I(t,\bx)=a^I_{ij}x^ix^j+a^I_ix^i+a^I_0.
\ee
The leading, quadratic part of this symmetry forbids the term (\ref{seffv}) 
allowing only terms of fourth order in $\p_i$ and higher to appear in the 
action in the free-field limit.  The subleading linear and constant terms 
have been included in (\ref{quadsh}) because they would be generated anyway by 
the action of spatial translations and rotations, which are a part of the 
assumed Lifshitz symmetry of the system.  Similarly, quadratic shift symmetries can
also be extended to the Gaussian limit of Type B NG modes. At the Gaussian fixed point,
$\Omega_I ( \pi )$ of \eqref{somega} reduces to a linear function of $\pi$, such that
$\tilde{S}_{\text{eff}, \, K}$ is invariant under the quadratic shift up to a total derivative,
and the extra shift symmetry yields Type B NG modes with a quartic dispersion
relation.

This construction can obviously be iterated.  The quadratic shift symmetry 
(\ref{quadsh}) can be promoted to a polynomial shift symmetry by polynomials 
of degree $\pd =2,3,\ldots$, leading to a natural protection of higher-order 
dispersion relations $\omega\propto k^\nd$ (or $\omega\propto k^{2 \nd}$) for 
Type A (or Type B) NG modes.  

\subsection{Polynomial Shift Symmetries}
\label{sec:pss}

Since the polynomial shift symmetries act on the fields $\pi^I(t,\bx)$ 
separately component by component, from now on we shall focus 
on just one field component, and rename it $\phi(t,\bx)$.  

The generators of the polynomial shift symmetry of degree $\pd$ act on 
$\phi$ by 
\be
\label{pssp}
\delta_\pd\phi=a_{i_1\ldots i_\pd}x^{i_1}\cdots x^{i_\pd}+\ldots + a_i x^i +a.
\ee
The multicritical Gaussian fixed point with dynamical exponent $z= \nd$ is 
described by
\be
\label{sn}
S_\nd =\int dt\,d\bx\,\left\{\frac{1}{2}\dot\phi^2-\frac{1}{2}\zeta^2_\nd
\left(\p_{i_1}\ldots\p_{i_\nd}\phi\right)\left(\p_{i_1}\ldots\p_{i_\nd}\phi\right)
\right\}.
\ee
In fact, it is a one-parameter family of fixed points, parametrized by 
the real positive coupling $\zeta^2_\nd$.  (Sometimes it is convenient to absorb 
$\zeta_\nd$ into the rescaling of space, and we will often do so when there is 
no competition between different fixed points.)  

The action $S_\nd$ is invariant under polynomial shift symmetries (\ref{pssp}) 
of degree 
$\pd \leq 2 \nd -1$:  It is strictly invariant under the symmetries of degree 
$\pd < \nd$, and invariant up to a total derivative for degrees $\nd \leq \pd \leq 2 \nd -1$.  

Morally, this infinite hierarchy of symmetries can be viewed as a natural 
generalization of the Galileon symmetry, proposed in \cite{gal} and much 
studied since, mostly in the cosmological literature.  In the case of the 
Galileons, the theory is relativistic, and the symmetry is linear in 
space-{\it time\/} coordinates.  The requirement of relativistic invariance 
is presumably the main reason that has precluded the generalization of the 
Galileon symmetries past the linear shift:  The higher polynomial shift 
symmetries in spacetime coordinates would lead to actions dominated by higher 
time derivatives, endangering perturbative unitarity.  

So far, we considered shifts by generic polynomials of degree $\pd$, whose 
coefficients $a_{i_1\ldots i_\ell}$ are arbitrary symmetric real tensors of 
rank $\ell$ for $\ell=0, \ldots, \pd$.  We note here in passing that for degrees 
$\pd\geq 2$, the polynomial shift symmetries allow an interesting refinement.   
To illustrate this feature, we use the example of the quadratic shift, 
\be
\delta_2\phi=a_{ij}x^ix^j+a_ix^i+a_0.
\ee
The coefficient $a_{ij}$ of the quadratic part is a general symmetric 
2-tensor.  It can be decomposed into its traceless part $\tilde a_{ij}$ and 
the trace part $a_{ii}$,
\be
a_{ij}=\tilde a_{ij}+\frac{1}{D}a_{kk}\delta_{ij}.  
\ee
Since this decomposition is compatible with the spacetime Lifshitz symmetries 
(\ref{lgtr}), one can restrict the symmetry group to be generated by a strictly 
smaller invariant subalgebra in the original algebra generated by $a_{ij}$.  
For example, setting the traceless part $\tilde a_{ij}$ of the quadratic shift 
symmetry to zero reduces the number of independent generators from 
$(D+2)(D+1)/2$ to $D+2$, but it is still sufficient to prevent 
$\p_i\phi \, \p_i\phi$ from being an invariant under the smaller symmetry.  
This intriguing pattern extends to $\pd>2$, leading to intricate hierarchies 
of polynomial shift symmetries whose coefficients $a_{i_1\ldots i_\ell}$ have been 
restricted by various invariant conditions.  As another example, invariance under the traceless part has been studied in \cite{extended_shift}.  In the interest of simplicity, 
we concentrate in the rest of this paper on the maximal case of 
polynomial shift symmetries with arbitrary unrestricted real coefficients 
$a_{i_1 \ldots i_\ell}$.  

The invariance of the action under each polynomial shift leads to a conserved 
Noether current.  Each such current then implies a set of Ward identities on 
the correlation functions and the effective action.  Take, for example, the 
case of $\nd =2$ in (\ref{sn}):  The currents for the infinitesimal shift by 
a general function $a(\bx)$ of the spatial coordinates $x^i$ are collectively 
given by
\be
\CJ_t =a(\bx)\dot\phi,\qquad\CJ_i=a(\bx)\p_i\p^2\phi-\p_j a(\bx)\p_i\p_j\phi
+\p_i\p_ja(\bx)\p_j\phi-\p_i\p^2a(\bx) \, \phi,
\ee
and their conservation requires
\be
\dot\CJ_t+\p_i\CJ_i\equiv a(\bx)\left\{\ddot\phi+(\p^2)^2\phi\right\}
-(\p^2)^2a(\bx) \, \phi=0.
\ee
The term in the curly brackets is zero on shell, and the current conservation 
thus reduces to the condition $(\p^2)^2a(\bx)\phi=0$, which is certainly 
satisfied by a polynomial of degree three, 
\be
a(\bx)=a_{ijk}x^ix^jx^k+a_{ij}x^ix^j+a_ix^i+a.  
\ee

Note that if we start instead with the equivalent form of the classical action
\be
\tilde S_2=\int dt\,d\bx\,\left\{\frac{1}{2}\dot\phi^2-\frac{1}{2}
\left(\p_i\p_i\phi\right)^2\right\},
\ee
the Noether currents will be related, as expected, by
\bea
\tilde\CJ_t&=&\CJ_t,\nonumber\\
\tilde\CJ_i&=&a(\bx)\p_i\p^2\phi-\p_ia(\bx)\p^2\phi
+\p^2a(\bx)\p_i\phi-\p_i\p^2a(\bx)\phi\\
&&\quad{}=\CJ_i+\p_j\left[\p_ia(\bx)\p_j\phi-\p_ja(\bx)\p_i\phi\right].\nonumber
\eea

From these conserved currents, one can formally define the charges 
\be
Q[a]=\int_\Sigma d\bx\,\CJ_t.
\ee
However, for infinite spatial slices $\Sigma=\R^\sdim$, such charges are all zero 
on the entire Hilbert space of states generated by the normalizable excitations 
of the fields $\phi$.  This behavior is quite analogous to the standard case 
of NG modes invariant under the constant shifts, and it simply indicates that 
the polynomial shift symmetry is being spontaneously broken by the vacuum.  

\subsection{Refinement of the Goldstone Theorem in the Nonrelativistic Regime}

In its original form, Goldstone's theorem guarantees the existence of a 
gapless mode when a global continuous internal symmetry is spontaneously 
broken.  However, in the absence of Lorentz symmetry, it does not predict 
the number of such modes, or their low-energy dispersion relation.  

The classification of the effective field theories which are available to 
describe the low-energy limit of the Nambu-Goldstone mode dynamics leads to a 
natural refinement of the Goldstone theorem in the nonrelativistic regime.  
In the specific case of spacetimes with Lifshitz symmetry, we get two 
hierarchies of NG modes:

\begin{itemize}
\item
Type A: One NG mode per broken symmetry generator (not paired by 
$\Omega_I$) The low-energy dispersion relation is $\omega\propto k^\nd$, 
where $\nd =1,2,3,\ldots$.
\item
Type B: One NG mode per each pair of broken symmetry generators (paired by 
$\Omega_I$). The low-energy dispersion relation is $\omega\propto k^{2 \nd}$, 
where $\nd =1,2,3,\ldots$. 
\end{itemize}

It is natural to label the members of these two hierarchies by the value of 
the dynamical critical exponent of their corresponding Gaussian fixed point.  
From now on, we will refer to these multicritical universality classes of 
Nambu-Goldstone modes as ``Type A${}_\nd$'' and ``Type B${}_{2 \nd}$'', respectively.

\begin{itemize}
\item[]
The following few comments may be useful:
\item[(1)]
While Type B NG modes represent a true infinite hierarchy of consistent 
fixed points, the Type A NG modes hit against the nonrelativistic analog of the 
Coleman-Hohenberg-Mermin-Wagner (CHMW) theorem:  At the critical value of 
$\nd = \sdim$, they develop infrared singularities and cease to exist as well-defined 
quantum fields.  We comment on this behavior further in \S\ref{secchmw}.
\item[(2)]
Type A preserve $\CT$ invariance, while Type B break $\CT$.  (This does not 
mean that a suitable time reversal invariance cannot be defined on Type B 
modes, but it would have to extend $\CT$ of (\ref{trev}) to act nontrivially 
on the fields.)
\item[(3)] 
Our classification shows the existence of A${}_\nd$ and B${}_{2 \nd}$ hierarchies of 
NG modes described by Gaussian fixed points, and therefore represents a 
refinement of the classifications studied in the literature so far.  However, 
it does not pretend to completeness:  We find it plausible that nontrivial 
fixed points (and fixed points at non-integer values of $\nd$) suitable for 
describing NG modes may also exist.  In this sense, the full classification 
of all possible types of nonrelativistic NG dynamics -- even under the 
assumption of Lifshitz symmetries -- still remains a fascinating open 
question.  
\item[(4)]
For simplicity, we worked under the assumption of spacetime Lifshitz symmetry.
Obviously, this simplifying restriction can be removed, and the classification 
of multicritical NG modes in principle extended to cases whereby some of the 
the spacetime symmetries are further broken by additional features of the 
system -- such as spatial anisotropy, layers, an underlying lattice structure, 
etc.  We also expect that the classification can be naturally extended to 
Nambu-Goldstone fermions associated with spontaneous breaking of symmetries 
associated with supergroups.  Such generalizations, however, are beyond the 
scope of this paper.
\end{itemize}

\subsection{Infrared Behavior and the Nonrelativistic CHMW Theorem}
\label{secchmw}

In this section, we consider the Type A${}_\nd$ and Type B${}_{2 \nd}$ hierarchies 
of NG modes, and their infrared behavior.  For simplicity, we will focus on theories that consist of Type A${}_\nd$ (or Type B${}_\nd$) NG modes with a fixed $\nd$, and leave the generalizations to interacting systems that mix different types of NG modes for future studies.

In relativistic systems, all NG bosons -- if they exist -- are Type A${}_1$. 
However, whether or not the corresponding symmetry is spontanously broken famously depends on the spacetime dimension.  This phenomenon is 
controlled by the celebrated theorem discovered independently in condensed 
matter by Mermin and Wagner \cite{mermw} and by Hohenberg \cite{hoh}, and in 
high-energy physics by Coleman \cite{cole}; we therefore refer to it, in the 
alphabetical order, as the Coleman-Hohenberg-Mermin-Wagner (CHMW) theorem.  

This famous CHMW theorem states that no spontaneous breaking of global 
continuous internal symmetries is possible in $1+1$ spacetime dimensions.  
The proof is beautifully simple:  $1+1$ represents the ``lower critical 
dimension'' of the massless scalar field $\phi$, defined as the dimension 
where $\phi$ is formally dimensionless at the Gaussian fixed point.  Quantum 
mechanically, this means that its propagator is logarithmically divergent, 
and we need to regulate it by introducing an infrared regulator $\muir$:
\be
\label{iras}
\langle\phi(x)\phi(0)\rangle=\int\frac{d^2k}{(2\pi)^2}\frac{e^{ik\cdot x}}{k^2
+\mu_{\rm IR}^2}\approx-\frac{1}{2\pi}\log(\muir|x|)+{\rm const.}
+\CO(\muir|x|).
\ee
The asymptotic expansion in (\ref{iras}), valid in the regime $\muir|x|\ll 1$, 
clearly shows that as we try to take $\muir\to 0$ the propagator stays 
sensitive at 
long scales to the infrared regulator $\muir$.  We can still construct 
various composite operators out of derivatives and exponentials of $\phi$, 
yielding consistent and finite renormalized correlation functions in the 
$\muir\to 0$ limit, but the field $\phi$ itself does not exist as a quantum 
object.  And since the candidate NG mode $\phi$ does not exist, the 
corresponding symmetry could never have been broken in the first place,
which concludes the proof.  

Going back to the general class of Type A${}_\nd$ NG modes, we find an intriguing 
nonrelativistic analog of the CHMW theorem.  The dimension of $\phi(t,\bx)$ at 
the A${}_\nd$ Gaussian fixed point in $\sdim + 1$ dimensions -- measured in the units 
of spatial momentum -- is
\be
[\phi(t,\bx)]^{\ }_{{\rm A}{}_\nd}=\frac{\sdim - \nd}{2}.
\ee
The Type A${}_\nd$ field $\phi$ is at its lower critical dimension when 
$\sdim = \nd$.  
Its propagator also requires an infrared regulator.  There are many ways 
how to introduce $\muir$ in this case, for example by
\be
\langle\phi(t,\bx)\phi(0)\rangle=\int\frac{d\omega\,d^\sdim \bk}{(2\pi)^{\sdim + 1}}
\frac{e^{i\bk\cdot\bx-i\omega t}}{\displaystyle{\omega^2+k^{2 \sdim}
+\mu_{\rm IR}^{2 \sdim}}},
\ee
or by
\be
\langle\phi(t,\bx)\phi(0)\rangle=\int\frac{d\omega\,d^\sdim \bk}{(2\pi)^{\sdim +1}}
\frac{e^{i\bk\cdot\bx-i\omega t}}{\displaystyle{\omega^2+(k^2+\mu_{\rm IR}^2)^\sdim}}.
\ee

Either way, as we try to take $\muir\to 0$, the asymptotics of the propagator 
again behaves logarithmically, both in space
\be
\langle\phi(t,\bx)\phi(0)\rangle\approx-\frac{1}{(4\pi)^{\sdim /2}\Gamma(
\sdim /2)}
\log(\muir|\bx|)+\ldots\qquad{\rm for}\ |\bx|^\sdim \gg t
\ee
and in time,
\be
\langle\phi(t,\bx)\phi(0)\rangle\approx-\frac{1}{(4\pi)^{\sdim /2} \sdim \, \Gamma( \sdim /2)}
\log(\mu_{\rm IR}^\sdim t)+\ldots\qquad{\rm for}\ |\bx|^\sdim \ll t.
\ee
Most importantly, the propagator remains sensitive to the infrared regulator 
$\muir$.  Consequently, we obtain the nonrelativistic, multicritical version 
of the CHMW 
theorem for Type A NG modes and their associated symmetry breaking:  

\begin{itemize}
\item[]
{\it The Type A${}_\nd$ would-be NG mode $\phi(t,\bx)$ at its lower critical 
dimension $\sdim = \nd$ exhibits a propagator which is logarithmically sensitive to 
the infrared regulator $\muir$, and therefore $\phi(t,\bx)$ does not exist as 
a quantum mechanical object.  Consequently, no spontaneous symmetry breaking 
with Type A${}_\nd$ NG modes is possible in $\sdim = \nd$ dimensions.}
\end{itemize}

By extension, this also invalidates all Type A${}_\nd$ would-be NG modes with 
$\nd > \sdim$:  Their propagator grows polynomially at long distances, destabilizing 
the would-be condensate and disallowing the associated symmetry breaking 
pattern. 

In contrast, in the Type B${}_{\nd}$ case (and assuming that all the NG field components are assigned the same dimension), we have
$$[\phi(t,\bx)]^{\ }_{{\rm B}{}_{2 \nd}}=\frac{\sdim}{2},$$
and the lower critical dimension is $\sdim =0$.  Hence, in all dimensions $\sdim > 0$, 
the Type B${}_{2 \nd}$ NG modes are free of infrared divergences and well-defined 
quantum mechanically for all $\nd =1,2,\ldots$, and the Type B 
nonrelativistic, multicritical CHMW theorem is limited to the following statement:
\begin{itemize}
\item[]
{\it The Type B${}_{2 \nd}$ symmetry breaking is possible in any $\sdim >0$ and for 
any $\nd =1,2,\ldots$.}
\end{itemize}

In the special cases for Type A${}_2$ and Type B${}_2$ NG modes, the multicritical CHMW theorems stated above reproduce the results reported in \cite{mw2}.

\subsection{Cascading Multicriticality}

The conclusions of the nonrelativistic CHMW theorem appear rather 
unfavorable for Type A${}_\nd$ NG modes with $\nd \geq \sdim$.  However, unlike in the 
relativistic case of $\nd =1$ in $1+1$ dimensions, the nonrelativistic systems 
offer an intriguing way out \cite{cmu}, as we now illustrate for the case of 
the lower critical dimension $\sdim = \nd$, with $\sdim >1$.    

At this Gaussian fixed point, the propagator for $\phi$ is logarithmically 
sensitive to the infrared regulator.  However, all is not lost -- unlike in 
the relativistic case, the system can now provide its own natural infrared 
regulator, and flow under the influence of some of the relevant terms to 
another infrared fixed point of Type A${}_{\nd'}$, with a lower value of 
$\nd' < \nd$.  And we know that this phenomenon can be arranged to happen 
hierarchically, in a pattern protected by the hierarchical breaking of the 
polynomial symmetries.  Thus, we can break the polynomial shift symmetry 
at a high energy scale $\mu$ only partially, to a polynomial symmetry of a 
lower degree which is then broken at a lower energy scale $\mu'$.  This 
process can continue until at some low scale $\mu''$ the symmetry is broken 
all the way to the constant shift and $\nd''=1$.%
\footnote{In the case of preudo-NG modes, even the constant shift symmetry can 
be broken explicitly at some scale.}
This process of consecutive partial symmetry breaking opens up a hierarchy 
of energy scales 
\be
\mu\gg\mu'\gg\ldots \gg\mu'',
\ee
over which the propagator for $\phi$ exhibits a cascading behavior:  
First it appears logarithmic and the formation of a condensate seems 
precluded, and then it undergoes a series of crossovers to lower values of 
$z < \sdim$ until in the far infrared the condensate is no longer 
destabilized by infrared fluctuations.  The separation between two consecutive 
scales $\mu$ and $\mu'$ can be kept large, as a result of the symmetry that is 
given by a larger-degree polynomial at scale $\mu$ than at scale $\mu'$.  
All in all, whether or not the original continuous global internal symmetry (for which 
the field is the NG mode) is spontaneously broken is now a question 
about the competition of various scales in the system. 

\subsection{Polynomial Shift Symmetries as Exact Symmetries}

We have established a new infinite sequence of symmetries in scalar field 
theories, and have shown that they can protect the smallness of quantum 
mechanical 
corrections to their low-energy dispersion relations near the Gaussian fixed 
points.  The symmetries are exact at the infrared Gaussian fixed point, and 
turning on interactions typically breaks them explicitly -- as we have seen in 
the series of examples in \cite{msb}.  Yet, the polynomial shift symmetry at 
the Gaussian fixed point is useful for the interacting theory as well:  It 
controls the interaction terms, allowing them to be naturally small, 
parametrized by the amount $\varepsilon$ of the explicit polynomial symmetry 
breaking near the fixed point.  

Generally, this explicit breaking by interactions breaks the polynomial 
shift symmetries of NG modes all the way to the constant shift, which 
remains mandated by the original form of the Goldstone theorem (guaranteeing 
the existence of gapless modes).\footnote{Strictly speaking, moving away from the Gaussian fixed point by turning on the self-interactions generally yields additional corrections to the constant shift symmetries, if the underlying symmetry group of the interacting theory is non-Abelian.  Such non-Abelian corrections vanish at the Gaussian fixed point, and each NG component effectively becomes an Abelian field with its own constant shift symmetry.  In this paper, we will concentrate solely on the simplest Abelian case, with one Type A NG field $\phi$ and the symmetry group $U(1)$.}
However, one can now turn the argument 
around, and ask the following question:  Starting at a given Type A${}_\nd$ or 
B${}_{2 \nd}$ fixed point,  what are the lowest-dimension scalar composite 
operators that involve $\nv$ fields $\phi$ and respect the polynomial shift 
symmetry of degree $\pd$ exactly, up to a total derivative?  Such operators can 
be added to the action, and 
for $\nv =3,4,\ldots$ they represent self-interactions of the system, invariant 
under the polynomial shift of degree $\pd$.  More generally, one can attempt to 
classify all independent composite operators invariant under the polynomial 
shift symmetry of degree $\pd$, organized in the order of their increasing 
dimensions.   

These are the questions on which we focus in the rest of this paper.  
In order to provide some answers, we will first translate this classification 
problem into a more precise mathematical language, and then we will develop 
techniques -- largely based on abstract graph theory -- that lead us to 
systematic answers.  For some low values of the degree $\pd$ of the polynomial 
symmetry and of the number $\nv$ of fields involved, we can even find the most 
relevant invariants and prove their uniqueness.


\section{Galileon Invariants} \label{sec: Galileons} 

Consider a quantum field theory of a single scalar field $\phi ( t, \textbf{x} )$ in $\sdim$ spatial dimensions and one time dimension.  Consider the transformation of the field which is linear in spatial coordinates: $\delta \phi = a_i x^i + a_0$, where $a_i$ and $a_0$ are arbitrary real coefficients.  Other than the split between time and space and the exclusion of the time coordinate from the linear shift transformation, this is the same as the theory of the Galileon \cite{gal}. 

	The goal is to find Lagrangian terms which are invariant (up to a total derivative) under this linear shift transformation.  We will classify the Lagrangian terms by their numbers of fields $\nv$ and derivatives $2 \Delta$.  Imposing spatial rotation invariance requires that spatial derivatives be contracted in pairs by the flat metric $\delta_{i j}$.  Thus $\Delta$ counts the number of contracted pairs of derivatives.  It is easy to find Lagrangian terms which are exactly invariant (i.e., not just up to a total derivative): Let $\Delta \geq \nv$ and let at least two spatial derivatives act on every $\phi$.  For the linear shift case, all terms with at least twice as many derivatives as there are fields are equal to exact invariants, up to total derivatives (Theorem \ref{thm: classification of 1-invariants}).  However, it is possible for a term to have fewer derivatives than this and still be invariant up to a non-vanishing total derivative.  For fixed $\nv$, the terms with the lowest $\Delta$ are more relevant in the sense of the renormalization group.  Therefore, we will focus on invariant terms with the lowest number of derivatives, which we refer to as \textit{minimal invariants}.  

These minimal invariants have already been classified for the case of the linear shift.  There is a unique (up to total derivatives and an overall constant prefactor) $\nv$-point minimal invariant, which contains $2(\nv-1)$ derivatives (i.e., $\Delta = \nv -1$).  These are listed below up to $\nv =5$.
\begin{subequations} \label{eq: galileon_invariants}
\begin{align}
	L_{1\text{-pt}} &= \phi, \label{eq: galileon_invariant_1} \\
	L_{2\text{-pt}} &= \partial_i \phi \, \partial_i \phi, \label{eq: galileon_invariant_2} \\
	L_{3\text{-pt}} &= 3 \, \partial_i \phi \, \partial_j \phi \, \partial_i \partial_j \phi, \label{eq: galileon_invariant_3} \\
	L_{4\text{-pt}} &= 12 \, \partial_i \phi \, \partial_i \partial_j \phi \partial_j \partial_k \phi \, \partial_k \phi + 4 \, \partial_i \phi \, \partial_j \phi \, \partial_k \phi \, \partial_i \partial_j \partial_k \phi, \label{eq: galileon_invariant_4} \\
	L_{5\text{-pt}} &= 60 \, \partial_i \phi \, \partial_i \partial_j \phi \, \partial_j \partial_k \phi \, \partial_k \partial_{\ell} \phi \, \partial_{\ell} \phi + 60 \, \partial_i \phi \, \partial_i \partial_j \phi \, \partial_j \partial_k \partial_{\ell} \phi \, \partial_k \phi \, \partial_{\ell} \phi \notag \\
	&\quad + 5 \, \partial_i \phi \, \partial_j \phi \, \partial_k \phi \, \partial_{\ell} \phi \, \partial_i \partial_j \partial_j \partial_k \partial_{\ell} \phi. \label{eq: galileon_invariant_5}
\end{align}
\end{subequations}

\noindent These are not identical to the usual expressions (e.g., in \cite{gal}), but one can easily check that they are equivalent.  


\subsection{The Graphical Representation} \label{subsec: graph_rep}

We can represent the terms in \eqref{eq: galileon_invariants} as formal linear combinations of graphs.  In these graphs, $\phi$ is represented by a $\bullet$-vertex.  An edge joining two vertices represents a pair of contracted derivatives, one derivative acting on each of the $\phi$'s representing the endpoints of the edge.  The graphical representations of the above terms are given below:
\begin{subequations} \label{eq: galileon_graphs}
\begin{eqnarray}
	L_{1\text{-pt}} &=& \bullet, \label{eq: galileon_graph_1} \\[5pt]
	L_{2\text{-pt}} &=& \hspace{0.1cm}
	\begin{minipage}{1.3cm}
    	\begin{tikzpicture}
        		\draw [thick] (0,0) -- (\1,0);
        		\filldraw (0,0) circle [radius=0.08];
        		\filldraw (\1,0) circle [radius=0.08];
    	\end{tikzpicture}
    	\end{minipage}
    	\label{eq: galileon_graph_2}, \\[5pt]
	L_{3\text{-pt}} &=& 3 \hspace{0.1cm}
	\begin{minipage}{1.3cm}
    	\begin{tikzpicture}
        		\draw [thick] (0.4,0.69) -- (0,0) -- (\1,0);
        		\filldraw (0,0) circle [radius=0.08];
        		\filldraw (\1,0) circle [radius=0.08];
		\filldraw (0.4,0.69) circle [radius=0.08];
    	\end{tikzpicture}
    	\end{minipage}
    	\label{eq: galileon_graph_3}, \\[5pt]
	L_{4\text{-pt}} &=& 12 \hspace{0.1cm}
	\begin{minipage}{1.3cm}
    	\begin{tikzpicture}
        		\draw [thick] (0,\1) -- (0,0) -- (\1,0) -- (\1,\1);
        		\filldraw (0,0) circle [radius=0.08];
        		\filldraw (\1,0) circle [radius=0.08];
		\filldraw (0,\1) circle [radius=0.08];
		\filldraw (\1,\1) circle [radius=0.08];
    	\end{tikzpicture}
    	\end{minipage}
	+ 4 \hspace{0.1cm}
	\begin{minipage}{1.3cm}
    	\begin{tikzpicture}
        		\draw [thick] (0,\1) -- (0,0) -- (\1,0);
		\draw [thick] (0,0) -- (\1,\1);
        		\filldraw (0,0) circle [radius=0.08];
        		\filldraw (\1,0) circle [radius=0.08];
		\filldraw (0,\1) circle [radius=0.08];
		\filldraw (\1,\1) circle [radius=0.08];
    	\end{tikzpicture}
    	\end{minipage}
    	\label{eq: galileon_graph_4}, \\[5pt]
	L_{5\text{-pt}} &=& 60 \hspace{0.1cm}
	\begin{minipage}{1.3cm}
	\begin{tikzpicture}
	\filldraw (0,0) circle [radius=0.08];
	\filldraw (.6,0) circle [radius=0.08];
	\filldraw (.8,.6) circle [radius=0.08];
	\filldraw (.3,1) circle [radius=0.08];
	\filldraw (-.2,.6) circle [radius=0.08];
	\draw [thick]	(0,0) -- (-0.2,0.6) -- (0.3,1) -- (0.8,0.6) -- (0.6,0);
	\end{tikzpicture}
	\end{minipage}
	+ 60 \hspace{0.1cm}
	\begin{minipage}{1.3cm}
	\begin{tikzpicture}
	\filldraw (0,0) circle [radius=0.08];
	\filldraw (.6,0) circle [radius=0.08];
	\filldraw (.8,.6) circle [radius=0.08];
	\filldraw (.3,1) circle [radius=0.08];
	\filldraw (-.2,.6) circle [radius=0.08];
	\draw [thick]	(0,0) -- (.6,0)
				(0,0) -- (.8,.6)
				(-.2,.6) -- (0,0)
				(.3,1) -- (-.2,.6);
	\end{tikzpicture}
	\end{minipage}
	+ 5 \hspace{0.1cm}
	\begin{minipage}{1.3cm}
	\begin{tikzpicture}
	\filldraw (0,0) circle [radius=0.08];
	\filldraw (.6,0) circle [radius=0.08];
	\filldraw (.8,.6) circle [radius=0.08];
	\filldraw (.3,1) circle [radius=0.08];
	\filldraw (-.2,.6) circle [radius=0.08];
	\draw [thick]	(-0.2,0.6) -- (0.3,1) -- (0.8,0.6)
				(0,0) -- (0.3,1) -- (0.6,0);
	\end{tikzpicture}
	\end{minipage}
    	\label{eq: galileon_graph_5}.
\end{eqnarray}
\end{subequations}

\noindent The structure of the graph (i.e., the connectivity of the vertices) is what distinguishes graphs; the placement of the vertices is immaterial.  This reflects the fact that the order of the $\phi$'s in the algebraic expressions is immaterial and the only thing that matters is which contracted pairs of derivatives act on which pairs of $\phi$'s.  Therefore, for example, the graphs below all represent the same algebraic expression.
\begin{equation} \label{eq: graph_iso_3}
	\begin{minipage}{1.8cm}
    	\begin{tikzpicture}
        		\draw [thick] (0.4,0.69) -- (0,0) -- (\1,0);
        		\filldraw (0,0) circle [radius=0.08];
        		\filldraw (\1,0) circle [radius=0.08];
		\filldraw (0.4,0.69) circle [radius=0.08];
    	\end{tikzpicture}
    	\end{minipage}
	\qquad
	\begin{minipage}{1.8cm}
    	\begin{tikzpicture}
        		\draw [thick] (0,0) -- (0.4,0.69) -- (\1,0);
        		\filldraw (0,0) circle [radius=0.08];
        		\filldraw (\1,0) circle [radius=0.08];
		\filldraw (0.4,0.69) circle [radius=0.08];
    	\end{tikzpicture}
    	\end{minipage}
	\qquad
	\begin{minipage}{1.3cm}
    	\begin{tikzpicture}
        		\draw [thick] (0,0) -- (\1,0) -- (0.4,0.69);
        		\filldraw (0,0) circle [radius=0.08];
        		\filldraw (\1,0) circle [radius=0.08];
		\filldraw (0.4,0.69) circle [radius=0.08];
    	\end{tikzpicture}
    	\end{minipage}
\end{equation}

\noindent Similarly, the four graphs below represent the same algebraic expression.
\begin{equation} \label{eq: graph_iso_4_star}
	\hspace{0.5cm}
	\begin{minipage}{1.3cm}
    	\begin{tikzpicture}
        		\draw [thick] (0,\1) -- (0,0) -- (\1,0);
		\draw [thick] (0,0) -- (\1,\1);
        		\filldraw (0,0) circle [radius=0.08];
        		\filldraw (\1,0) circle [radius=0.08];
		\filldraw (0,\1) circle [radius=0.08];
		\filldraw (\1,\1) circle [radius=0.08];
    	\end{tikzpicture}
    	\end{minipage}
	\qquad
	\begin{minipage}{1.3cm}
    	\begin{tikzpicture}
        		\draw [thick] (0,0) -- (\1,0) -- (\1,\1);
		\draw [thick] (0,\1) -- (\1,0);
        		\filldraw (0,0) circle [radius=0.08];
        		\filldraw (\1,0) circle [radius=0.08];
		\filldraw (0,\1) circle [radius=0.08];
		\filldraw (\1,\1) circle [radius=0.08];
    	\end{tikzpicture}
    	\end{minipage}
	\qquad
	\begin{minipage}{1.3cm}
    	\begin{tikzpicture}
        		\draw [thick] (\1,0) -- (\1,\1) -- (0,\1);
		\draw [thick] (0,0) -- (\1,\1);
        		\filldraw (0,0) circle [radius=0.08];
        		\filldraw (\1,0) circle [radius=0.08];
		\filldraw (0,\1) circle [radius=0.08];
		\filldraw (\1,\1) circle [radius=0.08];
    	\end{tikzpicture}
    	\end{minipage}
	\qquad
	\begin{minipage}{1.3cm}
    	\begin{tikzpicture}
        		\draw [thick] (\1,\1) -- (0,\1) -- (0,0);
		\draw [thick] (0,\1) -- (\1,0);
        		\filldraw (0,0) circle [radius=0.08];
        		\filldraw (\1,0) circle [radius=0.08];
		\filldraw (0,\1) circle [radius=0.08];
		\filldraw (\1,\1) circle [radius=0.08];
    	\end{tikzpicture}
    	\end{minipage}
\end{equation}

\noindent A more nontrivial example is given by the following twelve graphs, which all represent the same algebraic expression.
\begin{align} \label{eq: graph_iso_4_line}
	&
	\begin{minipage}{1.3cm}
    	\begin{tikzpicture}
        		\draw [thick] (0,\1) -- (0,0) -- (\1,0) -- (\1,\1);
        		\filldraw (0,0) circle [radius=0.08];
        		\filldraw (\1,0) circle [radius=0.08];
		\filldraw (0,\1) circle [radius=0.08];
		\filldraw (\1,\1) circle [radius=0.08];
    	\end{tikzpicture}
    	\end{minipage}
	\qquad
	\begin{minipage}{1.3cm}
    	\begin{tikzpicture}
        		\draw [thick]  (0,0) -- (\1,0) -- (\1,\1) -- (0,\1);
        		\filldraw (0,0) circle [radius=0.08];
        		\filldraw (\1,0) circle [radius=0.08];
		\filldraw (0,\1) circle [radius=0.08];
		\filldraw (\1,\1) circle [radius=0.08];
    	\end{tikzpicture}
    	\end{minipage}
	\qquad
	\begin{minipage}{1.3cm}
    	\begin{tikzpicture}
        		\draw [thick] (\1,0) -- (\1,\1) -- (0,\1) -- (0,0);
        		\filldraw (0,0) circle [radius=0.08];
        		\filldraw (\1,0) circle [radius=0.08];
		\filldraw (0,\1) circle [radius=0.08];
		\filldraw (\1,\1) circle [radius=0.08];
    	\end{tikzpicture}
    	\end{minipage}
	\qquad
	\begin{minipage}{1.3cm}
    	\begin{tikzpicture}
        		\draw [thick] (\1,\1) -- (0,\1) -- (0,0) -- (\1,0);
        		\filldraw (0,0) circle [radius=0.08];
        		\filldraw (\1,0) circle [radius=0.08];
		\filldraw (0,\1) circle [radius=0.08];
		\filldraw (\1,\1) circle [radius=0.08];
    	\end{tikzpicture}
    	\end{minipage}
	\notag \\[10pt]
	&
	\begin{minipage}{1.3cm}
    	\begin{tikzpicture}
        		\draw [thick] (0,\1) -- (\1,0) -- (0,0) -- (\1,\1);
        		\filldraw (0,0) circle [radius=0.08];
        		\filldraw (\1,0) circle [radius=0.08];
		\filldraw (0,\1) circle [radius=0.08];
		\filldraw (\1,\1) circle [radius=0.08];
    	\end{tikzpicture}
    	\end{minipage}
	\qquad
	\begin{minipage}{1.3cm}
    	\begin{tikzpicture}
        		\draw [thick] (0,0) -- (\1,\1) -- (\1,0) -- (0,\1);
        		\filldraw (0,0) circle [radius=0.08];
        		\filldraw (\1,0) circle [radius=0.08];
		\filldraw (0,\1) circle [radius=0.08];
		\filldraw (\1,\1) circle [radius=0.08];
    	\end{tikzpicture}
    	\end{minipage}
	\qquad
	\begin{minipage}{1.3cm}
    	\begin{tikzpicture}
        		\draw [thick] (\1,0) -- (0,\1) -- (\1,\1) -- (0,0);
        		\filldraw (0,0) circle [radius=0.08];
        		\filldraw (\1,0) circle [radius=0.08];
		\filldraw (0,\1) circle [radius=0.08];
		\filldraw (\1,\1) circle [radius=0.08];
    	\end{tikzpicture}
    	\end{minipage}
	\qquad
	\begin{minipage}{1.3cm}
    	\begin{tikzpicture}
        		\draw [thick] (\1,\1) -- (0,0) -- (0,\1) -- (\1,0);
        		\filldraw (0,0) circle [radius=0.08];
        		\filldraw (\1,0) circle [radius=0.08];
		\filldraw (0,\1) circle [radius=0.08];
		\filldraw (\1,\1) circle [radius=0.08];
    	\end{tikzpicture}
    	\end{minipage} \\[10pt]
	&
	\begin{minipage}{1.3cm}
    	\begin{tikzpicture}
        		\draw [thick] (\1,0) -- (0,0) -- (\1,\1) -- (0,\1);
        		\filldraw (0,0) circle [radius=0.08];
        		\filldraw (\1,0) circle [radius=0.08];
		\filldraw (0,\1) circle [radius=0.08];
		\filldraw (\1,\1) circle [radius=0.08];
    	\end{tikzpicture}
    	\end{minipage}
	\qquad
	\begin{minipage}{1.3cm}
    	\begin{tikzpicture}
        		\draw [thick] (\1,\1) -- (\1,0) -- (0,\1) -- (0,0);
        		\filldraw (0,0) circle [radius=0.08];
        		\filldraw (\1,0) circle [radius=0.08];
		\filldraw (0,\1) circle [radius=0.08];
		\filldraw (\1,\1) circle [radius=0.08];
    	\end{tikzpicture}
    	\end{minipage}
	\qquad
	\begin{minipage}{1.3cm}
    	\begin{tikzpicture}
        		\draw [thick] (\1,\1) -- (0,\1) -- (\1,0) -- (0,0);
        		\filldraw (0,0) circle [radius=0.08];
        		\filldraw (\1,0) circle [radius=0.08];
		\filldraw (0,\1) circle [radius=0.08];
		\filldraw (\1,\1) circle [radius=0.08];
    	\end{tikzpicture}
    	\end{minipage}
	\qquad
	\begin{minipage}{1.3cm}
    	\begin{tikzpicture}
        		\draw [thick] (0,\1) -- (0,0) -- (\1,\1) -- (\1,0);
        		\filldraw (0,0) circle [radius=0.08];
        		\filldraw (\1,0) circle [radius=0.08];
		\filldraw (0,\1) circle [radius=0.08];
		\filldraw (\1,\1) circle [radius=0.08];
    	\end{tikzpicture}
    	\end{minipage} \notag
\end{align}

\noindent The graphs in the second line above appear to have intersecting edges.  However, since there is no $\bullet$-vertex at the would-be intersection, these edges do not actually intersect.


\subsection{Galileon Invariants as Equal-Weight Sums of Trees} \label{subsec: equal_weight}

There are three times as many graphs in $\eqref{eq: graph_iso_4_line}$ as there are in $\eqref{eq: graph_iso_4_star}$.  It so happens that the coefficient with which the first graph in $\eqref{eq: graph_iso_4_line}$ appears in $L_{4\text{-pt}}$ $\eqref{eq: galileon_graph_4}$ is also three times the coefficient with which the first graph in $\eqref{eq: graph_iso_4_star}$ appears in $L_{4\text{-pt}}$.  This suggests that the coefficient with which a graph appears in a minimal term is precisely the number of graphs with the exact same structure (i.e., isomorphic), just with various vertices and edges permuted.

One simple way to state this is to actually label the vertices in the graphs.  If the vertices were labeled, and thus distinguished from each other, then all of the graphs in each one of $\eqref{eq: graph_iso_3}$, $\eqref{eq: graph_iso_4_star}$ and $\eqref{eq: graph_iso_4_line}$ would actually be distinct graphs.  Of course, this means that the corresponding algebraic expressions have $\phi$'s similarly labeled, but this labeling is fiducial and may be removed afterwards.  Note the simplicity that this labeled convention introduces: $L_{4\text{-pt}}$ is the sum of all of the graphs in $\eqref{eq: graph_iso_4_star}$ and $\eqref{eq: graph_iso_4_line}$ with unit coefficients.

The graphs in $\eqref{eq: graph_iso_4_star}$ and $\eqref{eq: graph_iso_4_line}$ have an elegant and unified interpretation in graph theory.  These graphs are called \textit{trees}.  A tree is a graph which is connected (i.e., cannot be split into two or more separate graphs without cutting an edge), and contains no loops (edges joining a vertex to itself) or cycles (edges joining vertices in a closed cyclic manner).  One can check that there are exactly 16 trees with four vertices and they are given by $\eqref{eq: graph_iso_4_star}$ and $\eqref{eq: graph_iso_4_line}$.  Cayley's formula, a well-known result in graph theory, says that the number of trees with $\nv$ vertices is $\nv^{\nv -2}$.

For $\nv =3$, the $3^{3-2} = 3$ trees are in $\eqref{eq: graph_iso_3}$, and we indeed find that $L_{3\text{-pt}}$ is the sum of all three graphs with unit coefficients.  The same can be said for $L_{2\text{-pt}}$ and $L_{1\text{-pt}}$.  Therefore, the minimal terms for $\nv =1, 2, 3$ and $4$ are represented graphically as a sum of trees with unit coefficients (an equal-weight sum of trees).  If this were to hold for the $\nv =5$ case, it would strongly suggest that this may hold for all $\nv$.  

There are $5^3 = 125$ trees for $\nv =5$.  They can be divided into three sets such that the trees in each set are isomorphic to one of the three graphs appearing in $L_{5\text{-pt}}$ $\eqref{eq: galileon_graph_5}$.  There are 60 graphs which are isomorphic to the first graph appearing in $L_{5\text{-pt}}$; 12 of these are listed below and the rest are given by the five rotations acting on each of these 12 graphs:
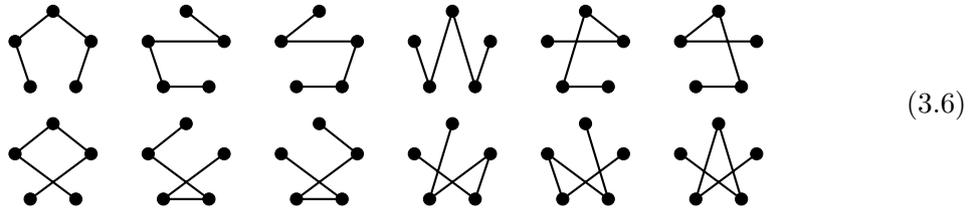
\begin{equation} \label{eq: graph_iso_5_line}
\begin{aligned}
	&
	\begin{minipage}{\mini cm}
	\begin{tikzpicture}
	\filldraw (0,0) circle [radius=0.08];
	\filldraw (.6,0) circle [radius=0.08];
	\filldraw (.8,.6) circle [radius=0.08];
	\filldraw (.3,1) circle [radius=0.08];
	\filldraw (-.2,.6) circle [radius=0.08];
	\draw [thick]	(0,0) -- (-0.2,0.6) -- (0.3,1) -- (0.8,0.6) -- (0.6,0);
	\end{tikzpicture}
	\end{minipage} &
	&
	\begin{minipage}{\mini cm}
	\begin{tikzpicture}
	\filldraw (0,0) circle [radius=0.08];
	\filldraw (.6,0) circle [radius=0.08];
	\filldraw (.8,.6) circle [radius=0.08];
	\filldraw (.3,1) circle [radius=0.08];
	\filldraw (-.2,.6) circle [radius=0.08];
	\draw [thick]	(0.6,0) -- (0,0) -- (-0.2,0.6) -- (0.8,0.6) -- (0.3,1);
	\end{tikzpicture}
	\end{minipage} &
	&
	\begin{minipage}{\mini cm}
	\begin{tikzpicture}
	\filldraw (0,0) circle [radius=0.08];
	\filldraw (.6,0) circle [radius=0.08];
	\filldraw (.8,.6) circle [radius=0.08];
	\filldraw (.3,1) circle [radius=0.08];
	\filldraw (-.2,.6) circle [radius=0.08];
	\draw [thick]	(0,0) -- (0.6,0) -- (0.8,0.6) -- (-0.2,0.6) -- (0.3,1);
	\end{tikzpicture}
	\end{minipage} &
	&
	\begin{minipage}{\mini cm}
	\begin{tikzpicture}
	\filldraw (0,0) circle [radius=0.08];
	\filldraw (.6,0) circle [radius=0.08];
	\filldraw (.8,.6) circle [radius=0.08];
	\filldraw (.3,1) circle [radius=0.08];
	\filldraw (-.2,.6) circle [radius=0.08];
	\draw [thick]	(-0.2,0.6) -- (0,0) -- (0.3,1) -- (0.6,0) -- (0.8,0.6);
	\end{tikzpicture}
	\end{minipage} &
	&
	\begin{minipage}{\mini cm}
	\begin{tikzpicture}
	\filldraw (0,0) circle [radius=0.08];
	\filldraw (.6,0) circle [radius=0.08];
	\filldraw (.8,.6) circle [radius=0.08];
	\filldraw (.3,1) circle [radius=0.08];
	\filldraw (-.2,.6) circle [radius=0.08];
	\draw [thick]	(0.6,0) -- (0,0) -- (0.3,1) -- (0.8,0.6) -- (-0.2,0.6);
	\end{tikzpicture}
	\end{minipage} &
	&
	\begin{minipage}{\mini cm}
	\begin{tikzpicture}
	\filldraw (0,0) circle [radius=0.08];
	\filldraw (.6,0) circle [radius=0.08];
	\filldraw (.8,.6) circle [radius=0.08];
	\filldraw (.3,1) circle [radius=0.08];
	\filldraw (-.2,.6) circle [radius=0.08];
	\draw [thick]	(0,0) -- (0.6,0) -- (0.3,1) -- (-0.2,0.6) -- (0.8,0.6);
	\end{tikzpicture}
	\end{minipage} \\[5pt]
	&
	\begin{minipage}{\mini cm}
	\begin{tikzpicture}
	\filldraw (0,0) circle [radius=0.08];
	\filldraw (.6,0) circle [radius=0.08];
	\filldraw (.8,.6) circle [radius=0.08];
	\filldraw (.3,1) circle [radius=0.08];
	\filldraw (-.2,.6) circle [radius=0.08];
	\draw [thick]	(0,0) -- (0.8,0.6) -- (0.3,1) -- (-0.2,0.6) -- (0.6,0);
	\end{tikzpicture}
	\end{minipage} &
	&
	\begin{minipage}{1.3cm}
	\begin{tikzpicture}
	\filldraw (0,0) circle [radius=0.08];
	\filldraw (.6,0) circle [radius=0.08];
	\filldraw (.8,.6) circle [radius=0.08];
	\filldraw (.3,1) circle [radius=0.08];
	\filldraw (-.2,.6) circle [radius=0.08];
	\draw [thick]	(0.3,1) -- (-0.2,0.6) -- (0.6,0) -- (0,0) -- (0.8,0.6);
	\end{tikzpicture}
	\end{minipage} &
	&
	\begin{minipage}{1.3cm}
	\begin{tikzpicture}
	\filldraw (0,0) circle [radius=0.08];
	\filldraw (.6,0) circle [radius=0.08];
	\filldraw (.8,.6) circle [radius=0.08];
	\filldraw (.3,1) circle [radius=0.08];
	\filldraw (-.2,.6) circle [radius=0.08];
	\draw [thick]	(0.3,1) -- (0.8,0.6) -- (0,0) -- (0.6,0) -- (-0.2,0.6);
	\end{tikzpicture}
	\end{minipage} &
	&
	\begin{minipage}{1.3cm}
	\begin{tikzpicture}
	\filldraw (0,0) circle [radius=0.08];
	\filldraw (.6,0) circle [radius=0.08];
	\filldraw (.8,.6) circle [radius=0.08];
	\filldraw (.3,1) circle [radius=0.08];
	\filldraw (-.2,.6) circle [radius=0.08];
	\draw [thick]	(-0.2,0.6) -- (0.6,0) -- (0.8,0.6) -- (0,0) -- (0.3,1);
	\end{tikzpicture}
	\end{minipage} &
	&
	\begin{minipage}{1.3cm}
	\begin{tikzpicture}
	\filldraw (0,0) circle [radius=0.08];
	\filldraw (.6,0) circle [radius=0.08];
	\filldraw (.8,.6) circle [radius=0.08];
	\filldraw (.3,1) circle [radius=0.08];
	\filldraw (-.2,.6) circle [radius=0.08];
	\draw [thick]	(0.8,0.6) -- (0,0) -- (-0.2,0.6) -- (0.6,0) -- (0.3,1);
	\end{tikzpicture}
	\end{minipage} &
	&
	\begin{minipage}{1.3cm}
	\begin{tikzpicture}
	\filldraw (0,0) circle [radius=0.08];
	\filldraw (.6,0) circle [radius=0.08];
	\filldraw (.8,.6) circle [radius=0.08];
	\filldraw (.3,1) circle [radius=0.08];
	\filldraw (-.2,.6) circle [radius=0.08];
	\draw [thick]	(-0.2,0.6) -- (0.6,0) -- (0.3,1) -- (0,0) -- (0.8,0.6);
	\end{tikzpicture}
	\end{minipage}
\end{aligned}
\end{equation}

\noindent There are 60 graphs which are isomorphic to the second graph appearing in $L_{5\text{-pt}}$; 12 of these are listed below and the rest are given by their rotations:
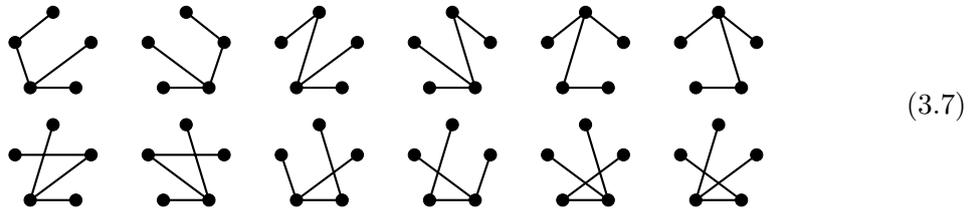
\begin{equation} \label{eq: graph_iso_5_scorpion}
\begin{aligned}
	&
	\begin{minipage}{\mini cm}
	\begin{tikzpicture}
	\filldraw (0,0) circle [radius=0.08];
	\filldraw (.6,0) circle [radius=0.08];
	\filldraw (.8,.6) circle [radius=0.08];
	\filldraw (.3,1) circle [radius=0.08];
	\filldraw (-.2,.6) circle [radius=0.08];
	\draw [thick]	(0.3,1) -- (-0.2,0.6) -- (0,0) -- (0.6,0);
	\draw [thick]	(0,0) -- (0.8,0.6);
	\end{tikzpicture}
	\end{minipage} &
	&
	\begin{minipage}{\mini cm}
	\begin{tikzpicture}
	\filldraw (0,0) circle [radius=0.08];
	\filldraw (.6,0) circle [radius=0.08];
	\filldraw (.8,.6) circle [radius=0.08];
	\filldraw (.3,1) circle [radius=0.08];
	\filldraw (-.2,.6) circle [radius=0.08];
	\draw [thick]	(0.3,1) -- (0.8,0.6) -- (0.6,0) -- (0,0);
	\draw [thick]	(0.6,0) -- (-0.2,0.6);
	\end{tikzpicture}
	\end{minipage} &
	&
	\begin{minipage}{\mini cm}
	\begin{tikzpicture}
	\filldraw (0,0) circle [radius=0.08];
	\filldraw (.6,0) circle [radius=0.08];
	\filldraw (.8,.6) circle [radius=0.08];
	\filldraw (.3,1) circle [radius=0.08];
	\filldraw (-.2,.6) circle [radius=0.08];
	\draw [thick]	(-0.2,0.6) -- (0.3,1) -- (0,0) -- (0.6,0);
	\draw [thick]	(0,0) -- (0.8,0.6);
	\end{tikzpicture}
	\end{minipage} &
	&
	\begin{minipage}{\mini cm}
	\begin{tikzpicture}
	\filldraw (0,0) circle [radius=0.08];
	\filldraw (.6,0) circle [radius=0.08];
	\filldraw (.8,.6) circle [radius=0.08];
	\filldraw (.3,1) circle [radius=0.08];
	\filldraw (-.2,.6) circle [radius=0.08];
	\draw [thick]	(0.8,0.6) -- (0.3,1) -- (0.6,0) -- (0,0);
	\draw [thick]	(0.6,0) -- (-0.2,0.6);
	\end{tikzpicture}
	\end{minipage} &
	&
	\begin{minipage}{\mini cm}
	\begin{tikzpicture}
	\filldraw (0,0) circle [radius=0.08];
	\filldraw (.6,0) circle [radius=0.08];
	\filldraw (.8,.6) circle [radius=0.08];
	\filldraw (.3,1) circle [radius=0.08];
	\filldraw (-.2,.6) circle [radius=0.08];
	\draw [thick]	(0.6,0) -- (0,0) -- (0.3,1) -- (-0.2,0.6);
	\draw [thick]	(0.3,1) -- (0.8,0.6);
	\end{tikzpicture}
	\end{minipage} &
	&
	\begin{minipage}{\mini cm}
	\begin{tikzpicture}
	\filldraw (0,0) circle [radius=0.08];
	\filldraw (.6,0) circle [radius=0.08];
	\filldraw (.8,.6) circle [radius=0.08];
	\filldraw (.3,1) circle [radius=0.08];
	\filldraw (-.2,.6) circle [radius=0.08];
	\draw [thick]	(0,0) -- (0.6,0) -- (0.3,1) -- (0.8,0.6);
	\draw [thick]	(0.3,1) -- (-0.2,0.6);
	\end{tikzpicture}
	\end{minipage} \\[5pt]
	&
	\begin{minipage}{1.3cm}
	\begin{tikzpicture}
	\filldraw (0,0) circle [radius=0.08];
	\filldraw (.6,0) circle [radius=0.08];
	\filldraw (.8,.6) circle [radius=0.08];
	\filldraw (.3,1) circle [radius=0.08];
	\filldraw (-.2,.6) circle [radius=0.08];
	\draw [thick]	(-0.2,0.6) -- (0.8,0.6) -- (0,0) -- (0.3,1);
	\draw [thick] 	(0,0) -- (0.6,0);
	\end{tikzpicture}
	\end{minipage} &
	&
	\begin{minipage}{1.3cm}
	\begin{tikzpicture}
	\filldraw (0,0) circle [radius=0.08];
	\filldraw (.6,0) circle [radius=0.08];
	\filldraw (.8,.6) circle [radius=0.08];
	\filldraw (.3,1) circle [radius=0.08];
	\filldraw (-.2,.6) circle [radius=0.08];
	\draw [thick]	(0.8,0.6) -- (-0.2,0.6) -- (0.6,0) -- (0.3,1);
	\draw [thick] 	(0,0) -- (0.6,0);
	\end{tikzpicture}
	\end{minipage} &
	&
	\begin{minipage}{1.3cm}
	\begin{tikzpicture}
	\filldraw (0,0) circle [radius=0.08];
	\filldraw (.6,0) circle [radius=0.08];
	\filldraw (.8,.6) circle [radius=0.08];
	\filldraw (.3,1) circle [radius=0.08];
	\filldraw (-.2,.6) circle [radius=0.08];
	\draw [thick]	(0.3,1) -- (0.6,0) -- (0,0) -- (0.8,0.6);
	\draw [thick] 	(0,0) -- (-0.2,0.6);
	\end{tikzpicture}
	\end{minipage} &
	&
	\begin{minipage}{1.3cm}
	\begin{tikzpicture}
	\filldraw (0,0) circle [radius=0.08];
	\filldraw (.6,0) circle [radius=0.08];
	\filldraw (.8,.6) circle [radius=0.08];
	\filldraw (.3,1) circle [radius=0.08];
	\filldraw (-.2,.6) circle [radius=0.08];
	\draw [thick]	(0.3,1) -- (0,0) -- (0.6,0) -- (0.8,0.6);
	\draw [thick] 	(0.6,0) -- (-0.2,0.6);
	\end{tikzpicture}
	\end{minipage} &
	&
	\begin{minipage}{1.3cm}
	\begin{tikzpicture}
	\filldraw (0,0) circle [radius=0.08];
	\filldraw (.6,0) circle [radius=0.08];
	\filldraw (.8,.6) circle [radius=0.08];
	\filldraw (.3,1) circle [radius=0.08];
	\filldraw (-.2,.6) circle [radius=0.08];
	\draw [thick]	(0.8,0.6) -- (0,0) -- (0.6,0) -- (0.3,1);
	\draw [thick]	(0.6,0) -- (-0.2,0.6);
	\end{tikzpicture}
	\end{minipage} &
	&
	\begin{minipage}{1.3cm}
	\begin{tikzpicture}
	\filldraw (0,0) circle [radius=0.08];
	\filldraw (.6,0) circle [radius=0.08];
	\filldraw (.8,.6) circle [radius=0.08];
	\filldraw (.3,1) circle [radius=0.08];
	\filldraw (-.2,.6) circle [radius=0.08];
	\draw [thick]	(-0.2,0.6) -- (0.6,0) -- (0,0) -- (0.3,1);
	\draw [thick]	(0,0) -- (0.8,0.6);
	\end{tikzpicture}
	\end{minipage}
\end{aligned}
\end{equation}

\noindent Finally, there are five graphs which are isomorphic to the third graph appearing in $L_{5\text{-pt}}$, which are simply the five rotations acting on that graph.  Therefore, $L_{5\text{-pt}}$ is indeed the sum with unit coefficients of all trees with five vertices! 

Thus, we arrive at the main result of this section (proven in Appendix \ref{sec: linear}): 

\begin{itemize}
\item[]
{\it The unique minimal $\nv$-point linear shift-invariant Lagrangian term is represented graphically as a sum with unit coefficients of all labeled trees with $\nv$ vertices.}
\end{itemize}


\section{Beyond the Galileons} \label{sec: beyond} 

Now, we extend the linear shift transformation to polynomials of higher degree.  We will need to develop the graphical approach further in order to tackle this problem and numerous technicalities will arise.  However, a rather elegant and beautiful description of these polynomial shift invariants will emerge.

Consider the problem of determining all possible terms in a Lagrangian that are invariant under the polynomial shift symmetry:
\begin{align}
	& \phi ( t , x^i ) \rightarrow \phi (t, x^i ) + \delta_\pd \phi, 
	& \delta_\pd \phi = a_{i_1 \cdots i_\pd} x^{i_1} \cdots x^{i_\pd} + \cdots + a_i x^i + a.  
\end{align}
\noindent The $a$'s are arbitrary real coefficients that parametrize the symmetry transformation, and are symmetric in any pair of indices.  $\pd = 0, 1, 2, \ldots$ corresponds to constant shift, linear shift, quadratic shift, and so on.  Obviously, if a term is invariant under a polynomial shift of order $\pd$, then it is also invariant under a polynomial shift of order $\pd'$ with $0 \leq \pd' \leq \pd$.

We will call a term with $\nv$ fields and $2\Delta$ derivatives an $(\nv, \Delta)$ term.  We are interested in interaction terms, for which $\nv \geq 3$.  As previously mentioned, terms with the lowest possible value of $\Delta$ are of greatest interest.  It is straightforward to write down invariant terms with $\Delta \geq \frac{1}{2} \nv (\pd +1)$ since, if each $\phi$ has more than $\pd$ derivatives acting on it, then the term is exactly invariant.  Are there any invariant terms with lower values of $\Delta$? If so, then these invariant terms will be more relevant than the exact invariants.

To be invariant, a term must transform into a total derivative under the polynomial shift symmetry.  In other words, for a specific $\pd$ and given $(\nv , \Delta)$, we are searching for terms $L$ such that
\begin{equation}\label{conseq}
	\delta_\pd L = \partial_i (L_i).
\end{equation}
\noindent Here $L$ is a linear combination of terms with $\nv$ $\phi$'s and $2\Delta$ $\partial$'s, and $L_i$ is a linear combination of terms with $\nv -1$ $\phi$'s.  Such $L$'s are called \emph{\pd-invariants}.

How might we determine such invariant terms in general? For a given $(\nv ,\Delta)$, the most brute-force method for determining invariant terms can be described as follows.  First, write down all possible terms in the Lagrangian with a given $(\nv ,\Delta)$ and ensure that they are independent up to integration by parts.  Next, take the variation of all these terms under the polynomial shift.  There may exist linear combinations of these variations which are equal to a total derivative, which we call \emph{total derivative relations}.  If we use these total derivative relations to maximally reduce the number of variation terms, then the required $\pd$-invariants form the kernel of the map from the independent Lagrangian terms to the independent variation terms (Corollary \ref{Cor: kernel}).  Let us consider some examples of this brute-force procedure in action.

\subsection{Brute-force Examples} \label{sec: brute force}

\subsubsection{\texorpdfstring{$(\pd , \nv , \Delta ) = (1,3,2)$}{P=1, n=3, Delta=2}} \label{P=1, n=3, m=2}

In this case, a general Lagrangian is made up of two independent terms, after integrating by parts, given by
\begin{align*}
	& L_1 = \partial_i\phi \, \partial_j\phi \, \partial_i\partial_j\phi, 
	& L_2 = \phi \, \partial_i\partial_j\phi \, \partial_i\partial_j\phi.
\end{align*}

\noindent The variation under the linear shift symmetry (for $\pd =1$) of these terms is given by
\begin{align*}
	& \delta_1(L_1) = 2L_a^{ \times }, 
	& \delta_1(L_2) = L_b^{ \times },
\end{align*}

\noindent where $L_a^{ \times } =a_i \, \partial_j\phi \, \partial_i\partial_j\phi$ and $L_b^{ \times } = ( a_k x^k + a ) \, \partial_i \partial_j \phi \, \partial_i \partial_j \phi$.  There is only one total derivative that can be formed from these terms, namely
\begin{equation*}
 \partial_i(a_i\partial_j\phi \, \partial_j\phi) = 2L_a^{ \times }.
\end{equation*}

\noindent Therefore, there is a single invariant term for $(\pd , \nv , \Delta ) = (1,3,2)$, given by 
\begin{equation*}
	L_1=\partial_i\phi \, \partial_j\phi \, \partial_i\partial_j\phi.
\end{equation*}

\subsubsection{\texorpdfstring{$(\pd , \nv , \Delta ) = (3,3,4)$}{P=3, n=3, Delta=4}}\label{P=3, n=3, m=4}

In this case, a general Lagrangian is made up of four independent terms, after integrating by parts, given by
\begin{align*}
	L_1 &= \partial_i \partial_j \phi \, \partial_k \partial_l \phi \, \partial_i \partial_j \partial_k \partial_l \phi, &%
	L_2 &= \partial_i \partial_j \phi \, \partial_i \partial_k \partial_l\phi \, \partial_j \partial_k \partial_l \phi, \\
	L_3 &= \partial_i \phi \, \partial_j \partial_k \partial_l \phi \, \partial_i \partial_j \partial_k \partial_l \phi, &%
	L_4 &= \phi \, \partial_i \partial_j \partial_k \partial_l \phi \, \partial_i \partial_j \partial_k \partial_l \phi.
\end{align*}
The variation under the cubic shift symmetry (for $\pd =3$) of these terms is given by
\begin{align*}
	\delta_3(L_1) &= 2L_a^{ \times }, &%
	\delta_3(L_2) &= L_b^{ \times }+2L_c^{ \times }, \\
	\delta_3(L_3) &= L_d^{ \times }+L_e^{ \times }, &%
	\delta_3(L_4) &= L_f^{ \times }. 
\end{align*}
where
\begin{eqnarray*}
L_a^{ \times }&=&(6a_{ijm}x^m+2a_{ij})\partial_k\partial_l\phi \, \partial_i\partial_j\partial_k\partial_l\phi\cr
L_b^{ \times }&=&(6a_{ijm}x^m+2a_{ij})\partial_i\partial_k\partial_l\phi \, \partial_j\partial_k\partial_l\phi\cr
L_c^{ \times }&=&6a_{ikl}\partial_i\partial_j\phi \, \partial_j\partial_k\partial_l\phi\cr
L_d^{ \times }&=&(3a_{imn}x^mx^n+2a_{im}x^m+a_i)\partial_i\phi \, \partial_j\partial_k\partial_l\phi \, \partial_i\partial_j\partial_k\partial_l\phi\cr
L_e^{ \times }&=&6a_{jkl}\partial_i\phi \, \partial_i\partial_j\partial_k\partial_l\phi\cr
L_f^{ \times }&=&(a_{mnp}x^mx^nx^p+a_{mn}x^mx^n+a_mx^m+a)\partial_i\partial_j\partial_k\partial_l\phi \, \partial_i\partial_j\partial_k\partial_l\phi.
\end{eqnarray*}
There are three independent total derivatives that can be formed out of these: 
\begin{align*}
	& \partial_i[2(6a_{ijm}x^m+2a_{ij})\partial_k\partial_l\phi \, \partial_j\partial_k\partial_l\phi-6a_{ijj}\partial_k\partial_l\phi \, \partial_k\partial_l\phi] = 2 ( L_a^{ \times }+ L_b^{ \times } ), \\
	& \partial_i[6a_{ijk}\partial_j\partial_l\phi \, \partial_k\partial_l\phi] = 2L_c^{ \times }, \\
	& \partial_i [ 6a_{ijk} \partial_j \partial_k \partial_{\ell} \phi \, \partial_{\ell} \phi ] = L^{ \times } _{c}+ L^{ \times }_{e}.
 \end{align*}
 
\noindent It is a non-trivial exercise to find and verify this, and a more systematic way of finding the total derivative relations will be introduced later.
 
Applying these relations, one finds a single invariant for $(\pd, \nv , \Delta ) = (3,3,4)$: 
\begin{equation*}
	L_1+2L_2=\partial_i\partial_j\phi \, \partial_k\partial_l\phi \, \partial_i\partial_j\partial_k\partial_l\phi
+2\partial_i\partial_j\phi \, \partial_i\partial_k\partial_l\phi \, \partial_j\partial_k\partial_l\phi. 
\end{equation*}

\noindent Note that $\delta_3(L_1+2L_2)=2(L_a^{ \times } +L_b^{ \times })+2(2L_c^{ \times })$, which is a total derivative.

\subsection{Introduction to the Graphical Representation} \label{sec: intro graphs}

	It is clear that even for these simple examples, the calculations quickly become unwieldy, and it becomes increasingly difficult to classify all of the total derivative relations.  At this point we will rewrite these results in a graphical notation which will make it easier to keep track of the contractions of indices in the partial derivatives.   Full details about this graphical approach can be found in Appendix \ref{AppB}, but we will summarize them here.  In addition to the $\bullet$-vertex and edges we introduced in \S\ref{sec: Galileons}, we represent $\delta_\pd \phi$ by a $\otimes$ (a $\times$-vertex).  Note that there are at most $\pd$ edges incident to a $\times$-vertex since $\pd +1$ derivatives acting on $\delta_\pd \phi$ yields zero, whereas an arbitrary number of edges can be incident to a $\bullet$-vertex.  Moreover, we introduce another vertex, called a {$\star$}-vertex, which will be used to represent terms that are total derivatives.  We require that a {$\star$}-vertex always be incident to exactly one edge, and that this edge be incident to a $\bullet$-vertex or $\times$-vertex.  This edge represents a derivative acting on the entire term as a whole, and the index of that derivative is contracted with the index of another derivative acting on the $\phi$ or $\delta_\pd \phi$ of the $\bullet$- or $\times$-vertex, respectively, to which the $\star$-vertex is adjacent.  Therefore, directly from the definition, any graph with a {$\star$}-vertex represents a total derivative term.  The expansion of this derivative using the Leibniz rule is graphically represented by the summation of the graphs formed by removing the {$\star$}-vertex and attaching the edge that was incident to the {$\star$}-vertex to each remaining vertex.  This operation is denoted by the \emph{derivative map} $\rho$.  The symbols $\nbv$, $\ntv$ and $\nsv$ represent the numbers of each type of vertex.  Note that $\nv = \nbv + \ntv$ does not include $\nsv$ since $\star$-vertices represent neither $\phi$ nor $\delta_\pd \phi$.

	We define three special types of graphs: A \emph{plain-graph} is a graph in which all vertices are $\bullet$-vertices.  A \emph{$\times$-graph} is a plain-graph with one $\bullet$-vertex replaced with a $\times$-vertex.  A \emph{{$\star$}-graph} is a graph with one $\times$-vertex and at least one {$\star$}-vertex. 

	Note that the variation $\delta_\pd$ of a plain-graph under the polynomial shift symmetry is given by summing over all graphs that have exactly one $\bullet$-vertex in the original graph replaced with a $\times$-vertex.  To illustrate the graphical approach, we rewrite the examples from sections \ref{P=1, n=3, m=2} and \ref{P=3, n=3, m=4} using this new graphical notation.  Since the algebraic expressions have unlabeled $\phi$'s, the graphs in this section will be unlabeled.

\subsubsection{\texorpdfstring{$(\pd , \nv , \Delta) = (1,3,2)$}{P=1, n=3, Delta=2}}

	The two independent terms are written in the graphical notation as
\begin{align*}
L_1 &= \hspace{0.1cm}
\begin{minipage}{\mini cm}
\begin{tikzpicture}
\draw [thick] (0,0) -- (0,\1);
\filldraw (0,0) circle [radius=0.08];
\filldraw (\1,0) circle [radius=0.08];
\filldraw (0,\1) circle [radius=0.08];
\draw [thick] (0,0) -- (\1,0);
\end{tikzpicture}
\end{minipage}
& 
L_2 &= \hspace{0.2cm}
\begin{minipage}{\mini cm}
\begin{tikzpicture}
\filldraw (0,0) circle [radius=0.08];
\filldraw (\1,0) circle [radius=0.08];
\filldraw (0,\1) circle [radius=0.08];
\draw [thick] (0,0) to [out=20,in=160] (\1,0);
\draw [thick] (0,0) to [out=-20,in=-160] (\1,0);
\end{tikzpicture}
\end{minipage}
\end{align*}

\noindent The variation under the linear shift symmetry (for $\pd =1$) is given by
\begin{align*}
\delta_1 \left( \hspace{-0.1cm} 
\begin{minipage}{\mini cm}
\begin{tikzpicture}
\filldraw [white] (0,\1) circle [radius=0.115];
\draw [thick,white] (0,\1) circle [radius=0.115];
\node [white] at (0,\1) {$\times$};
\draw [thick] (0,0) --  (0,\1);
\filldraw (0,0) circle [radius=0.08];
\filldraw (\1,0) circle [radius=0.08];
\filldraw (0,\1) circle [radius=0.08];
\draw [thick] (0,0) -- (\1,0);
\end{tikzpicture}
\end{minipage} 
\hspace{-0.1cm}
\right) &=2
\begin{minipage}{\mini cm}
\begin{tikzpicture}
\draw [thick] (0,0) -- (0,\1);
\filldraw (0,0) circle [radius=0.08];
\filldraw (\1,0) circle [radius=0.08];
\filldraw [white] (0,\1) circle [radius=0.115];
\draw [thick] (0,\1) circle [radius=0.115];
\node at (0,\1) {$\times$};
\draw [thick] (0,0) -- (\1,0);
\end{tikzpicture}
\end{minipage} 
&   
\delta_1 \left(  \hspace{-0.1cm}
\begin{minipage}{\mini cm}
\begin{tikzpicture}
\filldraw [white] (0,\1) circle [radius=0.115];
\draw [thick,white] (0,\1) circle [radius=0.115];
\node [white] at (0,\1) {$\times$};
\filldraw (0,0) circle [radius=0.08];
\filldraw (\1,0) circle [radius=0.08];
\filldraw (0,\1) circle [radius=0.08];
\draw [thick] (0,0) to [out=20,in=160] (\1,0);
\draw [thick] (0,0) to [out=-20,in=-160] (\1,0);
\end{tikzpicture}
\end{minipage}
\hspace{-0.1cm}
\right) &=
\begin{minipage}{\mini cm}
\begin{tikzpicture}
\filldraw (0,0) circle [radius=0.08];
\filldraw (\1,0) circle [radius=0.08];
\filldraw [white] (0,\1) circle [radius=0.115];
\draw [thick] (0,\1) circle [radius=0.115];
\node at (0,\1) {$\times$};
\draw [thick] (0,0) to [out=20,in=160] (\1,0);
\draw [thick] (0,0) to [out=-20,in=-160] (\1,0);
\end{tikzpicture}
\end{minipage}
\end{align*}

\noindent The only independent total derivative that can be formed out of these terms is
\begin{align*}
\rho \left(
\begin{minipage}{\mini cm}
\begin{tikzpicture}
\draw [thick] (0,\1) -- (\1,\1);
\node at (\1,\1) {\scalebox{0.8}{$\bigstar$}};
\filldraw (0,0) circle [radius=0.08];
\filldraw (\1,0) circle [radius=0.08];
\filldraw [white] (0,\1) circle [radius=0.115];
\draw [thick] (0,\1) circle [radius=0.115];
\node at (0,\1) {$\times$};
\draw [thick] (0,0) -- (\1,0);
\end{tikzpicture}
\end{minipage}
\right)
=2 
\begin{minipage}{\mini cm}
\begin{tikzpicture}
\filldraw (\1,0) circle [radius=0.08];
\filldraw (0,0) circle [radius=0.08];
\draw [thick] (0,\1) -- (0,0)  -- (\1,0);
\filldraw [white] (0,\1) circle [radius=0.115];
\draw [thick] (0,\1) circle [radius=0.115];
\node at (0,\1) {$\times$};
\end{tikzpicture}
\end{minipage}
\end{align*}

\noindent As before, there is a single invariant for $(\pd , \nv , \Delta ) = (1,3,2)$ given by $L_1$.  In this case, the graphical version of (\ref{conseq}) is given by
\begin{align*}
\delta_1 \left(  \hspace{-0.1cm}
\begin{minipage}{\mini cm}
\begin{tikzpicture}
\filldraw [white] (0,\1) circle [radius=0.115];
\draw [thick,white] (0,\1) circle [radius=0.115];
\node [white] at (0,\1) {$\times$};
\draw [thick] (0,0) --  (0,\1);
\filldraw (0,0) circle [radius=0.08];
\filldraw (\1,0) circle [radius=0.08];
\filldraw (0,\1) circle [radius=0.08];
\draw [thick] (0,0) -- (\1,0);
\end{tikzpicture}
\end{minipage} 
\hspace{-0.1cm}
\right)=2
\begin{minipage}{\mini cm}
\begin{tikzpicture}
\draw [thick] (0,0) -- (0,\1);
\filldraw (0,0) circle [radius=0.08];
\filldraw (\1,0) circle [radius=0.08];
\filldraw [white] (0,\1) circle [radius=0.115];
\draw [thick] (0,\1) circle [radius=0.115];
\node at (0,\1) {$\times$};
\draw [thick] (0,0) -- (\1,0);
\end{tikzpicture}
\end{minipage}
=\rho \left(
\begin{minipage}{\mini cm}
\begin{tikzpicture}
\draw [thick] (0,\1) -- (\1,\1);
\node at (\1,\1) {\scalebox{0.8}{$\bigstar$}};
\filldraw (0,0) circle [radius=0.08];
\filldraw (\1,0) circle [radius=0.08];
\filldraw [white] (0,\1) circle [radius=0.115];
\draw [thick] (0,\1) circle [radius=0.115];
\node at (0,\1) {$\times$};
\draw [thick] (0,0) -- (\1,0);
\end{tikzpicture}
\end{minipage}
\right) 
\end{align*}

\subsubsection{\texorpdfstring{$(\pd , \nv , \Delta ) = (3,3,4)$}{P=3, n=3, Delta=4}} \label{sec: 334} 

The four independent terms are written in the graphical notation as
\begin{align*}
& L_1= 
\begin{minipage}{\mini cm}
\begin{tikzpicture}
\draw [thick] (0,0) to [out=70,in=-70] (0,\1);
\draw [thick] (0,0) to [out=110,in=-110] (0,\1);
\filldraw (0,0) circle [radius=0.08];
\filldraw (\1,0) circle [radius=0.08];
\filldraw (0,\1) circle [radius=0.08];
\draw [thick] (0,0) to [out=20,in=160] (\1,0);
\draw [thick] (0,0) to [out=-20,in=-160] (\1,0);
\end{tikzpicture}
\end{minipage}
&
& L_2=
\begin{minipage}{\mini cm}
\begin{tikzpicture}
\filldraw (0,0) circle [radius=0.08];
\filldraw (\1,0) circle [radius=0.08];
\filldraw (0,\1) circle [radius=0.08];
\draw [thick] (0,0) -- (0,\1) -- (\1,0);
\draw [thick] (0,0) to [out=20,in=160] (\1,0);
\draw [thick] (0,0) to [out=-20,in=-160] (\1,0);
\end{tikzpicture}
\end{minipage}
&
& L_3= 
\begin{minipage}{\mini cm}
\begin{tikzpicture}
\draw [thick] (\1,0) -- (0,0) -- (0,\1);
\filldraw (0,0) circle [radius=0.08];
\filldraw (\1,0) circle [radius=0.08];
\filldraw (0,\1) circle [radius=0.08];
\draw [thick] (0,0) to [out=20,in=160] (\1,0);
\draw [thick] (0,0) to [out=-20,in=-160] (\1,0);
\end{tikzpicture}
\end{minipage}
&
& L_4=
\begin{minipage}{\mini cm}
\begin{tikzpicture}
\filldraw (0,0) circle [radius=0.08];
\filldraw (\1,0) circle [radius=0.08];
\filldraw (0,\1) circle [radius=0.08];
\draw [thick] (0,0) to [out=10,in=170] (\1,0);
\draw [thick] (0,0) to [out=-10,in=-170] (\1,0);
\draw [thick] (0,0) to [out=30,in=150] (\1,0);
\draw [thick] (0,0) to [out=-30,in=-150] (\1,0);
\end{tikzpicture}
\end{minipage}
\end{align*}

\noindent The variation under the cubic shift symmetry (for $\pd =3$) is given by
\begin{align*}
\delta_3 \Bigg(  \hspace{-0.1cm}
\begin{minipage}{\mini cm}
\begin{tikzpicture}
\filldraw [white] (0,\1) circle [radius=0.115];
\draw [thick,white] (0,\1) circle [radius=0.115];
\node [white] at (0,\1) {$\times$};
\draw [thick] (0,0) to [out=70,in=-70] (0,\1);
\draw [thick] (0,0) to [out=110,in=-110] (0,\1);
\filldraw (0,0) circle [radius=0.08];
\filldraw (\1,0) circle [radius=0.08];
\filldraw (0,\1) circle [radius=0.08];
\draw [thick] (0,0) to [out=20,in=160] (\1,0);
\draw [thick] (0,0) to [out=-20,in=-160] (\1,0);
\end{tikzpicture}
\end{minipage} 
\hspace{1mm}
\hspace{-0.2cm} \Bigg)
&=
2\begin{minipage}{\mini cm}
\begin{tikzpicture}
\draw [thick] (0,0) to [out=70,in=-70] (0,\1);
\draw [thick] (0,0) to [out=110,in=-110] (0,\1);
\filldraw (0,0) circle [radius=0.08];
\filldraw (\1,0) circle [radius=0.08];
\filldraw [white] (0,\1) circle [radius=0.115];
\draw [thick] (0,\1) circle [radius=0.115];
\node at (0,\1) {$\times$};
\draw [thick] (0,0) to [out=20,in=160] (\1,0);
\draw [thick] (0,0) to [out=-20,in=-160] (\1,0);
\end{tikzpicture}
\end{minipage}
&
\delta_3 \Bigg(  \hspace{-0.1cm}
\begin{minipage}{\mini cm}
\begin{tikzpicture}
\filldraw [white] (0,\1) circle [radius=0.115];
\draw [thick,white] (0,\1) circle [radius=0.115];
\node [white] at (0,\1) {$\times$};
\filldraw (0,0) circle [radius=0.08];
\filldraw (\1,0) circle [radius=0.08];
\filldraw (0,\1) circle [radius=0.08];
\draw [thick] (0,0) -- (0,\1) -- (\1,0);
\draw [thick] (0,0) to [out=20,in=160] (\1,0);
\draw [thick] (0,0) to [out=-20,in=-160] (\1,0);
\end{tikzpicture}
\end{minipage}
\hspace{-0.1cm} \Bigg)
&=
\begin{minipage}{\mini cm}
\begin{tikzpicture}
\filldraw (0,0) circle [radius=0.08];
\filldraw (\1,0) circle [radius=0.08];
\draw [thick] (0,0) -- (0,\1) -- (\1,0);
\draw [thick] (0,0) to [out=20,in=160] (\1,0);
\draw [thick] (0,0) to [out=-20,in=-160] (\1,0);
\filldraw [white] (0,\1) circle [radius=0.115];
\draw [thick] (0,\1) circle [radius=0.115];
\node at (0,\1) {$\times$};
\end{tikzpicture}
\end{minipage}
\hspace{-0.2cm}
+2 
\hspace{-0.1cm}
\begin{minipage}{\mini cm}
\begin{tikzpicture}
\filldraw (\1,0) circle [radius=0.08];
\filldraw (0,0) circle [radius=0.08];
\draw [thick] (0,\1) -- (0,0)  -- (\1,0);
\draw [thick] (0,\1) to [out=-25,in=115] (\1,0);
\draw [thick] (0,\1) to [out=-65,in=155] (\1,0);
\filldraw [white] (0,\1) circle [radius=0.115];
\draw [thick] (0,\1) circle [radius=0.115];
\node at (0,\1) {$\times$};
\end{tikzpicture}
\end{minipage}
\\[5pt]
\delta_3 \Bigg(  \hspace{-0.1cm}
\begin{minipage}{\mini cm}
\begin{tikzpicture}
\filldraw [white] (0,\1) circle [radius=0.115];
\draw [thick,white] (0,\1) circle [radius=0.115];
\node [white] at (0,\1) {$\times$};
\draw [thick, white] (0,0) to [out=110,in=-110] (0,\1);
\draw [thick] (\1,0) -- (0,0) -- (0,\1);
\filldraw (0,0) circle [radius=0.08];
\filldraw (\1,0) circle [radius=0.08];
\filldraw (0,\1) circle [radius=0.08];
\draw [thick] (0,0) to [out=20,in=160] (\1,0);
\draw [thick] (0,0) to [out=-20,in=-160] (\1,0);
\end{tikzpicture}
\end{minipage}
\hspace{-0.1cm} \Bigg)
&=
\phantom{2}
\begin{minipage}{\mini cm}
\begin{tikzpicture}
\draw [thick] (\1,0) -- (0,0) -- (0,\1);
\filldraw (0,0) circle [radius=0.08];
\filldraw (\1,0) circle [radius=0.08];
\draw [thick] (0,0) to [out=20,in=160] (\1,0);
\draw [thick] (0,0) to [out=70,in=-70] (0,\1);
\draw [thick] (0,0) to [out=110,in=-110] (0,\1);
\draw [thick] (0,0) to [out=-20,in=-160] (\1,0);
\filldraw [white] (0,\1) circle [radius=0.115];
\draw [thick] (0,\1) circle [radius=0.115];
\node at (0,\1) {$\times$};
\end{tikzpicture}
\end{minipage}
\hspace{-0.2cm} + \hspace{-0.1cm}
\begin{minipage}{\mini cm}
\begin{tikzpicture}
\draw [thick] (\1,0) -- (0,0) -- (0,\1);
\filldraw (0,0) circle [radius=0.08];
\filldraw (\1,0) circle [radius=0.08];
\draw [thick] (0,0) to [out=70,in=-70] (0,\1);
\draw [thick] (0,0) to [out=110,in=-110] (0,\1);
\filldraw [white] (0,\1) circle [radius=0.115];
\draw [thick] (0,\1) circle [radius=0.115];
\node at (0,\1) {$\times$};
\end{tikzpicture}
\end{minipage}
&
\delta_3 \Bigg(  \hspace{-0.1cm}
\begin{minipage}{\mini cm}
\begin{tikzpicture}
\filldraw [white] (0,\1) circle [radius=0.115];
\draw [thick,white] (0,\1) circle [radius=0.115];
\node [white] at (0,\1) {$\times$};
\filldraw (0,0) circle [radius=0.08];
\filldraw (\1,0) circle [radius=0.08];
\filldraw (0,\1) circle [radius=0.08];
\draw [thick] (0,0) to [out=10,in=170] (\1,0);
\draw [thick] (0,0) to [out=-10,in=-170] (\1,0);
\draw [thick] (0,0) to [out=30,in=150] (\1,0);
\draw [thick] (0,0) to [out=-30,in=-150] (\1,0);
\end{tikzpicture}
\end{minipage}
\hspace{-0.1cm} \Bigg)
&=
\begin{minipage}{\mini cm}
\begin{tikzpicture}
\filldraw (0,0) circle [radius=0.08];
\filldraw (\1,0) circle [radius=0.08];
\filldraw [white] (0,\1) circle [radius=0.115];
\draw [thick] (0,\1) circle [radius=0.115];
\node at (0,\1) {$\times$};
\draw [thick] (0,0) to [out=10,in=170] (\1,0);
\draw [thick] (0,0) to [out=-10,in=-170] (\1,0);
\draw [thick] (0,0) to [out=30,in=150] (\1,0);
\draw [thick] (0,0) to [out=-30,in=-150] (\1,0);
\end{tikzpicture}
\end{minipage}
\end{align*}

\noindent The independent total derivatives that can be formed out of these terms are
\begin{align*}
& \rho \Bigg(
2
\begin{minipage}{\mini cm}
\begin{tikzpicture}
\draw [thick] (0,\1) -- (\1,\1);
\draw [thick] (0,\1) -- (0,0);
\node at (\1,\1) {\scalebox{0.8}{$\bigstar$}};
\filldraw (0,0) circle [radius=0.08];
\filldraw (\1,0) circle [radius=0.08];
\filldraw [white] (0,\1) circle [radius=0.115];
\draw [thick] (0,\1) circle [radius=0.115];
\node at (0,\1) {$\times$};
\draw [thick] (0,0) to [out=20,in=160] (\1,0);
\draw [thick] (0,0) to [out=-20,in=-160] (\1,0);
\end{tikzpicture}
\end{minipage}
\hspace{-0.1cm} -
\begin{minipage}{\mini cm}
\begin{tikzpicture}
\draw [thick] (0,\1) -- (\1,\1);
\draw [thick] (0,\1) .. controls (0,0) and (\1,\1) .. (0,\1);
\node at (\1,\1) {\scalebox{0.8}{$\bigstar$}};
\filldraw (0,0) circle [radius=0.08];
\filldraw (\1,0) circle [radius=0.08];
\filldraw [white] (0,\1) circle [radius=0.115];
\draw [thick] (0,\1) circle [radius=0.115];
\node at (0,\1) {$\times$};
\draw [thick] (0,0) to [out=20,in=160] (\1,0);
\draw [thick] (0,0) to [out=-20,in=-160] (\1,0);
\end{tikzpicture}
\end{minipage}
\Bigg)
=2 
\begin{minipage}{\mini cm}
\begin{tikzpicture}
\draw [thick] (0,0) to [out=70,in=-70] (0,\1);
\draw [thick] (0,0) to [out=110,in=-110] (0,\1);
\filldraw (0,0) circle [radius=0.08];
\filldraw (\1,0) circle [radius=0.08];
\filldraw [white] (0,\1) circle [radius=0.115];
\draw [thick] (0,\1) circle [radius=0.115];
\node at (0,\1) {$\times$};
\draw [thick] (0,0) to [out=20,in=160] (\1,0);
\draw [thick] (0,0) to [out=-20,in=-160] (\1,0);
\end{tikzpicture}
\end{minipage}
\hspace{-0.2cm} + 2
\begin{minipage}{\mini cm}
\begin{tikzpicture}
\filldraw (0,0) circle [radius=0.08];
\filldraw (\1,0) circle [radius=0.08];
\draw [thick] (0,0) -- (0,\1) -- (\1,0);
\draw [thick] (0,0) to [out=20,in=160] (\1,0);
\draw [thick] (0,0) to [out=-20,in=-160] (\1,0);
\filldraw [white] (0,\1) circle [radius=0.115];
\draw [thick] (0,\1) circle [radius=0.115];
\node at (0,\1) {$\times$};
\end{tikzpicture}
\end{minipage} \\[5pt]
& \rho \Bigg(
\begin{minipage}{\mini cm}
\begin{tikzpicture}
\draw [thick] (0,\1) -- (\1,\1);
\node at (\1,\1) {\scalebox{0.8}{$\bigstar$}};
\filldraw (0,0) circle [radius=0.08];
\filldraw (\1,0) circle [radius=0.08];
\draw [thick] (0,0) -- (0,\1) -- (\1,0) -- (0,0);
\filldraw [white] (0,\1) circle [radius=0.115];
\draw [thick] (0,\1) circle [radius=0.115];
\node at (0,\1) {$\times$};
\end{tikzpicture}
\end{minipage}
\Bigg)
=2 \begin{minipage}{\mini cm}
\begin{tikzpicture}
\filldraw (\1,0) circle [radius=0.08];
\filldraw (0,0) circle [radius=0.08];
\draw [thick] (0,\1) -- (0,0)  -- (\1,0);
\draw [thick] (0,\1) to [out=-25,in=115] (\1,0);
\draw [thick] (0,\1) to [out=-65,in=165] (\1,0);
\filldraw [white] (0,\1) circle [radius=0.115];
\draw [thick] (0,\1) circle [radius=0.115];
\node at (0,\1) {$\times$};
\end{tikzpicture}
\end{minipage} \\[5pt]
& \rho \Bigg(
\begin{minipage}{\mini cm}
\begin{tikzpicture}
\draw [thick] (0,\1) -- (\1,\1);
\node at (\1,\1) {\scalebox{0.8}{$\bigstar$}};
\filldraw (0,0) circle [radius=0.08];
\filldraw (\1,0) circle [radius=0.08];
\draw [thick] (\1,0) -- (0,0);
\draw [thick] (0,0) to [out=70,in=-70] (0,\1);
\draw [thick] (0,0) to [out=110,in=-110] (0,\1);
\filldraw [white] (0,\1) circle [radius=0.115];
\draw [thick] (0,\1) circle [radius=0.115];
\node at (0,\1) {$\times$};
\end{tikzpicture}
\end{minipage}
\Bigg) 
= 
\begin{minipage}{\mini cm}
\begin{tikzpicture}
\draw [thick] (0,\1) -- (\1,0);
\filldraw (0,0) circle [radius=0.08];
\filldraw (\1,0) circle [radius=0.08];
\draw [thick] (\1,0) -- (0,0);
\draw [thick] (0,0) to [out=70,in=-70] (0,\1);
\draw [thick] (0,0) to [out=110,in=-110] (0,\1);
\filldraw [white] (0,\1) circle [radius=0.115];
\draw [thick] (0,\1) circle [radius=0.115];
\node at (0,\1) {$\times$};
\end{tikzpicture}
\end{minipage} 
\hspace{-0.1cm} +
\begin{minipage}{\mini cm}
\begin{tikzpicture}
\draw [thick] (0,0) -- (0,\1);
\filldraw (0,0) circle [radius=0.08];
\filldraw (\1,0) circle [radius=0.08];
\draw [thick] (\1,0) -- (0,0);
\draw [thick] (0,0) to [out=70,in=-70] (0,\1);
\draw [thick] (0,0) to [out=110,in=-110] (0,\1);
\filldraw [white] (0,\1) circle [radius=0.115];
\draw [thick] (0,\1) circle [radius=0.115];
\node at (0,\1) {$\times$};
\end{tikzpicture}
\end{minipage}
\end{align*}

\noindent Once again, there is a single invariant for $(\pd , \nv , \Delta ) = (3,4)$ given by $L_1+2L_2$.  The graphical version of \eqref{conseq} is given by
\begin{align} \label{eq: n=3 graph}
\delta_3 \Bigg( \hspace{-0.1cm} 
\begin{minipage}{\mini cm}
\begin{tikzpicture}
\filldraw [white] (0,\1) circle [radius=0.115];
\draw [thick,white] (0,\1) circle [radius=0.115];
\node [white] at (0,\1) {$\times$};
\draw [thick] (0,0) to [out=70,in=-70] (0,\1);
\draw [thick] (0,0) to [out=110,in=-110] (0,\1);
\filldraw (0,0) circle [radius=0.08];
\filldraw (\1,0) circle [radius=0.08];
\filldraw (0,\1) circle [radius=0.08];
\draw [thick] (0,0) to [out=20,in=160] (\1,0);
\draw [thick] (0,0) to [out=-20,in=-160] (\1,0);
\end{tikzpicture}
\end{minipage}
\hspace{-0.3cm} + 2 \
\begin{minipage}{\mini cm}
\begin{tikzpicture}
\filldraw [white] (0,\1) circle [radius=0.115];
\draw [thick,white] (0,\1) circle [radius=0.115];
\node [white] at (0,\1) {$\times$};
\filldraw (0,0) circle [radius=0.08];
\filldraw (\1,0) circle [radius=0.08];
\filldraw (0,\1) circle [radius=0.08];
\draw [thick] (0,0) -- (0,\1) -- (\1,0);
\draw [thick] (0,0) to [out=20,in=160] (\1,0);
\draw [thick] (0,0) to [out=-20,in=-160] (\1,0);
\end{tikzpicture}
\end{minipage}
\hspace{-0.1cm} \Bigg)
&=2 
\begin{minipage}{\mini cm}
\begin{tikzpicture}
\draw [thick] (0,0) to [out=70,in=-70] (0,\1);
\draw [thick] (0,0) to [out=110,in=-110] (0,\1);
\filldraw (0,0) circle [radius=0.08];
\filldraw (\1,0) circle [radius=0.08];
\filldraw [white] (0,\1) circle [radius=0.115];
\draw [thick] (0,\1) circle [radius=0.115];
\node at (0,\1) {$\times$};
\draw [thick] (0,0) to [out=20,in=160] (\1,0);
\draw [thick] (0,0) to [out=-20,in=-160] (\1,0);
\end{tikzpicture}
\end{minipage}
\hspace{-0.18cm} + 2
\begin{minipage}{\mini cm}
\begin{tikzpicture}
\filldraw (0,0) circle [radius=0.08];
\filldraw (\1,0) circle [radius=0.08];
\draw [thick] (0,0) -- (0,\1) -- (\1,0);
\draw [thick] (0,0) to [out=20,in=160] (\1,0);
\draw [thick] (0,0) to [out=-20,in=-160] (\1,0);
\filldraw [white] (0,\1) circle [radius=0.115];
\draw [thick] (0,\1) circle [radius=0.115];
\node at (0,\1) {$\times$};
\end{tikzpicture}
\end{minipage}
\hspace{-0.18cm} +4
\begin{minipage}{\mini cm}
\begin{tikzpicture}
\filldraw (\1,0) circle [radius=0.08];
\filldraw (0,0) circle [radius=0.08];
\draw [thick] (0,\1) -- (\1,0) (0,0)  -- (\1,0);
\draw [thick] (0,0) to [out=70,in=-70] (0,\1);
\draw [thick] (0,0) to [out=110,in=-110] (0,\1);
\filldraw [white] (0,\1) circle [radius=0.115];
\draw [thick] (0,\1) circle [radius=0.115];
\node at (0,\1) {$\times$};
\end{tikzpicture}
\end{minipage} \notag \\
&= \rho\left(
2\begin{minipage}{\mini cm}
\begin{tikzpicture}
\draw [thick] (0,\1) -- (\1,\1);
\draw [thick] (0,\1) -- (0,0);
\node at (\1,\1) {\scalebox{0.8}{$\bigstar$}};
\filldraw (0,0) circle [radius=0.08];
\filldraw (\1,0) circle [radius=0.08];
\filldraw [white] (0,\1) circle [radius=0.115];
\draw [thick] (0,\1) circle [radius=0.115];
\node at (0,\1) {$\times$};
\draw [thick] (0,0) to [out=20,in=160] (\1,0);
\draw [thick] (0,0) to [out=-20,in=-160] (\1,0);
\end{tikzpicture}
\end{minipage}
 -
\begin{minipage}{\mini cm}
\begin{tikzpicture}
\draw [thick] (0,\1) -- (\1,\1);
\draw [thick] (0,\1) .. controls (0,0) and (\1,\1) .. (0,\1);
\node at (\1,\1) {\scalebox{0.8}{$\bigstar$}};
\filldraw (0,0) circle [radius=0.08];
\filldraw (\1,0) circle [radius=0.08];
\filldraw [white] (0,\1) circle [radius=0.115];
\draw [thick] (0,\1) circle [radius=0.115];
\node at (0,\1) {$\times$};
\draw [thick] (0,0) to [out=20,in=160] (\1,0);
\draw [thick] (0,0) to [out=-20,in=-160] (\1,0);
\end{tikzpicture}
\end{minipage}
 + 2
\begin{minipage}{\mini cm}
\begin{tikzpicture}
\draw [thick] (0,\1) -- (\1,\1);
\node at (\1,\1) {\scalebox{0.8}{$\bigstar$}};
\filldraw (0,0) circle [radius=0.08];
\filldraw (\1,0) circle [radius=0.08];
\draw [thick] (0,0) -- (0,\1) -- (\1,0) -- (0,0);
\filldraw [white] (0,\1) circle [radius=0.115];
\draw [thick] (0,\1) circle [radius=0.115];
\node at (0,\1) {$\times$};
\end{tikzpicture}
\end{minipage}
\right)
\end{align}

\noindent So far all we have done is rewrite our results in a new notation.  But the graphical notation is more than just a succinct visual way of expressing the invariant terms.  The following section illustrates the virtue of this approach.

\subsection{New Invariants via the Graphical Approach}

As shown in Appendix \ref{AppB}, the graphical approach allows us to prove many general theorems.  In particular, we have the following useful outcomes:
\begin{enumerate}

\item \label{point: 1}

	Without loss of generality, we can limit our search for invariants to graphs with very specific properties (Appendix \ref{sec: graphical basis}).

\item \label{point: 2}

	There is a simple procedure for obtaining all the independent total derivative relations between the variation terms for each $\pd$, $\nv$ and $\Delta$ (Theorem \ref{thm: loopless times basis}).

\item \label{point: 3}
	
	The graphical method allows a complete classification of $1$-invariants (Theorem \ref{thm: classification of 1-invariants}).
	
\item \label{point: 4}

	The graphical method allows many higher $\pd$-invariant terms to be constructed from lower $\pd$ invariants (Appendix \ref{sec: superposition}).
\end{enumerate}

\noindent We will expound upon the above points by presenting explicit examples.  These examples are generalizable and their invariance is proven in Appendix \ref{AppB}.  However, the reader can also check by brute force that the terms we present are indeed invariant.   We will now summarize points \ref{point: 1} and \ref{point: 2} and will return to point \ref{point: 4} in \S\ref{sec: intro to superposition}.  Point \ref{point: 3} was discussed in \S\ref{sec: Galileons} with technical details in Appendix \ref{sec: linear}.

When building invariants, we need only consider loopless plain-graphs (Proposition \ref{prop: loopless_basis}), since a loop represents $\partial_i \partial_i$ acting on a single $\phi$ and one can always integrate by parts to move one of the $\partial_i$'s to act on the remaining $\phi$'s.  We can also restrict to plain-graphs with vertices of degree no less that $\frac{1}{2}(\pd +1)$ (Proposition \ref{prop: min plain}).  This represents a significant simplification from the previous procedure (\S\ref{sec: intro graphs}).  For instance, in \S\ref{sec: 334}, the graphs $L_3$ and $L_4$ are immediately discarded. 

	Taking the variation of these terms yields $\times$-graphs and we need to determine the total derivative relations between them.  Since all plain-graphs we are considering are loopless, any $\times$-graphs involved in these total derivative relations are also loopless.  The total derivative relations that we need to consider can be obtained with the use of graphs called \emph{Medusas} (Definition \ref{def: Medusa}).  A Medusa is a loopless {$\star$}-graph with all of its {$\star$}-vertices adjacent to the $\times$-vertex and such that the degree of the $\times$-vertex is given by:
	\begin{equation}\label{medusaeqn}
	 \text{deg}(\times)= \pd +1- \nsv,
	 \end{equation} 
	 where, again, $\nsv$ is the number of {$\star$}-vertices.  Note that because  $\text{deg}(\times)\geq \nsv$ for a Medusa, \eqref{medusaeqn} implies that $\nsv \leq \frac{1}{2}( \pd +1) \leq \text{deg} (\times)$ for any Medusa.  Furthermore, we need only consider Medusas with $\times$-vertex and $\bullet$-vertices of degree no less than $\frac{1}{2}( \pd +1)$ (Proposition \ref{prop: min star}).  From each of these Medusas, we obtain a total derivative relation, containing only loopless $\times$-graphs, by applying the map $\rho$ and then omitting all looped graphs (Proposition \ref{prop: pyramid}).  This map is denoted as $\rho^{(0)}$ in Definition \ref{def: loopless r}.  In \S\ref{sec: intro to Medusa}, we will give an introduction to the construction of such total derivative relations from Medusas.  Moreover, this procedure captures all relevant total derivative relations (Theorem \ref{thm: loopless times basis}).  Appendix \ref{sec: graphical basis} provides a systematic treatment of Medusas.  
	 
	 In general, for given $\nv$ and $\pd$, we call an invariant consisting of graphs containing the lowest possible value of $\Delta$ a \emph{minimal} invariant (Definition \ref{def: minimal}). Minimal invariants are of particular interest in a QFT, since they are the most relevant $\nv$-point interactions. In \S\ref{sec: 245} and \S\ref{sec: 346} we apply the graphical approach to classify all minimal invariants for $\nv =4$ and $\pd =2,3$.
	
\subsubsection{The Minimal Invariant: \texorpdfstring{$(\pd , \nv , \Delta )= (2,4,5)$}{P=2, n=4, Delta=5}} \label{sec: 245} 

	As our first example, let us find the minimal 2-invariant for $\nv =4$.  For $\pd = 2$, any Medusa must have $\nsv \leq \frac{1}{2}( \pd +1)=\frac{3}{2}$, and thus there is exactly one $\star$-vertex in a $\pd =2$ Medusa.  Furthermore, we need only consider Medusas in which each vertex has degree at least 2, since $\frac{1}{2}( \pd +1) = \frac{3}{2}$.  Therefore, the counting implies that we need only consider $\pd =2$ Medusas with $\Delta \geq \nv + 1$.  In particular, when $\nv =4$, the minimal $\Delta$ is 5 (representing terms with 10 derivatives).  In the following we show that there is exactly one 2-invariant with $\Delta = 5$.  The relevant Medusas are:
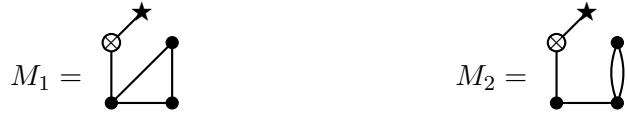
\begin{align*}
& M_1=\begin{minipage}{\mini cm}
\begin{tikzpicture}
\draw [white] (0,0) -- (0.5,-0.8);
\filldraw (0,0) circle [radius=0.08];
\filldraw (\1,\1) circle [radius=0.08];
\filldraw (\1,0) circle [radius=0.08];
\draw [thick] (0,\1) -- (0,0) -- (\1,0) -- (\1,\1) (0,0) -- (\1,\1);
\draw [thick] (0,\1) -- (0.4,1.2);
\node at (0.4,1.2) {\scalebox{0.8}{$\bigstar$}};
\filldraw [white] (0,\1) circle [radius=0.115];
\draw [thick] (0,\1) circle [radius=0.115];
\node at (0,\1) {$\times$};
\end{tikzpicture}
\end{minipage}
&
& M_2=
\begin{minipage}{\mini cm}
\begin{tikzpicture}
\draw [white] (0,0) -- (0.5,-0.8);
\filldraw (0,0) circle [radius=0.08];
\filldraw (\1,0) circle [radius=0.08];
\filldraw (\1,\1) circle [radius=0.08];
\draw [thick] (0,\1)--(0,0) -- (\1,0);
\draw [thick] (0,\1) -- (0.4,1.2);
\node at (0.4,1.2) {\scalebox{0.8}{$\bigstar$}};
\draw [thick] (\1,0) to [out=110,in=-110] (\1,\1);
\draw [thick] (\1,0) to [out=70,in=-70] (\1,\1);
\filldraw [white] (0,\1) circle [radius=0.115];
\draw [thick] (0,\1) circle [radius=0.115];
\node at (0,\1) {$\times$};
\end{tikzpicture}
\end{minipage}
\end{align*}

\vspace{-0.6cm}

\noindent The resulting loopless total derivative relations are:
\begin{equation} \label{eq: 245 derivative relations}
\begin{aligned}
\rho^{(0)}(M_1) = 2
\begin{minipage}{\mini cm}
\begin{tikzpicture}
\filldraw (0,0) circle [radius=0.08];
\filldraw (\1,\1) circle [radius=0.08];
\filldraw (\1,0) circle [radius=0.08];
\draw [thick] (0,\1) -- (0,0) -- (\1,0) -- (\1,\1) (0,0) -- (\1,\1);
\draw [thick] (0,\1) -- (\1,\1);
\filldraw [white] (0,\1) circle [radius=0.115];
\draw [thick] (0,\1) circle [radius=0.115];
\node at (0,\1) {$\times$};
\end{tikzpicture}
\end{minipage}
+
\begin{minipage}{\mini cm}
\begin{tikzpicture}
\filldraw (0,0) circle [radius=0.08];
\filldraw (\1,\1) circle [radius=0.08];
\filldraw (\1,0) circle [radius=0.08];
\draw [thick]  (0,0) -- (\1,0) -- (\1,\1) (0,0) -- (\1,\1);
\draw [thick] (0,0) to [out=110,in=-110] (0,\1);
\draw [thick] (0,0) to [out=70,in=-70] (0,\1);
\filldraw [white] (0,\1) circle [radius=0.115];
\draw [thick] (0,\1) circle [radius=0.115];
\node at (0,\1) {$\times$};
\end{tikzpicture}
\end{minipage}&\equiv 2L_a^{ \times }+L_e^{ \times }
\cr
\rho^{(0)}(M_2)=
\begin{minipage}{\mini cm}
\begin{tikzpicture}
\filldraw (0,0) circle [radius=0.08];
\filldraw (\1,0) circle [radius=0.08];
\filldraw (\1,\1) circle [radius=0.08];
\draw [thick] (0,0) -- (\1,0);
\draw [thick] (0,0) to [out=110,in=-110] (0,\1);
\draw [thick] (0,0) to [out=70,in=-70] (0,\1);
\draw [thick] (\1,0) to [out=110,in=-110] (\1,\1);
\draw [thick] (\1,0) to [out=70,in=-70] (\1,\1);
\filldraw [white] (0,\1) circle [radius=0.115];
\draw [thick] (0,\1) circle [radius=0.115];
\node at (0,\1) {$\times$};\end{tikzpicture}
\end{minipage}
+
\begin{minipage}{\mini cm}
\begin{tikzpicture}
\filldraw (0,0) circle [radius=0.08];
\filldraw (\1,0) circle [radius=0.08];
\filldraw (\1,\1) circle [radius=0.08];
\draw [thick] (0,\1)--(0,0) -- (\1,0);
\draw [thick] (0,\1) -- (\1,\1);
\draw [thick] (\1,0) to [out=110,in=-110] (\1,\1);
\draw [thick] (\1,0) to [out=70,in=-70] (\1,\1);
\filldraw [white] (0,\1) circle [radius=0.115];
\draw [thick] (0,\1) circle [radius=0.115];
\node at (0,\1) {$\times$};
\end{tikzpicture}
\end{minipage}
+
\begin{minipage}{\mini cm}
\begin{tikzpicture}
\filldraw (0,0) circle [radius=0.08];
\filldraw (\1,0) circle [radius=0.08];
\filldraw (\1,\1) circle [radius=0.08];
\draw [thick] (0,\1)--(0,0) -- (\1,0);
\draw [thick] (0,\1) -- (\1,0);
\draw [thick] (\1,0) to [out=110,in=-110] (\1,\1);
\draw [thick] (\1,0) to [out=70,in=-70] (\1,\1);
\filldraw [white] (0,\1) circle [radius=0.115];
\draw [thick] (0,\1) circle [radius=0.115];
\node at (0,\1) {$\times$};
\end{tikzpicture}
\end{minipage}
&\equiv L_b^{ \times }+L_c^{ \times }+L_d^{ \times }
\end{aligned}
\end{equation}

\noindent Note that, when acting on these $\pd =2$ Medusas, $\rho$ and $\rho^{(0)}$ are in fact the same.
	
	On the other hand, the invariants must be constructed out of plain-graphs containing vertices of degree no less than $\frac{1}{2}( \pd +1)=\frac{3}{2}$.  The only possibilities are:
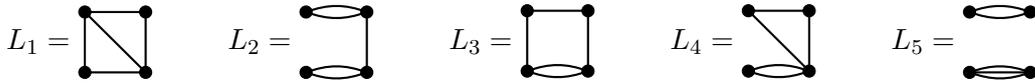
\begin{align*}
L_1 &= 
\begin{minipage}{1.4cm}
\begin{tikzpicture}
\filldraw (0,0) circle [radius=0.08];
\filldraw (\1,0) circle [radius=0.08];
\filldraw (\1,\1) circle [radius=0.08];
\filldraw (0,\1) circle [radius=0.08];
\draw [thick] (0,\1) -- (0,0) -- (\1,0) -- (\1,\1) -- (0,\1) -- (\1,0);
\end{tikzpicture}
\end{minipage} &%
L_2 &=
\begin{minipage}{1.4cm}
\begin{tikzpicture}
\filldraw (0,0) circle [radius=0.08];
\filldraw (\1,0) circle [radius=0.08];
\filldraw (\1,\1) circle [radius=0.08];
\filldraw (0,\1) circle [radius=0.08];
\draw [thick] (\1,0) -- (\1,\1);
\draw [thick] (0,0) to [out=20,in=160] (\1,0);
\draw [thick] (0,0) to [out=-20,in=-160] (\1,0);
\draw [thick] (0,\1) to [out=20,in=160] (\1,\1);
\draw [thick] (0,\1) to [out=-20,in=-160] (\1,\1);
\end{tikzpicture}
\end{minipage} &%
L_3 &= 
\begin{minipage}{1.4cm}
\begin{tikzpicture}
\filldraw (0,0) circle [radius=0.08];
\filldraw (\1,0) circle [radius=0.08];
\filldraw (\1,\1) circle [radius=0.08];
\filldraw (0,\1) circle [radius=0.08];
\draw [thick] (\1,\1) -- (0,\1) -- (0,0) (\1,0) -- (\1,\1);
\draw [thick] (0,0) to [out=20,in=160] (\1,0);
\draw [thick] (0,0) to [out=-20,in=-160] (\1,0);
\end{tikzpicture}
\end{minipage} &%
L_4 &= 
\begin{minipage}{1.4cm}
\begin{tikzpicture}
\filldraw (0,0) circle [radius=0.08];
\filldraw (\1,0) circle [radius=0.08];
\filldraw (\1,\1) circle [radius=0.08];
\filldraw (0,\1) circle [radius=0.08];
\draw [thick] (\1,0) -- (0,\1) --(\1,\1) (\1,0) -- (\1,\1);
\draw [thick] (0,0) to [out=20,in=160] (\1,0);
\draw [thick] (0,0) to [out=-20,in=-160] (\1,0);
\end{tikzpicture}
\end{minipage} &%
L_5 &= 
\begin{minipage}{1.4cm}
\begin{tikzpicture}
\filldraw (0,0) circle [radius=0.08];
\filldraw (\1,0) circle [radius=0.08];
\filldraw (\1,\1) circle [radius=0.08];
\filldraw (0,\1) circle [radius=0.08];
\draw [thick] (0,0) -- (\1,0);
\draw [thick] (0,0) to [out=20,in=160] (\1,0);
\draw [thick] (0,0) to [out=-20,in=-160] (\1,0);
\draw [thick] (0,\1) to [out=20,in=160] (\1,\1);
\draw [thick] (0,\1) to [out=-20,in=-160] (\1,\1);
\end{tikzpicture}
\end{minipage}
\end{align*}

\noindent The invariant cannot be constructed out of $L_5$ since $\delta_2(L_5)$ is absent from the total derivative relations \eqref{eq: 245 derivative relations}.  Hence, we need only consider the variations of $L_1$, $L_2$, $L_3$ and $L_4$.

	We can now determine the 2-invariants.  The total derivative relations allow us to identify $L_d^{ \times }\sim-L_b^{ \times }-L_c^{ \times }$ and $L_e^{ \times }\sim-2L_a^{ \times }$.  Up to total derivatives,
\begin{align} \label{eq: matrix}
	\delta_2\left(\begin{array}{c}L_1 \\L_2 \\L_3 \\L_4\end{array}\right)=\left(\begin{array}{ccccc}2 & 0 & 0 & 0 & 0 \\0 & 2 & 0 & 0 & 0 \\0 & 0 & 2 & 0 & 0 \\0 & 0 & 0 & 2 & 1\end{array}\right)\left(\begin{array}{c}L_a^{ \times } \\L_b^{ \times } \\L_c^{ \times } \\L_d^{ \times } \\L_e^{ \times }\end{array}\right)\sim\left(\begin{array}{ccc}2 & 0 & 0\\0 & 2 & 0 \\0 & 0 & 2 \\-2 & -1 & -1 \end{array}\right)\left(\begin{array}{c}L_a^{ \times } \\L_b^{ \times } \\L_c^{ \times } \end{array}\right)
\end{align}
The invariants form the nullspace of the transpose of the final $4 \times 3$ matrix in \eqref{eq: matrix}.  The nullspace is spanned by $(1,1,1,1)$.  Therefore, there is one 2-invariant given by the linear combination $L_1+L_2+L_3+L_4$, i.e.,
\begin{align} \label{eq: 245}
\begin{minipage}{1.4cm}
\begin{tikzpicture}
\filldraw (0,0) circle [radius=0.08];
\filldraw (\1,0) circle [radius=0.08];
\filldraw (\1,\1) circle [radius=0.08];
\filldraw (0,\1) circle [radius=0.08];
\draw [thick] (0,\1) -- (0,0) -- (\1,0) -- (\1,\1) -- (0,\1) -- (\1,0);
\end{tikzpicture}
\end{minipage}
+
\begin{minipage}{1.4cm}
\begin{tikzpicture}
\filldraw (0,0) circle [radius=0.08];
\filldraw (\1,0) circle [radius=0.08];
\filldraw (\1,\1) circle [radius=0.08];
\filldraw (0,\1) circle [radius=0.08];
\draw [thick] (\1,0) -- (\1,\1);
\draw [thick] (0,0) to [out=20,in=160] (\1,0);
\draw [thick] (0,0) to [out=-20,in=-160] (\1,0);
\draw [thick] (0,\1) to [out=20,in=160] (\1,\1);
\draw [thick] (0,\1) to [out=-20,in=-160] (\1,\1);
\end{tikzpicture}
\end{minipage}
+
\begin{minipage}{1.4cm}
\begin{tikzpicture}
\filldraw (0,0) circle [radius=0.08];
\filldraw (\1,0) circle [radius=0.08];
\filldraw (\1,\1) circle [radius=0.08];
\filldraw (0,\1) circle [radius=0.08];
\draw [thick] (\1,\1) -- (0,\1) -- (0,0) (\1,0) -- (\1,\1);
\draw [thick] (0,0) to [out=20,in=160] (\1,0);
\draw [thick] (0,0) to [out=-20,in=-160] (\1,0);
\end{tikzpicture}
\end{minipage}
+
\begin{minipage}{1.4cm}
\begin{tikzpicture}
\filldraw (0,0) circle [radius=0.08];
\filldraw (\1,0) circle [radius=0.08];
\filldraw (\1,\1) circle [radius=0.08];
\filldraw (0,\1) circle [radius=0.08];
\draw [thick] (\1,0) -- (0,\1) --(\1,\1) (\1,0) -- (\1,\1);
\draw [thick] (0,0) to [out=20,in=160] (\1,0);
\draw [thick] (0,0) to [out=-20,in=-160] (\1,0);
\end{tikzpicture}
\end{minipage}
\end{align}

\noindent Therefore, \eqref{eq: 245} gives the only independent minimal 2-invariant for $\nv =4$.

\subsubsection{The Minimal Invariant: \texorpdfstring{$(\pd , \nv , \Delta ) = (3,4,6)$}{P=3, n=4, Delta=6}} \label{sec: 346} 

	Next, let us consider $\pd =3$, $\nv =4$, for which each vertex degree must be at least $\frac{1}{2} ( \pd+1) = 2$.  Again we would like to find the minimal invariant in this case.  By counting alone, it is possible to write down Medusas with $\Delta=5$.  In fact, a 3-invariant with $\nv =4$ and $\Delta=5$ would also be 2-invariant.  The only possible 2-invariant with $(\nv , \Delta ) = (4,5)$ is \eqref{eq: 245}.  However, this is not a 3-invariant because it is impossible for some graphs contained in it to appear in a 3-invariant.  For example, by replacing a degree-2 vertex in $L_2$ in $\eqref{eq: 245}$ with a $\times$-vertex, a $\times$-graph $\Gamma^{ \times }$ is produced;  $\Gamma^{ \times }$ and the only $\pd =3$ Medusa $M$ that generates $\Gamma^{ \times }$ are given below:
\vspace{-0.5cm}
\begin{align*}
&
\Gamma^{ \times } = \hspace{-0.6cm}
\begin{minipage}{2cm}

\end{minipage}
\end{align*}
	
\vspace{-0.6cm}
	
\noindent But $M$ contains a $\bullet$-vertex of degree lower than 2, and therefore \eqref{eq: 245} cannot be 3-invariant.  This sets a lower bound for $\Delta$: $\Delta \geq 6$.  For $\Delta = 6$, the Medusas are:
\vspace{-2mm}
\begin{align*}
&M_1=
\hspace{-0.4cm}
\begin{minipage}{2cm}
%
\end{minipage}
\end{align*}

\vspace{-0.7cm}

\noindent These give the following total derivative relations:
\begin{eqnarray*}
\rho^{(0)}(M_1)=
\begin{minipage}{\mini cm}
%
\end{minipage}
&\equiv&
2L_p^{ \times }+L_q^{ \times }+2L_r^{ \times }+4L_s^{ \times }
\end{eqnarray*}

\noindent Then the invariants must be made up of plain-graphs whose variations are contained in the total derivative relations above.  In other words, the invariants are made from:
\begin{align*}
&L_1=\begin{minipage}{1.4cm}
%
\end{minipage}
\end{align*}

\noindent We can now determine the 3-invaraints:
\begin{equation*}
\delta_3\left(
\begin{array}{c}
	L_1 \\[5pt]
	L_2 \\[5pt]
	L_3 \\[5pt]
	L_4 \\[5pt]
	L_5 \\[5pt]
	L_6 \\[5pt]
	L_7 \\[5pt]
	L_8 \\[5pt]
	L_9 \\[5pt]
	L_{10} 
\end{array}\right)
= \left(
\begin{array}{ccccccccccccccccccc}
	4& 0 & 0 & 0 & 0 & 0 & 0 & 0 & 0 & 0 & 0 & 0 & 0 & 0 & 0 & 0 & 0 & 0 & 0 \\[5pt]
	0 & 2 & 0 & 0 & 0 & 0 & 0 & 0 & 0 & 0 & 0 & 0& 1 & 0 & 0 & 0 & 0 & 0 & 0 \\[5pt]
	0 & 0 & 0 & 0 & 1 & 0 & 0 & 1 & 0 & 0 & 0 & 0 & 0 & 1 & 0 & 0 & 0 & 0 & 0\\[5pt]
	0 & 0 & 0 & 0 & 0 & 4 & 0 & 0 & 0 & 0 & 0 & 0 & 0 & 0 & 0 & 0 & 0 & 0 & 0\\[5pt]
	0 & 0 & 0 & 0 & 0 & 0 & 0 & 0 & 2 & 0 & 0 & 0 & 0 & 0 & 0 & 0 & 0 & 1 & 0\\[5pt]
	0 & 0 & 0 & 0 & 0 & 0 & 0 & 0 & 0 & 4 & 0 & 0& 0 & 0 & 0 & 0 & 0 & 0 & 0 \\[5pt]
	0 & 0 & 0 & 0 & 0 & 0 & 0 & 0 & 0 & 0 & 1 & 1 & 0 & 0 & 0 & 0 & 0 & 0 & 1\\[5pt]
	0 & 0 & 0 & 0 & 0 & 0 & 0 & 0 & 0 & 0 & 0 & 0 & 0 & 0 & 2 & 0 & 0 & 0 & 0\\[5pt]
	0 & 0 & 0 & 0 & 0 & 0 & 0 & 0 & 0 & 0 & 0 & 0 & 0 & 0 & 0 & 2 & 0 & 0 & 0\\[5pt]
	0 & 0 & 0 & 0 & 0 & 0 & 0 & 0 & 0 & 0 & 0 & 0 & 0 & 0 & 0 & 0 & 3 & 0 & 0
\end{array}
\right) \left(
\begin{array}{c}
	L_a^{ \times } \\L_b^{ \times } \\L_c^{ \times } \\L_d^{ \times } \\L_e^{ \times } \\L_f^{ \times } \\L_g^{ \times } \\L_h^{ \times } \\L_i^{\times} \\L_j^{\times} \\L_k^{\times} \\L_\ell^{\times} \\L_m^{\times} \\L_n^{\times} \\L_o^{\times} \\L_p^{\times} \\L_q^{\times} \\L_r^{\times} \\L_s^{ \times }
\end{array}\right)
\end{equation*}
\begin{equation*}
\phantom{\delta_3\left(\begin{array}{c}L_1 \\L_2 \\L_3 \\L_4 \\L_5 \\L_6 \\L_7 \\L_8 \\L_9 \\L_{10} \end{array}\right)}
\sim \left(
\begin{array}{ccccccccccccccccccc}
	4& 0  & 0  & 0 & 0 & 0 & 0 & 0 & 0 & 0 & 0 \\
	0 & 2  & 0 & 0 & 0 & -2&  -4 & -2 & 0  & 0 & 0 \\
	0 & 0  & 1 & 0  & 1  & 0 & 1 & 0 & 0 & 0 & 0\\
	0 & 0 & 0 & 4  & 0 & 0 & 0 & 0 & 0  & 0 & 0\\
	-2 & 0 & 0&  0 &  -2  &0   & 0 & 0 & 0  & 1 & 0\\
	0 & 0 & -8 & 0 & 0  & 0& 0 & 0 & 0 &  0 & 0 \\
	0 & -1  & 0 & 0 & -1 & 1 & 0 & 0 & 0  & 0 & 1\\
	0 & 0  & 0 & 0 & 0  & 0 &  0 & 2 & 0  & 0 & 0\\
	0 & 0 & 0 & 0 & 0  & 0 & 0 & 0 & 2  & 0 & 0\\
	0 & 0 & 0 & 0 & 0  & 0 & 0 & 0 & -6  & -6 & -12
\end{array}
\right) \left(
\begin{array}{c}
	L_a^{ \times } \\L_b^{ \times } \\L_e^{ \times } \\L_f^{ \times } \\L_h^{ \times }  \\L_l^{ \times }\\L_n^{ \times }\\L_o^{ \times }\\L_p^{ \times }\\L_r^{ \times }\\L_s^\times
\end{array}
\right)
\end{equation*}

\noindent The nullspace of the transpose of this matrix is spanned by $(3, 6, 24, 0, 6, 3, 12, 6, 3, 1)$, giving the only invariant linear combination,
\begin{align} \label{eq: 346}
3 \, &
\begin{minipage}{1.4cm}
%
\end{minipage}
\end{align}

\vspace{1mm}

\noindent This gives the only independent minimal 3-invariant for $\nv =4$.

	Note that $L_4$ does not appear in the invariant. Indeed, it can be discarded immediately, since the unique Medusa associated with $L_4$ is
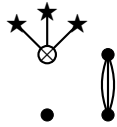
\vspace{-0.2cm}
\begin{align*}
\begin{minipage}{2cm}
%
\end{minipage}
\end{align*}

\noindent This Medusa has an empty vertex, which violates the lower bound on vertex degree.

\subsubsection{Medusas and Total Derivative Relations} \label{sec: intro to Medusa} 

	We have seen that Medusas play a central role in the search for $\pd$-invariants.  It is thus worthwhile to discuss the key features of Medusas and to demonstrate how a total derivative relation consisting of loopless $\times$-graphs is constructed from a Medusa.  Given any Medusa, $\rho^{(0)} ( M )$ is in fact a total derivative relation, as can be seen from the following construction.
	
	For fixed $\pd$ and $\nv$, consider a Medusa $M$ that contains $\nsv$ $\star$-vertices.  Then, by definition, it has a $\times$-vertex of degree deg$(\times) = \pd +1- \nsv$.  Since the maximal degree of a $\times$-vertex is $\pd$, graphs in $\rho ( M )$ contain at most $\nsv - 1$ loops.  Form a $\star$-graph $\Gamma^{(\ell)}$ from $M$ by deleting $\ell \leq \nsv - 1$ $\star$-vertices in $M$ and then adding $\ell$ loops to the $\times$-vertex.  By this definition, $M = \Gamma^{(0)}$.   In the algebraic expression represented by $\Gamma^{(\ell)}$, the $\nsv - \ell$ $\star$-vertices stand for $\nsv - \ell$ partial derivatives acting on the whole term.  Distributing $\nsv - \ell - 1$ $\partial$'s over all $\phi$'s in this algebraic expression will result in a linear combination of total derivative terms.  In the graphical representation, this is equivalent to acting $\rho$ on $\Gamma^{(\ell)}$ but keeping fixed exactly one $\star$-vertex and its incident edge.  Setting to zero all coefficients of graphs in the resulting linear combination, except for the ones containing exactly $\ell$ loops, generates a linear combination $L^{(\ell)}$ of $\star$-graphs, each containing exactly one $\star$-vertex.  By construction,
\begin{equation} \label{eq: loopless rel Medusa}
	\rho^{(0)} ( M ) = \rho \left ( \sum_{\alpha=0}^{\nsv -1} (-1)^\alpha L^{(\alpha)} \right ).
\end{equation}

\noindent The algebraic form of the RHS of \eqref{eq: loopless rel Medusa} is explicitly a total derivative relation.  For a rigorous treatment of the above discussion, refer to Proposition \ref{prop: pyramid} in Appendix \ref{sec: pyramid}.  

	As a simple example, we consider the Medusa $M_7$ referred to in \S\ref{sec: 346}. We have
\vspace{-0.2cm}
\begin{align*}
M_7 = 
\hspace{-0.4cm}
\begin{minipage}{2cm}
\begin{tikzpicture}
\node at (-0.4,1.2) {$\phantom{\bigstar}$};
\draw [white] (0,0) -- (0.5,-0.6);
\filldraw (\1,0) circle [radius=0.08];
\filldraw (0,0) circle [radius=0.08];
\filldraw (\1,\1) circle [radius=0.08];
\draw [thick] (\1,0) -- (0,0) -- (\1,\1) ;
\draw [thick] (0,\1) -- (0.4,1.2);
\node at (0.4,1.2) {\scalebox{0.8}{$\bigstar$}};
\draw [thick] (0,\1) -- (-0.4,1.2);
\node at (-0.4,1.2) {\scalebox{0.8}{$\bigstar$}};
\draw [thick] (\1,0) to [out=110,in=-110] (\1,\1);
\draw [thick] (\1,0) to [out=70,in=-70] (\1,\1);
\filldraw [white] (0,\1) circle [radius=0.115];
\draw [thick] (0,\1) circle [radius=0.115];
\node at (0,\1) {$\times$};
\end{tikzpicture}
\end{minipage}
\Rightarrow
\hspace{0.2cm}
\rho^{(0)} ( M_7 )
=
\rho \Biggl ( 
\begin{minipage}{\mini cm}
\begin{tikzpicture}
\draw [white] (0,0) -- (0.5,-0.6);
\draw [thick] (0,\1) -- (0,0);
\filldraw (\1,0) circle [radius=0.08];
\filldraw (0,0) circle [radius=0.08];
\filldraw (\1,\1) circle [radius=0.08];
\draw [thick] (\1,0) -- (0,0) -- (\1,\1) ;
\draw [thick] (0,\1) -- (0.4,1.2);
\node at (0.4,1.2) {\scalebox{0.8}{$\bigstar$}};
\draw [thick] (\1,0) to [out=110,in=-110] (\1,\1);
\draw [thick] (\1,0) to [out=70,in=-70] (\1,\1);
\filldraw [white] (0,\1) circle [radius=0.115];
\draw [thick] (0,\1) circle [radius=0.115];
\node at (0,\1) {$\times$};
\end{tikzpicture}
\end{minipage}
+ 2
\begin{minipage}{\mini cm}
\begin{tikzpicture}
\draw [white] (0,0) -- (0.5,-0.6);
\draw [thick] (0,\1) -- (\1,\1);
\filldraw (\1,0) circle [radius=0.08];
\filldraw (0,0) circle [radius=0.08];
\filldraw (\1,\1) circle [radius=0.08];
\draw [thick] (\1,0) -- (0,0) -- (\1,\1) ;
\draw [thick] (0,\1) -- (0.4,1.2);
\node at (0.4,1.2) {\scalebox{0.8}{$\bigstar$}};
\draw [thick] (\1,0) to [out=110,in=-110] (\1,\1);
\draw [thick] (\1,0) to [out=70,in=-70] (\1,\1);
\filldraw [white] (0,\1) circle [radius=0.115];
\draw [thick] (0,\1) circle [radius=0.115];
\node at (0,\1) {$\times$};
\end{tikzpicture}
\end{minipage} 
-
\begin{minipage}{\mini cm}
\begin{tikzpicture}
\draw [white] (0,0) -- (0.5,-0.6);
\draw [thick] (0,\1) .. controls (0,0) and (\1,\1) .. (0,\1);
\filldraw (\1,0) circle [radius=0.08];
\filldraw (0,0) circle [radius=0.08];
\filldraw (\1,\1) circle [radius=0.08];
\draw [thick] (\1,0) -- (0,0) -- (\1,\1) ;
\draw [thick] (0,\1) -- (0.4,1.2);
\node at (0.4,1.2) {\scalebox{0.8}{$\bigstar$}};
\draw [thick] (\1,0) to [out=110,in=-110] (\1,\1);
\draw [thick] (\1,0) to [out=70,in=-70] (\1,\1);
\filldraw [white] (0,\1) circle [radius=0.115];
\draw [thick] (0,\1) circle [radius=0.115];
\node at (0,\1) {$\times$};
\end{tikzpicture}
\end{minipage}
\Biggr ) 
\end{align*}

\vspace{-0.5cm}

\noindent For a second example, we consider $(\pd , \nv , \Delta ) = (5,3,6)$ and the Medusa
\begin{align*}
M = \hspace{-0.4cm}
\begin{minipage}{2.3cm}
\begin{tikzpicture}
\draw [white] (0,0) -- (0.5,-0.8);
\draw [thick] (0,\1) -- (0.4,1.2);
\node at (0.4,1.2) {\scalebox{0.8}{$\bigstar$}};
\draw [thick] (0,\1) -- (-0.4,1.2);
\node at (-0.4,1.2) {\scalebox{0.8}{$\bigstar$}};
\draw [thick] (0,\1) -- (\1,\1);
\node at (\1,\1) {\scalebox{0.8}{$\bigstar$}};
\filldraw (0,0) circle [radius=0.08];
\filldraw (\1,0) circle [radius=0.08];
\draw [thick] (0,0) to [out=20,in=160] (\1,0);
\draw [thick] (0,0) to [out=-20,in=-160] (\1,0);
\draw [thick] (\1,0) -- (0,0);
\filldraw [white] (0,\1) circle [radius=0.115];
\draw [thick] (0,\1) circle [radius=0.115];
\node at (0,\1) {$\times$};
\end{tikzpicture}
\end{minipage}
\end{align*}

\vspace{-0.7cm}

\noindent By \eqref{eq: loopless rel Medusa}, we obtain
\begin{equation*}
	\rho^{(0)} ( M ) = \rho \left ( 
2
\begin{minipage}{\mini cm}
\begin{tikzpicture}
\draw [thick] (0,0) to [out=110,in=-110] (0,\1);
\draw [thick] (0,0) to [out=70,in=-70] (0,\1);
\draw [thick] (0,\1) -- (\1,\1);
\node at (\1,\1) {\scalebox{0.8}{$\bigstar$}};
\filldraw (0,0) circle [radius=0.08];
\filldraw (\1,0) circle [radius=0.08];
\draw [thick] (0,0) to [out=20,in=160] (\1,0);
\draw [thick] (0,0) to [out=-20,in=-160] (\1,0);
\draw [thick] (0,0) to [out=20,in=160] (\1,0);
\draw [thick] (0,0) to [out=-20,in=-160] (\1,0);
\draw [thick] (\1,0) -- (0,0);
\filldraw [white] (0,\1) circle [radius=0.115];
\draw [thick] (0,\1) circle [radius=0.115];
\node at (0,\1) {$\times$};
\end{tikzpicture}
\end{minipage}
+2
\begin{minipage}{\mini cm}
\begin{tikzpicture}
\draw [thick] (0,0) -- (0,\1) -- (\1,\1) (\1,0) -- (0,\1);
\node at (\1,\1) {\scalebox{0.8}{$\bigstar$}};
\filldraw (0,0) circle [radius=0.08];
\filldraw (\1,0) circle [radius=0.08];
\draw [thick] (0,0) to [out=20,in=160] (\1,0);
\draw [thick] (0,0) to [out=-20,in=-160] (\1,0);
\draw [thick] (\1,0) -- (0,0);
\filldraw [white] (0,\1) circle [radius=0.115];
\draw [thick] (0,\1) circle [radius=0.115];
\node at (0,\1) {$\times$};
\end{tikzpicture}
\end{minipage}
- 2
\begin{minipage}{\mini cm}
\begin{tikzpicture}
\draw [thick] (0,\1) .. controls (0,0) and (\1,\1) .. (0,\1);
\draw [thick] (0,0) -- (0,\1) -- (\1,\1);
\node at (\1,\1) {\scalebox{0.8}{$\bigstar$}};
\filldraw (0,0) circle [radius=0.08];
\filldraw (\1,0) circle [radius=0.08];
\draw [thick] (0,0) to [out=20,in=160] (\1,0);
\draw [thick] (0,0) to [out=-20,in=-160] (\1,0);
\draw [thick] (\1,0) -- (0,0);
\filldraw [white] (0,\1) circle [radius=0.115];
\draw [thick] (0,\1) circle [radius=0.115];
\node at (0,\1) {$\times$};
\end{tikzpicture}
\end{minipage}
+\hspace{-0.4cm}
\begin{minipage}{3cm}
\begin{tikzpicture}
\draw [thick] (0,\1) .. controls (0,0) and (\1,\1) .. (0,\1);
\draw [thick] (0,\1) .. controls (0,0) and (-\1,\1) .. (0,\1);
\draw [thick] (0,\1) -- (\1,\1);
\node at (\1,\1) {\scalebox{0.8}{$\bigstar$}};
\filldraw (0,0) circle [radius=0.08];
\filldraw (\1,0) circle [radius=0.08];
\draw [thick] (0,0) to [out=20,in=160] (\1,0);
\draw [thick] (0,0) to [out=-20,in=-160] (\1,0);
\draw [thick] (\1,0) -- (0,0);
\filldraw [white] (0,\1) circle [radius=0.115];
\draw [thick] (0,\1) circle [radius=0.115];
\node at (0,\1) {$\times$};
\end{tikzpicture}
\end{minipage}
\hspace{-1cm} 
\right ).
\end{equation*}

\noindent This Medusa is involved in a 5-invariant that we will construct in \S\ref{sec: intro to superposition}.

\subsection{Superposition of Graphs} \label{sec: intro to superposition} 

	In \S\ref{sec: Galileons}, we discovered an intriguing construction of the minimal 1-invariant for given $\nv$, which is a sum with equal coefficients of all possible trees with $\nv$ vertices.  A close study of
the $\pd$-invariants with $\pd > 1$ in \S\ref{sec: beyond} also reveals an elegant structure in these invariants: They can all be decomposed as a \emph{superposition} of equal-weight tree summations and \emph{exact} invariants.  Recall that an exact invariant is a linear combination that is invariant exactly, instead of up to a total derivative.  In the graphical representation, a linear combination of graphs is an exact $\pd_E$-invariant if and only if all vertices are of degree larger than $\pd_E$ (Corollary \ref{lem: exact}).  Each graph in an exact invariant is itself exactly invariant.  
	
	Next we illustrate ``the superposition of linear combinations" by explicit examples.  When appropriate, we will consider labeled graphs and only remove the labels at the end.  Consider two labeled graphs,
\begin{align*}
\Gamma_1 = 
\begin{minipage}{1.3cm}
\begin{tikzpicture}
\draw [thick] (0.4,0.69) -- (0,0) -- (\1,0);
\filldraw (0,0) circle [radius=0.08];
\filldraw (\1,0) circle [radius=0.08];
\filldraw (0.4,0.69) circle [radius=0.08];
\end{tikzpicture}
\end{minipage}
\quad\quad
\Gamma_2 = 
\begin{minipage}{1.3cm}
\begin{tikzpicture}
\draw [thick, \red] (0,0) to [out=-30,in=-150] (\1,0);
\draw [thick, \red] (\1,0) to [out=90,in=-30] (0.4,0.69);
\filldraw (0,0) circle [radius=0.08];
\filldraw (\1,0) circle [radius=0.08];
\filldraw (0.4,0.69) circle [radius=0.08];
\end{tikzpicture}
\end{minipage}
\end{align*}

\noindent The superposition of $\Gamma_1$ and $\Gamma_2$ is defined to be the graph formed by taking all edges in $\Gamma_2$ and adding them to $\Gamma_1$, i.e.,
\begin{align*}
\Gamma_1 \cup \Gamma_2 =
\begin{minipage}{1.3cm}
\begin{tikzpicture}
\draw [thick] (0,0) -- (\1,0);
\draw [thick, \red] (0,0) to [out=-30,in=-150] (\1,0);
\draw [thick, \red] (\1,0) to [out=90,in=-30] (0.4,0.69);
\draw [thick] (0.4,0.69) -- (0,0);
\filldraw (0,0) circle [radius=0.08];
\filldraw (\1,0) circle [radius=0.08];
\filldraw (0.4,0.69) circle [radius=0.08];
\end{tikzpicture}
\end{minipage}
\end{align*}

\noindent The superposition of two linear combinations, $L_A = \sum_{i=1}^{k_A} a_i \hspace{0.6mm} \Gamma^A_i$ and $L_B = \sum_{i=1}^{k_B} b_i \hspace{0.6mm} \Gamma^B_i$, of plain graphs $\Gamma^A_i, \Gamma^B_j$ with the same $N$, is defined as 
\begin{equation*}
	L_A \cup L_B \equiv \sum_{i=1}^{k_A} \sum_{j=1}^{k_B} a_i \hspace{0.6mm} b_j \hspace{0.6mm} \Gamma^A_i \cup \Gamma^B_j.
\end{equation*}

\noindent In the following we present numerous examples of invariants constructed by superposing equal-weight tree summations and exact invariants for various $\pd$'s.  In fact, Theorem \ref{thm: superposition summary} of Appendix \ref{sec: superposition} states: 
\begin{itemize}
\item[]
{\it For fixed $\nv$, the superposition of an exact $\pd_E$-invariant with the superposition of $\nts$ minimal loopless 1-invariants results in a $\pd$-invariant, provided $\pd_E + 2 \nts \geq \pd$.}\footnote{Note that this theorem also applies when $\pd_E<0$, where we take an exact $\pd_E$-invariant for $\pd_E<0$ to mean any linear combination of plain-graphs.  In particular, the plain-graph consisting only of empty vertices is a $\pd_E$-invariant for any $\pd_E<0$, and superposing this graph on any other is equivalent to not superposing anything at all.}
\end{itemize}

\noindent We conjecture that the above result captures all $\pd$-invariants, up to total derivatives.  Since we have classified all exact invariants and all 1-invariants (Theorem \ref{thm: classification of 1-invariants}), it is straightforward to construct the $\pd$-invariants in the above statement for any specific case.  We now proceed to construct the minimal $\pd$-invariants for some important cases.

\subsubsection{Quadratic Shift (\textit{\pd}=2)}

\noindent \textbf{\textit{\nv}} = \textbf{3 Case}: A 2-invariant can be constructed by superposing an equal-weight tree summation with an exact 0-invariant.  In the labeled representation, all possible trees for $\nv =3$ are given by \eqref{eq: graph_iso_3},
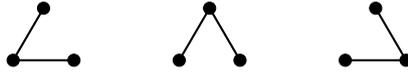
\begin{align*} 
	\begin{minipage}{1.8cm}
    	\begin{tikzpicture}
        		\draw [thick] (0.4,0.69) -- (0,0) -- (\1,0);
        		\filldraw (0,0) circle [radius=0.08];
        		\filldraw (\1,0) circle [radius=0.08];
		\filldraw (0.4,0.69) circle [radius=0.08];
    	\end{tikzpicture}
    	\end{minipage}
	\quad
	\begin{minipage}{1.8cm}
    	\begin{tikzpicture}
        		\draw [thick] (0,0) -- (0.4,0.69) -- (\1,0);
        		\filldraw (0,0) circle [radius=0.08];
        		\filldraw (\1,0) circle [radius=0.08];
		\filldraw (0.4,0.69) circle [radius=0.08];
    	\end{tikzpicture}
    	\end{minipage}
	\quad
	\begin{minipage}{1.3cm}
    	\begin{tikzpicture}
        		\draw [thick] (0,0) -- (\1,0) -- (0.4,0.69);
        		\filldraw (0,0) circle [radius=0.08];
        		\filldraw (\1,0) circle [radius=0.08];
		\filldraw (0.4,0.69) circle [radius=0.08];
    	\end{tikzpicture}
    	\end{minipage}
\end{align*}

\noindent The sum of all three $\nv =3$ trees with unit coefficients gives a 1-invariant, $L_{3\text{-pt}}$.  On the other hand, up to total derivatives, an exact 0-invariant with $\Delta =2$ is isomorphic to 
\begin{equation*}
\Gamma_0 = 
\begin{minipage}{1.8cm}
\begin{tikzpicture}
\draw [thick] (0.4,0.69) -- (0,0) -- (\1,0);
\filldraw (0,0) circle [radius=0.08];
\filldraw (\1,0) circle [radius=0.08];
\filldraw (0.4,0.69) circle [radius=0.08];
\end{tikzpicture}
\end{minipage}
\end{equation*}

\noindent Then $\Gamma_0 \cup L_{3\text{-pt}}$ contains the three superposed graphs as follows:
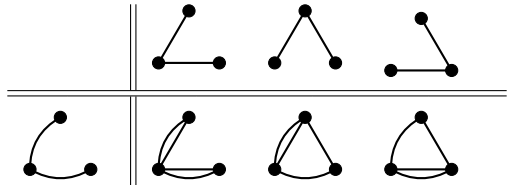
\begin{center}
\begin{tabular}{l || cc}
&
\begin{minipage}{\mini cm}
\begin{tikzpicture}
\draw [white] (0,0) -- (0,-0.28);
\draw [thick] (0.4,0.69) -- (0,0) -- (\1,0);
\filldraw (0,0) circle [radius=0.08];
\filldraw (\1,0) circle [radius=0.08];
\filldraw (0.4,0.69) circle [radius=0.08];
\end{tikzpicture}
\end{minipage}
\begin{minipage}{\mini cm}
\begin{tikzpicture}
\draw [white] (0,0) -- (0,-0.28);
\draw [thick] (0,0) -- (0.4,0.69) -- (\1,0);
\filldraw (0,0) circle [radius=0.08];
\filldraw (\1,0) circle [radius=0.08];
\filldraw (0.4,0.69) circle [radius=0.08];
\end{tikzpicture}
\end{minipage}
\begin{minipage}{\mini cm}
\begin{tikzpicture}
\draw [white] (0,0) -- (0,--0.28);
\draw [thick] (0.4,0.69) -- (\1,0) -- (0,0);
\filldraw (0.4,0.69) circle [radius=0.08];
\filldraw (0,0) circle [radius=0.08];
\filldraw (\1,0) circle [radius=0.08];
\end{tikzpicture}
\end{minipage} \\[10pt]
\hline\hline 
\begin{minipage}{1.2cm}
\begin{tikzpicture}
\draw [white] (0.4,0.69) -- (0.5,0.96);
\draw [white] (0,0) -- (0,-0.2);
\draw [thick, \red] (0,0) to [out=-30,in=-150] (\1,0);
\draw [thick, \red] (0,0) to [out=90,in=-150] (0.4,0.69);
\filldraw (0,0) circle [radius=0.08];
\filldraw (\1,0) circle [radius=0.08];
\filldraw (0.4,0.69) circle [radius=0.08];
\end{tikzpicture}
\end{minipage} 
&
\begin{minipage}{\mini cm}
\begin{tikzpicture}
\draw [white] (0.4,0.69) -- (0.5,0.96);
\draw [white] (0,0) -- (0,-0.2);
\draw [thick, \red] (0,0) to [out=-30,in=-150] (\1,0);
\draw [thick, \red] (0,0) to [out=90,in=-150] (0.4,0.69);
\draw [thick] (0,0) -- (\1,0);
\draw [thick] (0,0) -- (0.4,0.69);
\filldraw (0,0) circle [radius=0.08];
\filldraw (\1,0) circle [radius=0.08];
\filldraw (0.4,0.69) circle [radius=0.08];
\end{tikzpicture}
\end{minipage}
\begin{minipage}{\mini cm}
\begin{tikzpicture}
\draw [white] (0.4,0.69) -- (0.5,0.96);
\draw [white] (0,0) -- (0,-0.2);
\draw [thick, \red] (0,0) to [out=-30,in=-150] (\1,0);
\draw [thick, \red] (0,0) to [out=90,in=-150] (0.4,0.69);
\draw [thick] (\1,0) -- (0.4,0.69);
\draw [thick] (0,0) -- (0.4,0.69);
\filldraw (0,0) circle [radius=0.08];
\filldraw (\1,0) circle [radius=0.08];
\filldraw (0.4,0.69) circle [radius=0.08];
\end{tikzpicture}
\end{minipage}
\begin{minipage}{\mini cm}
\begin{tikzpicture}
\draw [white] (0.4,0.69) -- (0.5,0.96);
\draw [white] (0,0) -- (0,-0.2);
\draw [thick, \red] (0,0) to [out=-30,in=-150] (\1,0);
\draw [thick, \red] (0,0) to [out=90,in=-150] (0.4,0.69);
\draw [thick] (0.4,0.69) -- (\1,0);
\draw [thick] (0,0) -- (\1,0);
\filldraw (0,0) circle [radius=0.08];
\filldraw (\1,0) circle [radius=0.08];
\filldraw (0.4,0.69) circle [radius=0.08];
\end{tikzpicture}
\end{minipage}
\end{tabular}
\end{center}

\noindent Summing over all superposed graphs with unit coefficients gives a 2-invariant for $\nv =3$ (after identifying isomorphic graphs),
\begin{align} \label{eq: p=2 n=3 Delta=4}
\delta_2 \left ( \hspace{1mm}
\begin{minipage}{\mini cm}
\begin{tikzpicture}
\draw [thick] (0,0) to [out=80,in=-140] (0.4,0.69);
\draw [thick] (0,0) to [out=35,in=-105] (0.4,0.69);
\filldraw (0,0) circle [radius=0.08];
\filldraw (\1,0) circle [radius=0.08];
\filldraw (0.4,0.69) circle [radius=0.08];
\draw [thick] (0,0) to [out=20,in=160] (\1,0);
\draw [thick] (0,0) to [out=-20,in=-160] (\1,0);
\end{tikzpicture}
\end{minipage}
\hspace{-0.2cm} + 2
\begin{minipage}{\mini cm}
\begin{tikzpicture}
\filldraw (0,0) circle [radius=0.08];
\filldraw (\1,0) circle [radius=0.08];
\filldraw (0.4,0.69) circle [radius=0.08];
\draw [thick] (0,0) -- (0.4,0.69) -- (\1,0);
\draw [thick] (0,0) to [out=20,in=160] (\1,0);
\draw [thick] (0,0) to [out=-20,in=-160] (\1,0);
\end{tikzpicture}
\end{minipage}
\hspace{-3mm}
\right )
= \rho^{(0)} 
\left ( 2
\begin{minipage}{\mini cm}
\begin{tikzpicture}
\filldraw (0,0) circle [radius=0.08];
\filldraw (\1,0) circle [radius=0.08];
\draw [thick] (0,0) to [out=20,in=160] (\1,0);
\draw [thick] (0,0) to [out=-20,in=-160] (\1,0);
\draw [thick] (\1,\1) -- (0,\1) -- (0,0);
\node at (\1,\1) {\scalebox{0.8}{$\bigstar$}};
\filldraw [white] (0,\1) circle [radius=0.115];
\draw [thick] (0,\1) circle [radius=0.115];
\node at (0,\1) {$\times$};
\end{tikzpicture}
\end{minipage}
\right ).
\end{align}

\noindent Note that for $\pd =2$ and $\nv =3$, we need only consider Medusas with at least four edges, since that the $\times$-vertex and $\bullet$-vertices have degree no less than 2.  The Medusa in \eqref{eq: p=2 n=3 Delta=4} is the only such Medusa with $\Delta=4$.  Therefore, this is the only independent minimal 2-invariant.  Note that the 2-invariant given in \eqref{eq: p=2 n=3 Delta=4} and the 3-invariant in \eqref{eq: n=3 graph} happen to be the same.  

\vspace{0.3cm}

	In fact, we can prove a general minimality statement for $\nv =3$.  Consider a Medusa with $\Delta = \pd +1$ for odd $\pd$, and $\Delta = \pd +2$ for even $\pd$.  The $\bullet$-vertices of this Medusa have degree at least $\frac{1}{2} \Delta$.  For odd $P$ this already saturates the lower bound for the degree of a $\bullet$-vertex; no edge joining the two $\bullet$-vertices can be removed and thus $\Delta$ cannot be lowered further.  For even $P$, one $\bullet$-vertex saturates the lower bound on vertex degree and the other $\bullet$-vertex has an excess of exactly one edge.  Nevertheless, the same conclusion holds.

\vspace{0.3cm}

\begin{figure}[t!]
\centering
\begin{subfigure}[t]{.48\textwidth} 
\begin{eqnarray*}
\begin{minipage}{1.5cm}
\begin{tikzpicture}
\filldraw (0,0) circle [radius=0.08];
\filldraw (\1,0) circle [radius=0.08];
\filldraw (\1,\1) circle [radius=0.08];
\filldraw (0,\1) circle [radius=0.08];
\draw [thick] (0,\1) -- (0,0) (0,\1) -- (\1,0) (0,\1) -- (\1,\1);
\end{tikzpicture}
\end{minipage}
\begin{minipage}{1.5cm}
\begin{tikzpicture}
\filldraw (0,0) circle [radius=0.08];
\filldraw (\1,0) circle [radius=0.08];
\filldraw (\1,\1) circle [radius=0.08];
\filldraw (0,\1) circle [radius=0.08];
\draw [thick] (0,\1) -- (0,0) (0,\1) -- (\1,\1) -- (\1,0);
\end{tikzpicture}
\end{minipage}
\begin{minipage}{1.5cm}
\begin{tikzpicture}
\filldraw (0,0) circle [radius=0.08];
\filldraw (\1,0) circle [radius=0.08];
\filldraw (\1,\1) circle [radius=0.08];
\filldraw (0,\1) circle [radius=0.08];
\draw [thick] (\1,\1) -- (0,\1) -- (0,0) -- (\1,0);
\end{tikzpicture}
\end{minipage}
\begin{minipage}{1cm}
\begin{tikzpicture}
\filldraw (0,0) circle [radius=0.08];
\filldraw (\1,0) circle [radius=0.08];
\filldraw (\1,\1) circle [radius=0.08];
\filldraw (0,\1) circle [radius=0.08];
\draw [thick] (\1,\1) -- (0,\1) -- (\1,0) -- (0,0);
\end{tikzpicture}
\end{minipage} \\[5pt]
\begin{minipage}{1.5cm}
\begin{tikzpicture}
\filldraw (0,0) circle [radius=0.08];
\filldraw (\1,0) circle [radius=0.08];
\filldraw (\1,\1) circle [radius=0.08];
\filldraw (0,\1) circle [radius=0.08];
\draw [thick] (0,0) -- (\1,\1) -- (0,\1) -- (\1,0);
\end{tikzpicture}
\end{minipage}
\begin{minipage}{1.5cm}
\begin{tikzpicture}
\filldraw (0,0) circle [radius=0.08];
\filldraw (\1,0) circle [radius=0.08];
\filldraw (\1,\1) circle [radius=0.08];
\filldraw (0,\1) circle [radius=0.08];
\draw [thick] (\1,\1) -- (0,0) (\1,\1) -- (0,\1) (\1,\1) -- (\1,0);
\end{tikzpicture}
\end{minipage}
\begin{minipage}{1.5cm}
\begin{tikzpicture}
\filldraw (0,0) circle [radius=0.08];
\filldraw (\1,0) circle [radius=0.08];
\filldraw (\1,\1) circle [radius=0.08];
\filldraw (0,\1) circle [radius=0.08];
\draw [thick] (0,\1) -- (\1,\1) -- (0,0) -- (\1,0);
\end{tikzpicture}
\end{minipage}
\begin{minipage}{1cm}
\begin{tikzpicture}
\filldraw (0,0) circle [radius=0.08];
\filldraw (\1,0) circle [radius=0.08];
\filldraw (\1,\1) circle [radius=0.08];
\filldraw (0,\1) circle [radius=0.08];
\draw [thick] (0,\1) -- (\1,\1) -- (\1,0) --(0,0);
\end{tikzpicture}
\end{minipage} \\[5pt]
\begin{minipage}{1.5cm}
\begin{tikzpicture}
\filldraw (0,0) circle [radius=0.08];
\filldraw (\1,0) circle [radius=0.08];
\filldraw (\1,\1) circle [radius=0.08];
\filldraw (0,\1) circle [radius=0.08];
\draw [thick] (\1,\1) -- (0,0) -- (0,\1) --(\1,0);
\end{tikzpicture}
\end{minipage}
\begin{minipage}{1.5cm}
\begin{tikzpicture}
\filldraw (0,0) circle [radius=0.08];
\filldraw (\1,0) circle [radius=0.08];
\filldraw (\1,\1) circle [radius=0.08];
\filldraw (0,\1) circle [radius=0.08];
\draw [thick] (0,\1) -- (0,0) -- (\1,\1) -- (\1,0);
\end{tikzpicture}
\end{minipage}
\begin{minipage}{1.5cm}
\begin{tikzpicture}
\filldraw (0,0) circle [radius=0.08];
\filldraw (\1,0) circle [radius=0.08];
\filldraw (\1,\1) circle [radius=0.08];
\filldraw (0,\1) circle [radius=0.08];
\draw [thick] (0,0) -- (0,\1) (0,0) -- (\1,0) (0,0) -- (\1,\1);
\end{tikzpicture}
\end{minipage}
\begin{minipage}{1cm}
\begin{tikzpicture}
\filldraw (0,0) circle [radius=0.08];
\filldraw (\1,0) circle [radius=0.08];
\filldraw (\1,\1) circle [radius=0.08];
\filldraw (0,\1) circle [radius=0.08];
\draw [thick] (0,\1) -- (0,0) -- (\1,0) -- (\1,\1);
\end{tikzpicture}
\end{minipage} \\[5pt]
\begin{minipage}{1.5cm}
\begin{tikzpicture}
\filldraw (0,0) circle [radius=0.08];
\filldraw (\1,0) circle [radius=0.08];
\filldraw (\1,\1) circle [radius=0.08];
\filldraw (0,\1) circle [radius=0.08];
\draw [thick] (0,0) -- (0,\1) -- (\1,0) -- (\1,\1);
\end{tikzpicture}
\end{minipage}
\begin{minipage}{1.5cm}
\begin{tikzpicture}
\filldraw (0,0) circle [radius=0.08];
\filldraw (\1,0) circle [radius=0.08];
\filldraw (\1,\1) circle [radius=0.08];
\filldraw (0,\1) circle [radius=0.08];
\draw [thick] (0,\1) -- (\1,0) -- (\1,\1) -- (0,0);
\end{tikzpicture}
\end{minipage}
\begin{minipage}{1.5cm}
\begin{tikzpicture}
\filldraw (0,0) circle [radius=0.08];
\filldraw (\1,0) circle [radius=0.08];
\filldraw (\1,\1) circle [radius=0.08];
\filldraw (0,\1) circle [radius=0.08];
\draw [thick] (0,\1) -- (\1,0) -- (0,0) -- (\1,\1);
\end{tikzpicture}
\end{minipage}
\begin{minipage}{1cm}
\begin{tikzpicture}
\filldraw (0,0) circle [radius=0.08];
\filldraw (\1,0) circle [radius=0.08];
\filldraw (\1,\1) circle [radius=0.08];
\filldraw (0,\1) circle [radius=0.08];
\draw [thick] (\1,0) -- (0,0) (\1,0) -- (0,\1) (\1,0) -- (\1,\1);
\end{tikzpicture}
\end{minipage}
\end{eqnarray*}
\caption[Caption for LOF]%
      {All 16 trees for $N = 4$.\protect\footnotemark}
\label{fig: spanning trees}
\end{subfigure}
\hfill
\begin{subfigure}[t]{.48\textwidth}
\begin{eqnarray*}
\begin{minipage}{1.5cm}
\begin{tikzpicture}
\draw [thick, \red] (0,0) to [out=135,in=-135] (0,\1);
\draw [thick, \red] (\1,0) to [out=45,in=-45] (\1,\1);
\filldraw (0,0) circle [radius=0.08];
\filldraw (\1,0) circle [radius=0.08];
\filldraw (\1,\1) circle [radius=0.08];
\filldraw (0,\1) circle [radius=0.08];
\draw [thick] (0,\1) -- (0,0) (0,\1) -- (\1,0) (0,\1) -- (\1,\1);
\end{tikzpicture}
\end{minipage}
\begin{minipage}{1.5cm}
\begin{tikzpicture}
\draw [thick, \red] (0,0) to [out=135,in=-135] (0,\1);
\draw [thick, \red] (\1,0) to [out=45,in=-45] (\1,\1);
\filldraw (0,0) circle [radius=0.08];
\filldraw (\1,0) circle [radius=0.08];
\filldraw (\1,\1) circle [radius=0.08];
\filldraw (0,\1) circle [radius=0.08];
\draw [thick] (0,\1) -- (0,0) (0,\1) -- (\1,\1) -- (\1,0);
\end{tikzpicture}
\end{minipage}
\begin{minipage}{1.5cm}
\begin{tikzpicture}
\draw [thick, \red] (0,0) to [out=135,in=-135] (0,\1);
\draw [thick, \red] (\1,0) to [out=45,in=-45] (\1,\1);
\filldraw (0,0) circle [radius=0.08];
\filldraw (\1,0) circle [radius=0.08];
\filldraw (\1,\1) circle [radius=0.08];
\filldraw (0,\1) circle [radius=0.08];
\draw [thick] (\1,\1) -- (0,\1) -- (0,0) -- (\1,0);
\end{tikzpicture}
\end{minipage}
\begin{minipage}{1.3cm}
\begin{tikzpicture}
\draw [thick,\red] (0,0) to [out=135,in=-135] (0,\1);
\draw [thick, \red] (\1,0) to [out=45,in=-45] (\1,\1);
\filldraw (0,0) circle [radius=0.08];
\filldraw (\1,0) circle [radius=0.08];
\filldraw (\1,\1) circle [radius=0.08];
\filldraw (0,\1) circle [radius=0.08];
\draw [thick] (\1,\1) -- (0,\1) -- (\1,0) -- (0,0);
\end{tikzpicture}
\end{minipage} \\[5pt]
\begin{minipage}{1.5cm}
\begin{tikzpicture}
\draw [thick, \red] (0,0) to [out=135,in=-135] (0,\1);
\draw [thick, \red] (\1,0) to [out=45,in=-45] (\1,\1);
\filldraw (0,0) circle [radius=0.08];
\filldraw (\1,0) circle [radius=0.08];
\filldraw (\1,\1) circle [radius=0.08];
\filldraw (0,\1) circle [radius=0.08];
\draw [thick] (0,0) -- (\1,\1) -- (0,\1) -- (\1,0);
\end{tikzpicture}
\end{minipage}
\begin{minipage}{1.5cm}
\begin{tikzpicture}
\draw [thick, \red] (0,0) to [out=135,in=-135] (0,\1);
\draw [thick, \red] (\1,0) to [out=45,in=-45] (\1,\1);
\filldraw (0,0) circle [radius=0.08];
\filldraw (\1,0) circle [radius=0.08];
\filldraw (\1,\1) circle [radius=0.08];
\filldraw (0,\1) circle [radius=0.08];
\draw [thick] (\1,\1) -- (0,0) (\1,\1) -- (0,\1) (\1,\1) -- (\1,0);
\end{tikzpicture}
\end{minipage}
\begin{minipage}{1.5cm}
\begin{tikzpicture}
\draw [thick, \red] (0,0) to [out=135,in=-135] (0,\1);
\draw [thick, \red] (\1,0) to [out=45,in=-45] (\1,\1);
\filldraw (0,0) circle [radius=0.08];
\filldraw (\1,0) circle [radius=0.08];
\filldraw (\1,\1) circle [radius=0.08];
\filldraw (0,\1) circle [radius=0.08];
\draw [thick] (0,\1) -- (\1,\1) -- (0,0) -- (\1,0);
\end{tikzpicture}
\end{minipage}
\begin{minipage}{1.3cm}
\begin{tikzpicture}
\draw [thick, \red] (0,0) to [out=135,in=-135] (0,\1);
\draw [thick, \red] (\1,0) to [out=45,in=-45] (\1,\1);
\filldraw (0,0) circle [radius=0.08];
\filldraw (\1,0) circle [radius=0.08];
\filldraw (\1,\1) circle [radius=0.08];
\filldraw (0,\1) circle [radius=0.08];
\draw [thick] (0,\1) -- (\1,\1) -- (\1,0) --(0,0);
\end{tikzpicture}
\end{minipage} \\[5pt]
\begin{minipage}{1.5cm}
\begin{tikzpicture}
\draw [thick, \red] (0,0) to [out=135,in=-135] (0,\1);
\draw [thick, \red] (\1,0) to [out=45,in=-45] (\1,\1);
\filldraw (0,0) circle [radius=0.08];
\filldraw (\1,0) circle [radius=0.08];
\filldraw (\1,\1) circle [radius=0.08];
\filldraw (0,\1) circle [radius=0.08];
\draw [thick] (\1,\1) -- (0,0) -- (0,\1) --(\1,0);
\end{tikzpicture}
\end{minipage}
\begin{minipage}{1.5cm}
\begin{tikzpicture}
\draw [thick, \red] (0,0) to [out=135,in=-135] (0,\1);
\draw [thick, \red] (\1,0) to [out=45,in=-45] (\1,\1);
\filldraw (0,0) circle [radius=0.08];
\filldraw (\1,0) circle [radius=0.08];
\filldraw (\1,\1) circle [radius=0.08];
\filldraw (0,\1) circle [radius=0.08];
\draw [thick] (0,\1) -- (0,0) -- (\1,\1) -- (\1,0);
\end{tikzpicture}
\end{minipage}
\begin{minipage}{1.5cm}
\begin{tikzpicture}
\draw [thick, \red] (0,0) to [out=135,in=-135] (0,\1);
\draw [thick, \red] (\1,0) to [out=45,in=-45] (\1,\1);
\filldraw (0,0) circle [radius=0.08];
\filldraw (\1,0) circle [radius=0.08];
\filldraw (\1,\1) circle [radius=0.08];
\filldraw (0,\1) circle [radius=0.08];
\draw [thick] (0,0) -- (0,\1) (0,0) -- (\1,0) (0,0) -- (\1,\1);
\end{tikzpicture}
\end{minipage}
\begin{minipage}{1.3cm}
\begin{tikzpicture}
\filldraw (0,0) circle [radius=0.08];
\filldraw (\1,0) circle [radius=0.08];
\filldraw (\1,\1) circle [radius=0.08];
\filldraw (0,\1) circle [radius=0.08];
\draw [thick] (0,\1) -- (0,0) -- (\1,0) -- (\1,\1);
\draw [thick, \red] (0,0) to [out=135,in=-135] (0,\1);
\draw [thick, \red] (\1,0) to [out=45,in=-45] (\1,\1);
\end{tikzpicture}
\end{minipage} \\[5pt]
\begin{minipage}{1.5cm}
\begin{tikzpicture}
\draw [thick, \red] (0,0) to [out=135,in=-135] (0,\1);
\draw [thick, \red] (\1,0) to [out=45,in=-45] (\1,\1);
\filldraw (0,0) circle [radius=0.08];
\filldraw (\1,0) circle [radius=0.08];
\filldraw (\1,\1) circle [radius=0.08];
\filldraw (0,\1) circle [radius=0.08];
\draw [thick] (0,0) -- (0,\1) -- (\1,0) -- (\1,\1);
\end{tikzpicture}
\end{minipage}
\begin{minipage}{1.5cm}
\begin{tikzpicture}
\draw [thick, \red] (0,0) to [out=135,in=-135] (0,\1);
\draw [thick, \red] (\1,0) to [out=45,in=-45] (\1,\1);
\filldraw (0,0) circle [radius=0.08];
\filldraw (\1,0) circle [radius=0.08];
\filldraw (\1,\1) circle [radius=0.08];
\filldraw (0,\1) circle [radius=0.08];
\draw [thick] (0,\1) -- (\1,0) -- (\1,\1) -- (0,0);
\end{tikzpicture}
\end{minipage}
\begin{minipage}{1.5cm}
\begin{tikzpicture}
\draw [thick, \red] (0,0) to [out=135,in=-135] (0,\1);
\draw [thick, \red] (\1,0) to [out=45,in=-45] (\1,\1);
\filldraw (0,0) circle [radius=0.08];
\filldraw (\1,0) circle [radius=0.08];
\filldraw (\1,\1) circle [radius=0.08];
\filldraw (0,\1) circle [radius=0.08];
\draw [thick] (0,\1) -- (\1,0) -- (0,0) -- (\1,\1);
\end{tikzpicture}
\end{minipage}
\begin{minipage}{1.3cm}
\begin{tikzpicture}
\draw [thick, \red] (0,0) to [out=135,in=-135] (0,\1);
\draw [thick, \red] (\1,0) to [out=45,in=-45] (\1,\1);
\filldraw (0,0) circle [radius=0.08];
\filldraw (\1,0) circle [radius=0.08];
\filldraw (\1,\1) circle [radius=0.08];
\filldraw (0,\1) circle [radius=0.08];
\draw [thick] (\1,0) -- (0,0) (\1,0) -- (0,\1) (\1,0) -- (\1,\1);
\end{tikzpicture}
\end{minipage}
\end{eqnarray*}
\caption{Superposition of \eqref{eq: 0_invariant} and the 16 trees.}
\label{fig: superposition_P=2}
\end{subfigure}
\caption{The most relevant 2-invariant for $N = 4$ from superposition of graphs.}
\end{figure}
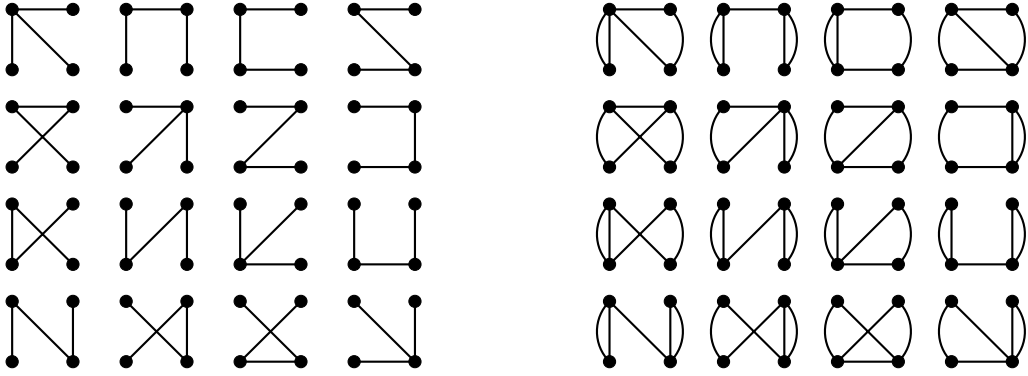

\noindent \textbf{\textit{\nv}} = \textbf{4 Case}: In \S\ref{sec: 245} we found that \eqref{eq: 245} gives the only independent minimal 2-invariant for $\nv =4$.  It has the structure of a superposition of the sum with unit coefficients of all $\nv =4$ trees (Figure \ref{fig: spanning trees}) and an exact 0-invariant:
\begin{equation} \label{eq: 0_invariant}
\begin{tikzpicture}
\draw [thick, \red] (0,0) to [out=135,in=-135] (0,\1);
\draw [thick, \red] (\1,0) to [out=45,in=-45] (\1,\1);
\filldraw (0,0) circle [radius=0.08];
\filldraw (\1,0) circle [radius=0.08];
\filldraw (\1,\1) circle [radius=0.08];
\filldraw (0,\1) circle [radius=0.08];
\end{tikzpicture}
\end{equation}

\noindent The superposition of this 0-invariant with the trees in Figure \ref{fig: spanning trees} is given in Figure \ref{fig: superposition_P=2}.  The sum of all graphs in Figure \ref{fig: superposition_P=2} with unit coefficients gives the 2-invariant in \eqref{eq: 245} (with an overall prefactor of 4).

\footnotetext{The trees are arranged in order of their Pr\"{u}fer sequences \cite{graph_theory}.}

\subsubsection{Cubic Shift (\textit{\pd}=3)}

\noindent \textbf{\textit{\nv}} = \textbf{3 Case}: For $\pd = 3$, the only independent minimal invariant for $\nv =3$ is given in \eqref{eq: n=3 graph}, which can be written as a superposition of two equal-weight tree summations:
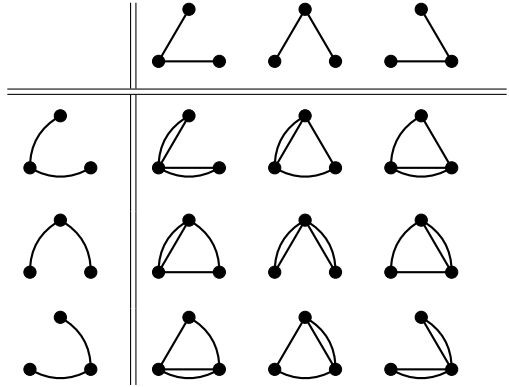
\begin{center}
\begin{tabular}{l || cc}
&
\begin{minipage}{\mini cm}
\begin{tikzpicture}
\draw [white] (0,0) -- (0,-0.28);
\draw [thick] (0.4,0.69) -- (0,0) -- (\1,0);
\filldraw (0,0) circle [radius=0.08];
\filldraw (\1,0) circle [radius=0.08];
\filldraw (0.4,0.69) circle [radius=0.08];
\end{tikzpicture}
\end{minipage}
\begin{minipage}{\mini cm}
\begin{tikzpicture}
\draw [white] (0,0) -- (0,-0.28);
\draw [thick] (0,0) -- (0.4,0.69) -- (\1,0);
\filldraw (0,0) circle [radius=0.08];
\filldraw (\1,0) circle [radius=0.08];
\filldraw (0.4,0.69) circle [radius=0.08];
\end{tikzpicture}
\end{minipage}
\begin{minipage}{\mini cm}
\begin{tikzpicture}
\draw [white] (0,0) -- (0,-0.28);
\draw [thick] (0.4,0.69) -- (\1,0) -- (0,0);
\filldraw (0.4,0.69) circle [radius=0.08];
\filldraw (0,0) circle [radius=0.08];
\filldraw (\1,0) circle [radius=0.08];
\end{tikzpicture}
\end{minipage} \\
\end{tabular}
\end{center}

\noindent As we already pointed out, this 3-invariant happens to be the minimal 2-invariant as well.

We can produce more 3-invariants by superposing an exact 1-invariant on the equal-weight sum of $\nv =3$ trees.  This would have two more derivatives compared to the minimal term above.  For example, there are two independent exact 1-invariants for $\nv =3$ with 3 edges:
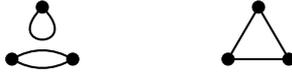
\begin{equation*}
\begin{minipage}{\mini cm}
%
\end{minipage}
\end{equation*}

\noindent which yield the following two 3-invariants:
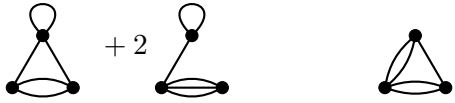
\begin{equation*}
\begin{minipage}{\mini cm}
%
\end{minipage}
\end{equation*}

\vspace{0.3cm}

\noindent \textbf{\textit{\nv}} = \textbf{4 Case}: In \S\ref{sec: 346} we found that \eqref{eq: 346} gives the only independent minimal 3-invariant for $\nv =4$.  It has the structure of a superposition of two sums with unit coefficients of all trees in Figure \ref{fig: spanning trees}.  As mentioned in \S\ref{sec: Galileons}, there are two isomorphism classes of $\nv =4$ trees:
\begin{equation} \label{eq: n=4 star and path}
\Tred =
\begin{minipage}{\mini cm}
%
\end{minipage} 
\end{eqnarray*}
\caption{Superposition of $\Tred$ and Figure \ref{fig: spanning trees}.}
\label{fig: (4,6)a}
\end{subfigure}
\hfill
\begin{subfigure}[t]{.48\textwidth}
\begin{eqnarray*}
\begin{minipage}{1.5cm}
%
\end{minipage}
\end{eqnarray*}
\caption{Superposition of $\Tblue$ and Figure \ref{fig: spanning trees}.}
\label{fig: (4,6)b}
\end{subfigure}
\caption{Superposition of graphs in \eqref{eq: n=4 star and path} on trees in Figure \ref{fig: spanning trees}.}
\label{fig: (4,6)}
\end{figure}

\noindent Superposing these graphs on the $\nv =4$ trees produces the graphs in Figure \ref{fig: (4,6)}.  If $T$ and $T'$ are isomorphic trees, then superposing $T$ on the trees in Figure \ref{fig: spanning trees} produces 16 graphs which are isomorphic to the 16 graphs formed by superposing $T'$ on the same trees.  There are four trees in the isomorphism class of $\Tred$ and twelve for $\Tblue$.  Therefore, we just have to give the 16 graphs in Figure \ref{fig: (4,6)a} weight 4 and the 16 graphs in Figure \ref{fig: (4,6)b} weight 12 and then add them all up.  The result is \eqref{eq: 346} with an overall prefactor of 4.  Again, we have already shown that this is the unique minimal 3-invariant for $\nv =4$.

As in the $\nv =3$ case, we can produce non-minimal 3-invariants by superposing an exact 1-invariant on the equal-weight sum of $\nv =4$ trees.  For example, there are four independent exact 1-invariants for $\nv =4$ with the lowest number of edges:
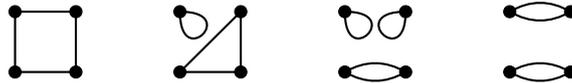
\begin{equation*}
\begin{minipage}{\mini cm}
\begin{tikzpicture}
\draw [thick, white] (0,\1) to [out=30,in=150] (\1,\1);
\draw [thick, white] (0,0) to [out=-30,in=-150] (\1,0);
\draw [thick] (0,0) -- (\1,0) -- (\1,\1) -- (0,\1) -- (0,0);
\filldraw (0,0) circle [radius=0.08];
\filldraw (\1,0) circle [radius=0.08];
\filldraw (\1,\1) circle [radius=0.08];
\filldraw (0,\1) circle [radius=0.08];
\end{tikzpicture}
\end{minipage}
\qquad
\begin{minipage}{\mini cm}
\begin{tikzpicture}
\draw [thick, white] (0,\1) to [out=30,in=150] (\1,\1);
\draw [thick, white] (0,0) to [out=-30,in=-150] (\1,0);
\draw [thick] (0,0) -- (\1,0) -- (\1,\1) -- (0,0);
\draw [thick] (0,\1) .. controls (0,0) and (\1,\1) .. (0,\1);
\filldraw (0,0) circle [radius=0.08];
\filldraw (\1,0) circle [radius=0.08];
\filldraw (\1,\1) circle [radius=0.08];
\filldraw (0,\1) circle [radius=0.08];
\end{tikzpicture}
\end{minipage}
\qquad
\begin{minipage}{\mini cm}
\begin{tikzpicture}
\draw [thick, white] (0,\1) to [out=30,in=150] (\1,\1);
\draw [thick] (0,0) to [out=30,in=150] (\1,0);
\draw [thick] (0,0) to [out=-30,in=-150] (\1,0);
\draw [thick] (0,\1) .. controls (0,0) and (\1,\1) .. (0,\1);
\draw [thick] (\1,\1) .. controls (0,\1) and (\1,0) .. (\1,\1);
\filldraw (0,0) circle [radius=0.08];
\filldraw (\1,0) circle [radius=0.08];
\filldraw (\1,\1) circle [radius=0.08];
\filldraw (0,\1) circle [radius=0.08];
\end{tikzpicture}
\end{minipage}
\qquad
\begin{minipage}{\mini cm}
\begin{tikzpicture}
\draw [thick] (0,0) to [out=30,in=150] (\1,0);
\draw [thick] (0,0) to [out=-30,in=-150] (\1,0);
\draw [thick] (0,\1) to [out=30,in=150] (\1,\1);
\draw [thick] (0,\1) to [out=-30,in=-150] (\1,\1);
\filldraw (0,0) circle [radius=0.08];
\filldraw (\1,0) circle [radius=0.08];
\filldraw (\1,\1) circle [radius=0.08];
\filldraw (0,\1) circle [radius=0.08];
\end{tikzpicture}
\end{minipage}
\end{equation*}

\subsubsection{Quartic Shift (\textit{\pd}=4)}

\noindent \textbf{\textit{\nv} = 3 Case}: As argued earlier, $\Delta_{\text{min}} = 6$ in this case.  There are two Medusas with the fewest edges such that the $\times$-vertex and the $\bullet$-vertices have degree no less than 3 (note that $\frac{1}{2} ( \pd +1) = \frac{5}{2}$ in this case):
\vspace{-3mm}
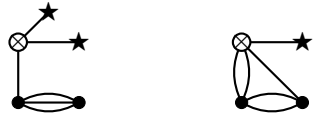
\begin{align*}
\begin{minipage}{\mini cm}
\begin{tikzpicture}
\filldraw (0,0) circle [radius=0.08];
\filldraw (\1,0) circle [radius=0.08];
\draw [thick] (0,0) to [out=30,in=150] (\1,0);
\draw [thick] (0,0) to [out=-30,in=-150] (\1,0);
\draw [thick] (\1,\1) -- (0,\1) -- (0,0) -- (\1,0);
\draw [thick] (0,\1) -- (0.4,1.2);
\node at (0.4,1.2) {\scalebox{0.8}{$\bigstar$}};
\node at (\1,\1) {\scalebox{0.8}{$\bigstar$}};
\filldraw [white] (0,\1) circle [radius=0.115];
\draw [thick] (0,\1) circle [radius=0.115];
\node at (0,\1) {$\times$};
\end{tikzpicture}
\end{minipage}
\qquad\qquad
\begin{minipage}{\mini cm}
\begin{tikzpicture}
\draw [white, thick] (0,\1) -- (0.4,1.2);
\node [white] at (0.4,1.2) {\scalebox{0.8}{$\bigstar$}};
\draw [thick] (0,0) to [out=115,in=-115] (0,\1);
\draw [thick] (0,0) to [out=65,in=-65] (0,\1);
\filldraw (0,0) circle [radius=0.08];
\filldraw (\1,0) circle [radius=0.08];
\draw [thick] (0,0) to [out=30,in=150] (\1,0);
\draw [thick] (0,0) to [out=-30,in=-150] (\1,0);
\draw [thick] (\1,0) -- (0,\1) -- (\1,\1);
\node at (\1,\1) {\scalebox{0.8}{$\bigstar$}};
\filldraw [white] (0,\1) circle [radius=0.115];
\draw [thick] (0,\1) circle [radius=0.115];
\node at (0,\1) {$\times$};
\end{tikzpicture}
\end{minipage}
\end{align*}

\vspace{1mm}

\noindent There is exactly one minimal 4-invariant in this case, which is constructed by superposing two equal-weight sums of trees with an exact 0-invariant:
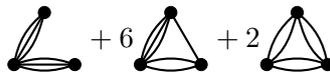
\begin{align*}
\hspace{0.5cm}
\begin{minipage}{\mini cm}
\begin{tikzpicture}
\draw [thick] (0,0) to [out=80,in=-140] (0.4,0.69);
\draw [thick] (0,0) to [out=35,in=-105] (0.4,0.69);
\filldraw (0,0) circle [radius=0.08];
\filldraw (\1,0) circle [radius=0.08];
\filldraw (0.4,0.69) circle [radius=0.08];
\draw [thick] (0,0) to [out=20,in=160] (\1,0);
\draw [thick] (0,0) to [out=-20,in=-160] (\1,0);
\draw [thick] (\1,0) -- (0,0) -- (0.4,0.69);
\end{tikzpicture}
\end{minipage}
\hspace{-0.4cm} + 6
\begin{minipage}{\mini cm}
\begin{tikzpicture}
\draw [thick] (0,0) to [out=80,in=-140] (0.4,0.69);
\draw [thick] (0,0) to [out=35,in=-105] (0.4,0.69);
\filldraw (0,0) circle [radius=0.08];
\filldraw (\1,0) circle [radius=0.08];
\filldraw (0.4,0.69) circle [radius=0.08];
\draw [thick] (0,0) to [out=20,in=160] (\1,0);
\draw [thick] (0,0) to [out=-20,in=-160] (\1,0);
\draw [thick] (0,0) -- (0.4,0.69) -- (\1,0);
\end{tikzpicture}
\end{minipage}
\hspace{-0.4cm} + 2
\begin{minipage}{\mini cm}
\begin{tikzpicture}
\draw [thick] (0,0) to [out=80,in=-140] (0.4,0.69);
\draw [thick] (0,0) to [out=35,in=-105] (0.4,0.69);
\filldraw (0,0) circle [radius=0.08];
\filldraw (\1,0) circle [radius=0.08];
\filldraw (0.4,0.69) circle [radius=0.08];
\draw [thick] (0,0) to [out=20,in=160] (\1,0);
\draw [thick] (0,0) to [out=-20,in=-160] (\1,0);
\draw [thick] (\1,0) to [out=100,in=-40] (0.4,0.69);
\draw [thick] (\1,0) to [out=145,in=-75] (0.4,0.69);
\end{tikzpicture}
\end{minipage}
\end{align*}

\noindent Note that the sum of the coefficients is 9, as it should be, since there are three $\nv =3$ trees, and thus there are nine superpositions of two $\nv =3$ trees.

Examples of non-minimal invariants can be constructed by superposing an equal-weight sum of trees with an exact 2-invariant, or two equal-weight sums of trees with an exact 1-invariant.  

\vspace{0.3cm}

	The proofs of uniqueness and minimality for the remaining $\nv =4$ examples are lengthy and involve many more Medusas than the previous examples, but the process is the same.  Therefore, we will simply write the invariants and state that they are unique and minimal.

\vspace{0.3cm}

\noindent \textbf{\textit{\nv}=4 Case}: We construct the minimal 4-invariant by superposing two copies of equal-weight sums of trees with an exact 0-invariant.  There is one independent exact 0-invariant:
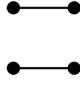
\begin{equation*}
\begin{minipage}{\mini cm}

\end{minipage}
\end{equation*}

\noindent Superposing this on the superposition of two copies of equal-weight sums of trees yields
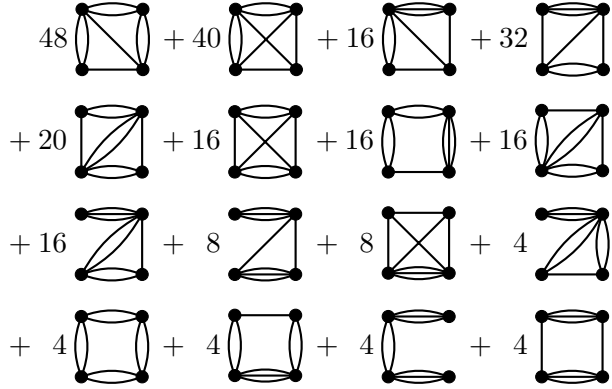
\begin{align*}
& \hspace{4.8mm}
48 \,
\begin{minipage}{\mini cm}
%
\end{minipage}
\end{align*}

\vspace{2mm}

\noindent Note that the sum of the coefficients is $256 = 16^2$.

\subsubsection{Quintic Shift (\textit{\pd}=5)}

\noindent \textbf{\textit{\nv}} = \textbf{3 Case}: In this case, $\Delta_{\text{min}} = 6$.  There are three Medusas with the fewest edges such that the $\times$-vertex and the $\bullet$-vertices have degree no less than 3 (note that $\frac{1}{2} ( \pd +1) = 3$):
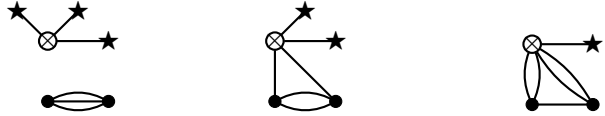
\vspace{-2.8mm}
\begin{align*}
\hspace{1.6cm}
\begin{minipage}{3cm}
%
\end{minipage}
\end{align*}

\vspace{1mm}

\noindent There is exactly one minimal 5-invariant in this case, which is constructed by superposing three equal-weight sums of trees: 
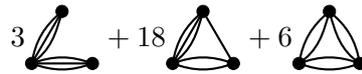
\begin{align*}
\hspace{0.2cm}
3 
\begin{minipage}{\mini cm}
%
\end{minipage}
\end{align*}

\noindent Note that the sum of the coefficients is $27 = 3^3$.  Also, note that this is proportional to the unique independent minimal 4-invariant found in the previous section, which was the superposition of two equal-weight sums of trees and an exact 0-invariant.

\vspace{0.3cm}

\noindent \textbf{\textit{\nv}} = \textbf{4 Case}: The superposition of three equal-weight sums of $\nv =4$ trees yields
\begin{align*}
4 \, &
\begin{minipage}{\mini cm}
%
\end{minipage} 
\end{align*}

\vspace{2mm}

\noindent Note that the sum of the coefficients is $4096 = 16^3$.

	To facilitate the check of the invariance of the above linear combination, denoted as $L$, we provide the linear combination of Medusas $L_M$ such that $\delta_5 (L) = \rho^{(0)} (L_M)$:
\vspace{-0.1cm}
\begin{align*}
\hspace{0.4cm}
12 \hspace{-0.4cm}
& 
\begin{minipage}{2cm}
%
\end{minipage}
\end{align*}

\vspace{-0.7cm}

\section{Conclusions and Outlook} \label{sec: conclusion}

In this paper, we studied nonrelativistic scalar field theories with polynomial 
shift symmetries.  In the free-field limit, such field theories arise in the 
context of Goldstone's theorem, where they lead to the hierarchies of possible 
universality classes of Nambu-Goldstone modes, as reviewed in \S{\ref{sec: mng}}.  Our main 
focus in \S{\ref{sec: Galileons}} and \S{\ref{sec: beyond}} has been on {\it interacting\/} effective field 
theories which respect the polynomial shift symmetries of degree 
$\pd =1,2,\ldots$.  In order to 
find such theories, one needs to identify possible Lagrangian terms invariant 
under the polynomial shift up to total derivatives, and organize them by 
their scaling dimension, starting from the most relevant.  As we showed in 
\S{\ref{sec: Galileons}}, \S{\ref{sec: beyond}} and Appendix~\ref{AppB}, this essentially cohomological classification 
problem can be usefully translated into the language of graph theory.  This 
graphical technique is important for two reasons.  First, it is quite 
powerful:  The translation of 
the classification problem into a graph-theory problem allows us to generate 
sequences of invariants for various values of $\pd$, number $\nv$ of fields, 
the number $2 \Delta$ of spatial derivatives, and as a function of the spatial 
dimension $\sdim$, in a way that is much more efficient than any ``brute 
force'' 
technique.  Secondly, and perhaps more importantly, the graphical technique 
reveals some previously hidden structure even in those invariants 
already known in the literature.  For example, the known Galileon $\nv$-point 
invariants are given by the equal-weight sums of all labeled trees 
with $\nv$ vertices!  This hidden simplicity of the Galileon invariants is 
a feature previously unsuspected in the literature, and its mathematical 
explanation deserves further study.  In addition, we also discovered patterns 
that allow the construction of higher polynomials from the superposition of 
graphs representing a collection of invariants of a lower degree -- again a
surprising result, revealing glimpses of intriguing connections among the 
{\it a priori\/} unrelated spaces of invariants across the various values of 
$P$, $N$ and $\Delta$.  

Throughout this paper, we focused for simplicity on the unrestricted 
polynomial shift symmetries of degree $\pd$, whose coefficients 
$a_{i_1\ldots i_\ell}$ are general real symmetric tensors of rank 
$\ell=0,\ldots , \pd$.  As we pointed out in \S\ref{sec:pss}, at $\pd\geq 2$, this 
maximal polynomial shift symmetry algebra allows various subalgebras, obtained 
by imposing additional conditions on the structure of $a_{i_1\ldots i_\ell}$'s.  
While this refinement does not significantly impact the classification of 
Gaussian fixed points, reducing the symmetry to one of the subalgebras inside 
the maximal polynomial shift symmetry can lead to new $N$-point 
invariants, beyond the ones presented in this paper.  It is possible 
to extend our graphical technique to the various reduced polynomial shift 
symmetries, and to study the refinement of the structure of polynomial shift 
invariants associated with the reduced symmetries.

Our main motivation for the study of scalar field theories with polynomial 
shift symmetries has originated from our desire to map out phenomena in which 
technical naturalness plays a crucial role, in general classes of field 
theories with or without relativistic symmetries.  The refined classification 
of the universality classes of NG modes and the nonrelativistic refinement of 
Goldstone's theorem have provided an example of scenarios where our naive 
relativistic intuition about technical naturalness may be misleading, and new 
interesting phenomena can emerge.  We anticipate that other surprises of 
naturalness are still hidden not only in the landscape of quantum field 
theories, but also in the landscape of nonrelativistic theories of quantum 
gravity.   

\acknowledgments

We wish to thank Hitoshi Murayama and Haruki Watanabe for useful discussions.  This work has been supported by NSF Grant PHY-1214644 
and by Berkeley Center for Theoretical Physics.  

\appendix

\section{Glossary of Graph Theory} \label{graph theory} 

In this section, we list the standard terminologies in graph theory to which we will refer.  (These essentially coincide with the ones in \cite{graph_theory}.) 
\begin{description}

\item [Graph] 

	A \emph{graph} $\Gamma$ is an ordered pair $(V(\Gamma), E(\Gamma))$ consisting of a set  $V(\Gamma)$ of \emph{vertices} and a set $E(\Gamma)$, disjoint from $V(\Gamma)$, of \emph{edges}, together with an \emph{incident function} $\Psi_\Gamma$ that associates with each edge of $\Gamma$ an ordered pair of (not necessarily distinct) vertices of $\Gamma$.  If $e$ is an edge and $u$ and $v$ are vertices such that $\Psi_\Gamma ( e ) = \{ u, v \}$, then $e$ is said to \emph{join} $u$ and $v$.  

\item [Isomorphism]
	
	Two graphs $\Gamma_A$ and $\Gamma_B$ are \emph{isomorphic} if there exist a pair of bijections $f: V(\Gamma_A) \rightarrow V(\Gamma_B)$ and $\phi : E( \Gamma_A ) \rightarrow E ( \Gamma_B )$ such that $\Psi_{\Gamma_A} (e) = \{ u, v \}$ if and only if $\Psi_{\Gamma_B} ( \phi ( e ) ) = \{ f ( u ), f( v ) \}$.
		
\item [Identical Graphs] Two graphs are \emph{identical}, written $\Gamma_A = \Gamma_B$, if $V(G) = V(H)$, $E(G) = E(H)$ and $\Psi_G = \Psi_H$.

\item [Labeled Graph]\hspace{-0.2cm} A graph in which the vertices are labeled but the edges are not, is called a \emph{labeled graph}.  This will be the notion of graphs that we will refer to most frequently.
	
\item [Unlabeled Graph]

	An \emph{unlabeled graph} is a representative of an equivalence class of isomorphic graphs.

\item [Finite Graph]

	A graph is \emph{finite} if both of its vertex set and edge set are finite.

\item [Null Graph] The graph with no vertices (and hence no edges) is the \emph{null graph}.
	
\item [Incident]

	The ends of an edge are said to be \emph{incident} to the edge, and \emph{vice versa}.
	
\item [Adjacent]

	Two vertices which are incident to a common edge are \emph{adjacent}.

\item [Loop]

	A \emph{loop} is an edge that joins a vertex to itself.  

\item [Cycle]

	A \emph{cycle} on two or more vertices is a graph in which the vertices can be arranged in a cyclic sequence such that two vertices are joined by exactly one edge if they are consecutive in the sequence, and are nonadjacent otherwise.  A \emph{cycle} on one vertex is a graph consisting of a single vertex with a loop.  
	
\item [Loopless Graph]

	A \emph{loopless} graph contains no loops.  Note that a loopless graph may still contain cycles on two or more vertices.
	
\item [Vertex Degree]

	The \emph{degree} of a vertex $v$, denoted by deg$(v)$, in a graph $\Gamma$ is the number of edges of $\Gamma$ incident to $v$, with each loop counting as two edges.  
	
\item [Empty Vertex]

	A vertex of degree 0 is called an \emph{empty vertex}.  
	
\item [Leaf] 

	A vertex of degree 1 is called a \emph{leaf}.

\item [Edge Deletion]

	The \emph{edge deletion} of an edge $e$ in a graph $\Gamma$ is defined by deleting from $\Gamma$ the edge $e$ but leaving the vertices and the remaining edges intact.
	
\item [Vertex Deletion] 

	The \emph{vertex deletion} of a vertex $v$ in a graph $\Gamma$ is defined by deleting from $\Gamma$ the vertex $v$ together with all the edges incident to $v$. The resulting graph is denoted by $\Gamma - v$.
	
\item [Connected Graph] 

	A graph is \emph{connected} if, for every partition of its vertex set into two nonempty sets $X$ and $Y$, there is an edge with one end in $X$ and one end in $Y$.

\item [Connected Component] 

	A \emph{connected component} of a graph $\Gamma$ is a connected subgraph $\Gamma^\prime$ of $\Gamma$ such that any vertex $v$ in $\Gamma^\prime$ satisfies the following condition: all edges incident to $v$ in $\Gamma$ are also contained in $\Gamma^\prime$.

\item [Tree]
	
	A \emph{tree} is a connected graph that contains no cycles.  In particular, note that a tree has no empty vertices if it contains more than one vertex.

\item [Cayley's Formula]
 	
	The number of labeled trees on $\nv$ vertices is $\nv^{\nv -2}$.

\end{description}


\section{Theorems and Proofs} \label{AppB} 

\subsection{The Graphical Representation}

Consider the polynomial shift symmetry applied to a real scalar field $\phi$,
\begin{align} \label{eq: polyshift}
	& \phi ( t , x^i ) \rightarrow \phi (t, x^i ) + \delta_\pd \phi, 
	& \delta_\pd \phi = a_{i_1 \cdots i_\pd} x^{i_1} \cdots x^{i_\pd} + \cdots + a_i x^i + a.  
\end{align}

\noindent The polynomial ends at $\pd^{\mathrm{th}}$ order in the spatial coordinate $x^i$ with $\pd = 0, 1, 2, \ldots$, respectively corresponding to constant shift, linear shift, quadratic shift, and so on.  The $a$'s are arbitrary real coefficients that parametrize the symmetry transformation, and are symmetric in any pair of indices.  In the algebraic language, for a specific $\pd$, we are searching for a Lagrangian that is invariant under the polynomial shift up to a total derivative.  Let $L$ be a term in the Lagrangian with $\nv$ $\phi$'s and $2 \Delta$ spatial derivatives.  Then,
\begin{equation} \label{eq: alg. consistency}
	\delta_P ( L ) = \partial_i (L_i),
\end{equation}

\noindent where $L_i$ is an expression containing $\nv -1$ $\phi$'s and an index $i$, which is not contracted.  Such $L$'s are called \emph{\pd -invariants}.  We will mainly focus on interaction terms, i.e., $\nv \geq 3$.
	
	We want to express these $\pd$-invariants using a graphical representation.  The ingredients of the graphical representation are:
\begin{enumerate}
	
\item
	
	 $\bullet$-vertices, denoted in a graph by \scalebox{1.2}{\textbullet}.
	 
\item
	 
	 $\times$-vertices, denoted in a graph by $\otimes$.
	 
\item

	$\star$-vertices, denoted in a graph by \scalebox{0.8}{$\bigstar$}.
	 
\item

	Edges, denoted in a graph by a line, that join the above vertices.
	
\end{enumerate}

\noindent In this context, a graph contains up to three types of vertices.  This means that these graphs carry an additional structure regarding vertex type, compared to the conventional definition of a graph in Appendix \ref{graph theory}.  

	We construct graphs using the following rules:
\begin{enumerate}

\item 

	The maximal degree of a $\times$-vertex is $\pd$.  Any graph containing a $\times$-vertex of degree greater than $\pd$ is identified with the null graph.
	
\item 

	There is at most one $\times$-vertex in a graph.

\item

	A $\star$-vertex is always a leaf (i.e., it has degree one).
	
\item 

	Two $\star$-vertices are not allowed to be adjacent to each other.		
		
\end{enumerate}

    We now describe what these graph ingredients represent.  A $\bullet$-vertex represents a $\phi$ and a $\times$-vertex represents $\delta_\pd \phi$.  A pair of derivatives with contracted indices, each one acting on a certain $\phi$ or $\delta_\pd \phi$, is represented by an edge joining the relevant $\bullet$- and $\times$-vertices.  Note that Rule 1, which requires that there be at most $\pd$ edges incident to the $\times$-vertex, is justified since $\pd +1$ derivatives acting on $\delta_\pd \phi$ gives zero.

	A graph with $\star$-vertices will represent terms which are total derivatives.  By Rules 3 and 4, a $\star$-vertex must always have exactly one edge incident to it, and this edge is incident to a $\bullet$-vertex or $\times$-vertex.  This edge represents a derivative acting on the entire term as a whole, and the index of that derivative is contracted with the index of another derivative acting on the $\phi$ or $\delta_\pd \phi$ of the $\bullet$- or $\times$-vertex, respectively, to which the $\star$-vertex is adjacent.  Therefore, any graph with a $\star$-vertex represents a total derivative term.  
	
	Since the Lagrangian terms that these graphs represent have a finite number of $\phi$'s and $\partial$'s, we will consider only finite graphs.  In addition, by the definition of graphs, all vertices and edges are automatically labeled, due to the fact that all elements in a set are distinct from each other.  Therefore, a graph represents an algebraic expression in which each $\phi$ and $\partial$ carries a label.  It will be convenient to keep the labels on $\phi$, but it is unnecessary to label the derivatives.  This motivates the definition given in Appendix \ref{graph theory} for ``labeled" graphs.  In the rest of Appendix \ref{AppB}, unless otherwise stated, a graph is understood to be a labeled graph.
	
	The desired algebraic expressions in which all $\phi$'s are identical can be recovered by identifying all isomorphic graphs (for examples, refer to \S\ref{sec: Galileons}).  In fact, the labeled $\pd$-invariants already capture all of the unlabeled ones (Appendix \ref{sec: unlabeled}).
	
	Note that not all algebraic expressions are captured in the graphical representation described above.  For example, $\partial^2 (\partial_j \phi \, \partial_j \phi)$ cannot be represented by a graph, since two $\star$-vertices are forbidden to be adjacent to each other by Rule 4.  However, this algebraic expression can be written as $2 \partial_i ( \partial_i \partial_j \phi \, \partial_j \phi)$, which is graphically represented in Figure \ref{fig: ex_a}, disregarding the coefficient 2.  Another peculiar example is $\partial_i (\partial_j \phi \, \partial_j \phi) \, \partial_i \partial^2 \phi$, which is equal to $4 (\partial_i \partial_j \phi) (\partial_j \phi) (\partial_i \partial^2 \phi)$.  Although the graphical representation for the former expression is beyond the current framework, the graphical representation for the latter one is given in Figure \ref{fig: ex_b}.  One could generalize the graphical representation to include all possible algebraic expressions.  However, for our purposes, the present framework will suffice.  

\begin{figure}[t!]
    \centering
    \begin{subfigure}[t]{4cm} 
    \hspace{0.5cm}
    \begin{minipage}{3cm}
    \begin{tikzpicture}
        \draw [white] (0,0) -- (-1,0);
        \draw [thick] (0,0) -- (\1,0) -- (0.4,0.69);
        \node at (0.4,0.69) {\scalebox{0.8}{$\bigstar$}};
        \filldraw (0,0) circle [radius=0.08];
        \filldraw (\1,0) circle [radius=0.08];
    \end{tikzpicture}
    \end{minipage}
    \caption{$\partial_i (\partial_j \phi \partial_i \partial_j \phi)$}
    \label{fig: ex_a}
    \end{subfigure}
    \qquad\qquad
    \begin{subfigure}[t]{4cm}
    \hspace{0.5cm}
    \begin{minipage}{3cm}
    \begin{tikzpicture}
        \draw [thick] (0,0) .. controls (0,\1) and (-\1,0) .. (0,0);
        \draw [thick] (0,0) -- (\1,0) -- (0.4,0.69);
        \filldraw (0.4,0.69) circle [radius=0.08];
        \filldraw (0,0) circle [radius=0.08];
        \filldraw (\1,0) circle [radius=0.08];
    \end{tikzpicture}
    \end{minipage}
    \caption{$\partial_j \phi \partial_i \partial_j \phi \partial_i \partial^2 \phi$}
    \label{fig: ex_b}
    \end{subfigure}
\caption{Examples for the graphical representation of algebraic expressions.}
\label{fig: ex}
\end{figure}
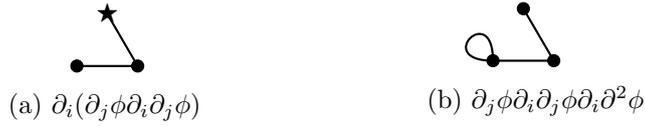

\subsubsection{Types of Graphs and Vector Spaces} 

	We classify graphs by different combinations of vertices:

\begin{Def} \quad
\begin{enumerate}

\item

A \emph{plain-graph} is a graph in which all vertices are $\bullet$'s.  

\item

A \emph{$\star$-ed plain-graph} is a graph with vertex set consisting of only $\bullet$-vertices and at least one $\star$-vertex.  

\item

A \emph{$\times$-graph} is a plain-graph with one $\bullet$-vertex replaced with a $\times$-vertex.  

\item

A \emph{$\star$-graph} is a graph with one $\times$-vertex and at least one $\star$-vertex.

\end{enumerate}
\end{Def}
   
\noindent We define sets of graphs and the real vector spaces that they generate:
\begin{Def} \label{def: vector space} \quad
\begin{enumerate}

\item

$\mathcal{G}_{\nv , \Delta}^{}$ is the set of plain-graphs with $\nv$ $\bullet$-vertices and $\Delta$ edges.  

\item

$\mathcal{G}^{ \times }_{\nv , \Delta}$ is the set of $\times$-graphs with $\nv - 1$ $\bullet$-vertices, one $\times$-vertex and $\Delta$ edges.  

\item

$\mathcal{G}^{ \star }_{\nv , \Delta}$ is the set of $\star$-graphs with $\nv - 1$ $\bullet$-vertices, one $\times$-vertex, at least one $\star$-vertex and $\Delta$ edges.  

\end{enumerate}

\noindent In the above graphs, we choose the labels of the $\bullet$- and $\times$-vertices to go from $v_1$ to $v_\nv$ and the labels of the $\star$-vertices to go from $v_{\raisebox{-0.2ex}{$\scriptstyle{1}$}}^\star$ to $v_{\nsv}^\star$, where $\nsv$ is the number of $\star$-vertices.  Let $\mathcal{L}_{\nv , \Delta}^{}$, $\mathcal{L}^{ \times }_{\nv , \Delta}$ and $\mathcal{L}^{ \star }_{\nv , \Delta}$ be the real vector spaces of formal linear combinations generated by $\mathcal{G}_{\nv , \Delta}^{}$, $\mathcal{G}^{ \times }_{\nv , \Delta}$ and $\mathcal{G}^{ \star }_{\nv , \Delta}$, respectively.  The zero vector in any of these vector spaces is the null graph.
\end{Def}

\noindent Note that $\nv = \nv (\bullet) + \nv (\times)$, where $\nv (\bullet)$ is the number of $\bullet$-vertices and $\nv (\times)$ is the number of $\times$-vertices.  $\nv$ does not include $\nv (\star)$ since $\star$-vertices represent neither $\phi$ nor $\delta_\pd \phi$.  These sets of graphs are finite and therefore the vector spaces of formal linear combinations are finite-dimensional.  

	By Definition \ref{def: vector space}, graphs in a linear combination $L \in \mathcal{L}_{\nv , \Delta}^{}$ ($\mathcal{L}^{ \times }_{\nv , \Delta}$ or $\mathcal{L}^{ \star }_{\nv , \Delta})$ share the same number of $\nv$ and $\Delta$.  In most of the following discussion, $\nv$ and $\Delta$ are fixed.  We will therefore omit theses subscripts as long as no confusion arises.  However, the number of $\star$-vertices $\nsv$ is not fixed in a generic linear combination of $\star$-graphs.
		
\subsubsection{Maps}

	We now define some maps between the sets and vector spaces in Definition \ref{def: vector space}.  This will model the operations that act on the algebraic expressions represented by the graphs.
	
	Firstly, the variation under the polynomial shift $\delta_\pd$ of an algebraic term, expressed by a graph $\Gamma$, is represented graphically by summing over all graphs that have one $\bullet$-vertex in $\Gamma$ replaced with a $\times$-vertex.

\begin{Def} [Variation Map]
	Given a plain-graph $\Gamma \in \mathcal{G}$, with $V(\Gamma) = (  v_1, \ldots, v_{\nv})$, the map $\delta_\pd: \mathcal{G} \rightarrow \mathcal{L}^{ \times }$ is defined by $\delta_\pd ( \Gamma ) = \sum_{i = 1}^{\nv} \Gamma_i^{ \times }$, where $\Gamma_i^{ \times }$ is a graph given by replacing $v_i$ with a $\times$-vertex.  This map extends to $\mathcal{L} \rightarrow \mathcal{L}^{ \times }$ by distributing $\delta_\pd$ over the formal sum.
\end{Def}

\noindent Note that $\Gamma^{ \times }_{i}$ is the null graph if $v_i$ has degree greater than $\pd$.  We will omit the subscript $\pd$ in $\delta_\pd$ as long as no confusion arises.  It is also necessary to define a map that operates in the reverse direction:

\begin{Def} 
	The map $v : \mathcal{G}^{ \times } \rightarrow \mathcal{G}$ is defined by replacing the $\times$-vertex with a $\bullet$-vertex.
\end{Def}

	In the algebraic expressions, a total derivative term looks like $\partial_i L_i$, and the $\partial_i$ can be distributed over $L_i$ as usual, by applying the Leibniz rule.  This feature will be captured by the graphical representation in the following definition.  

\begin{Def}[Derivative Map] \label{def: der map}
For a given $\star$-graph $\Gamma^{ \star } \in \mathcal{G}^{ \star }$, the \emph{derivative map} $\rho : \mathcal{G}^{ \star } \rightarrow \mathcal{L}^{ \times }$ is defined using the following construction: 
\begin{enumerate}

\item \label{Step 1 Def der map}
	
	For the $\star$-graph $\Gamma^{ \star }$, denote the $\bullet$-vertices by $v_1,\ldots,v_{\nv -1}$, the $\times$-vertex by $v_\nv$ and the $\star$-vertices by $v_1^{ \star },\ldots,v_k^{ \star }$, $k = N (\star)$.  Take any $\star$-vertex $v^{ \star }_i$ in $\Gamma^{ \star }$.  For each ${j_1 \in \{ 1, \ldots , \nv \}}$, form a graph $\Gamma_{j_1}$ by deleting $v^{ \star }_1$ in $\Gamma^{ \star }$ and then adding an edge joining $v_{j_1}$ and the vertex that was adjacent to $v^{ \star }_1$ in $\Gamma^{ \star }$.  
	
\item \label{Step 2 Def der map}

    Apply the above procedure to each of the $\Gamma_{j_1}$ to form $\Gamma_{j_1j_2}$ by removing the next $v_2^\star$.  Iterate this procedure until all $\star$-vertices have been removed, forming the $\times$-graph $\Gamma_{j_1\ldots j_k}$.  
    
\item
      Define $\rho (\Gamma^{ \star }) \equiv \sum_{j_1,\ldots,j_k=1}^{\nv}\Gamma^{ \times }_{j_1\ldots j_k}$.

\end{enumerate}

\noindent The domain of this map can be extended to $\mathcal{L}^{ \star }$ by distributing $\rho$ over the formal sum.  The derivative map $\rho$ can be similarly defined on $\star$-ed plain-graphs.  Furthermore, we take $\rho$ to be the identity map when it acts on $\times$-graphs.
\end{Def}

\noindent Note that the above definition is well-defined since $\rho$ is independent of the order in which the $\star$-vertices are deleted.  

\subsubsection{Relations}

	There are many linear combinations of plain- and $\times$-graphs representing terms that can be written as a total derivative.  To take this feature into account, we define two notions of relations for plain- and $\times$-graphs, respectively.

\begin{Def}[Relations] \label{def: relations} 
        If a linear combination of plain-graphs $L\in\mathcal{L}$ can be written as $\rho ( L' )$, where $L'$ is a sum of $\star$-ed plain-graphs, then $L$ is called a \emph{plain-relation}.  If a linear combination of $\times$-graphs $L^{ \times }\in\mathcal{L}^{ \times }$ can be written as $\rho ( L^{ \star } )$, with $L^{ \star }$ a sum of $\star$-graphs, then $L^{ \times }$ is called a $\times$-\emph{relation}.  
\end{Def}

\noindent We shall denote the set of all plain-relations by $\mathcal{R}$ and the set of all $\times$-relations by $\mathcal{R}^{ \times }$.  $\mathcal{R}$ and $\mathcal{R}^{ \times }$ have a natural vector space structure and are subspaces of $\mathcal{L}$ and $\mathcal{L}^{ \times }$, respectively.

\subsubsection{The Consistency Equation and Associations}\label{consandassoc}

	Recall that $P$-invariants are defined algebraically by equation (\ref{eq: alg. consistency}), $\delta_\pd ( L ) = \partial_i( L_i )$.  This equation is written in the graphical representation as
\begin{equation} \label{eq: consistency}
	\delta_\pd ( L ) = \rho ( L^{ \star } ),
\end{equation}
\noindent for $L \in \mathcal{L}$ and $L^{ \star }\in\mathcal{L}^{ \star }$.  We call (\ref{eq: consistency}) the \emph{consistency equation}.  Searching for $\pd$-invariants is equivalent to constructing all consistency equations.  Note that, by Definition \ref{def: relations}, the consistency equation implies that $\delta_\pd ( L )$ is a $\times$-relation and so we make the following definition:
	
\begin{Def}[$\pd$-Invariant]
	$L \in \mathcal{L}$ is a \emph{$\pd$-invariant} if $\delta_\pd ( L ) \in \mathcal{R}^{ \times }$.
\end{Def}

\noindent Furthermore, there is a simple class of $\pd$-invariants, which we call \emph{exact $\pd$-invariants}.  These represent terms which are \emph{exactly} invariant under the polynomial shift symmetry \eqref{eq: polyshift}, not just up to a total derivative.

\begin{Def} [Exact $\pd$-Invariant]
	$L \in \mathcal{L}$ is an \emph{exact $\pd$-invariant} if $\delta_\pd ( L ) = 0$.
\end{Def}

   The following notion, called ``association between graphs", will turn out to be indispensable in constructing consistency equations.

\begin{Def}[Associations] The \emph{associations} between pairs of plain- and $\times$-graphs, plain- and plain-graphs, $\star$- and $\times$-graphs, $\star$- and $\star$-graphs and plain- and $\star$-graphs are defined as follows:
\begin{enumerate}

\item	

$\Gamma \in \mathcal{G}$ and $\Gamma^{ \times } \in \mathcal{G}^{ \times }$ are associated with each other if $\Gamma^{ \times }$ is contained in $\delta_\pd ( \Gamma )$, or, equivalently, $v (\Gamma^{ \times }) = \Gamma$. 

\item 

Any two graphs $\Gamma_1, \Gamma_2 \in \mathcal{G}$ are associated with each other if either they are associated with the same $\Gamma^{ \times } \in \mathcal{G}^{ \times }$, or $\Gamma_1$ is identical to $\Gamma_2$. 

\item

$\Gamma^{ \star }\in\mathcal{G}^{ \star }$ and $\Gamma^{ \times }\in\mathcal{G}^{ \times }$ are associated with each other if $\Gamma^{ \times }$ is contained in $\rho (\Gamma^{ \star })$. 

\item

Any two graphs $\Gamma_1^{ \star }, \Gamma_2^{ \star } \in \mathcal{G}^{ \star }$ are associated with each other if either they are associated with the same $\Gamma^{ \times }\in\mathcal{G}^{ \times }$, or $\Gamma_1^{ \star }$ and $\Gamma_2^{ \star }$ are identical to each other. 

\item

Any two graphs $\Gamma\in\mathcal{G}$ and $\Gamma^{ \star } \in \mathcal{G}^{ \star }$ are associated with each other if they are associated with the same $\Gamma^{ \times }\in\mathcal{G}^{ \times }$.

\end{enumerate}
\end{Def}

\noindent It turns out that the associations between only plain-graphs and $\times$-graphs have a simple structure.  Note that for any $\times$-graph $\Gamma^{ \times }\in \mathcal{G}^{ \times }$, $v ( \Gamma^{ \times } )$ uniquely defines the associated plain-graph.  Hence, 

\begin{Prop} \label{lem: 1}
	A $\times$-graph is associated with exactly one plain-graph.  
\end{Prop}

\noindent The corollaries below directly follow:

\begin{Cor} \label{cor: d}
	For $L \in \mathcal{L}$ and $\Gamma^{ \times }$ a $\times$-graph in $\delta ( L )$, $L$ contains the plain-graph $v (\Gamma^{ \times })$.
\end{Cor}

\begin{Cor} \label{prop: 1}
	Any two associated plain-graphs are identical to each other.
\end{Cor}

\begin{proof}
	If two distinct plain-graphs are associated with each other, then they are associated with a common $\times$-graph, which violates Proposition \ref{lem: 1}.  Therefore, only identical plain-graphs are associated with each other.  
\end{proof}

\begin{Cor} \label{cor: a}
	For $L \in \mathcal{L}$ and a plain-graph $\Gamma$ in $L$, $\delta (L)$ contains all $\times$-graphs in $\delta (\Gamma)$.
\end{Cor}

\begin{proof} 
	Without loss of generality, suppose $\Gamma$ appears in $L$ with unit coefficient (otherwise, simply divide $L$ by the coefficient of $\Gamma$).  Let $\Gamma^{ \times }$ be a $\times$-graph in $\delta ( \Gamma )$.  By Proposition \ref{lem: 1}, $\Gamma$ is the only plain-graph associated with $\Gamma^{ \times }$.  Therefore, $\Gamma^{ \times }$ cannot drop out of $\delta ( \Gamma + L' )$ for any $L' \in \mathcal{L}$ that does not contain $\Gamma$.  Applying this statement to $L' = L- \Gamma$ proves that $\Gamma^{ \times }$ must appear in $\delta( L )$.  
\end{proof}

\noindent This now allows us to find all \emph{exact} $\pd$-invariants in a simple manner:

\begin{Cor} \label{lem: exact}
	$L  \in \mathcal{L}$  is an exact $\pd$-invariant if and only if all vertices in all graphs contained in $L$ have degree at least $\pd + 1$.    
\end{Cor}

\begin{proof} 
	If there is a vertex $v$ in some plain-graph $\Gamma$ in $L$ of degree lower than $\pd + 1$, then $\delta_\pd ( \Gamma )$ contains the $\times$-graph where $v$ is replaced with a $\times$-vertex.  But by Corollary \ref{cor: a}, this means that $\delta_\pd ( L ) $ also contains this $\times$-graph, which contradicts $\delta_\pd ( L ) = 0$.  
\end{proof}
\noindent All of the above definitions and conclusions make sense when extended to $\pd < 0$.  Even though $\pd < 0$ no longer corresponds to any polynomial shift symmetry, it will occasionally be useful to consider graphs with $\pd < 0$.  Since any vertex has a non-negative degree, Corollary \ref{lem: exact} implies:
\begin{Cor} \label{lem: Pnegative}
	If $\pd < 0$, any $L  \in \mathcal{L}$  is an exact $\pd$-invariant.    
\end{Cor}

\noindent Associations between $\star$-graphs and $\times$-graphs also have a simple and useful property:

\begin{Lem} \label{lem: star star association}
	Suppose that a $\times$-graph $\Gamma^{ \times } \in \mathcal{G}^{ \times }$ is associated with a $\star$-graph $\Gamma^{ \star } \in \mathcal{G}^{ \star }$ that contains a single $\star$-vertex.  Then $\Gamma^{ \times }$ appears in $\rho ( \Gamma^{ \star } )$ with coefficient 1.
\end{Lem} 

\noindent  In general, a $\times$-graph can be associated with more than one $\star$-graph.  Figure $\ref{fig: times relation}$ presents a simple example with $(\pd , \nv , \Delta) = (2, 3, 2)$.  Consequently, there can exist multiple consistency equations for the same $\pd$-invariant.  In the next section we will develop techniques to deal with this difficulty.  

\begin{figure}[t!]
\centering
\begin{align*}
\rho \scalebox{1.3}{\Bigg{(}}\hspace{-0.1cm}
\begin{minipage}{\mini cm}
\begin{tikzpicture}
\draw [thick] (\1,\1) -- (0,\1);
\draw [thick] (0,\1) -- (0.4,0.4);
\filldraw (0,0) circle [radius=0.08];
\filldraw (\1,0) circle [radius=0.08];
\filldraw [white] (0,\1) circle [radius=0.115];
\draw [thick] (0,\1) circle [radius=0.115];
\node at (0,\1) {$\times$};
\node at (\1,\1) {\scalebox{0.8}{$\bigstar$}};
\node at (0.4,0.4) {\scalebox{0.8}{$\bigstar$}};
\end{tikzpicture}
\end{minipage}\scalebox{1.3}{\Bigg{)}}
= 2
\begin{minipage}{\mini cm}
\begin{tikzpicture}
\draw [thick] (0,\1) -- (0,0);
\draw [thick] (0,\1) -- (\1,0);
\filldraw (0,0) circle [radius=0.08];
\filldraw (\1,0) circle [radius=0.08];
\filldraw [white] (0,\1) circle [radius=0.115];
\draw [thick] (0,\1) circle [radius=0.115];
\node at (0,\1) {$\times$};
\end{tikzpicture}
\end{minipage}
+
\begin{minipage}{\mini cm}
\begin{tikzpicture}
\filldraw (0,0) circle [radius=0.08];
\filldraw (\1,0) circle [radius=0.08];
\draw [thick] (0,0) to [out=110,in=-110] (0,\1);
\draw [thick] (0,0) to [out=70,in=-70] (0,\1);
\filldraw [white] (0,\1) circle [radius=0.115];
\draw [thick] (0,\1) circle [radius=0.115];
\node at (0,\1) {$\times$};
\end{tikzpicture}
\end{minipage}
+
\begin{minipage}{\mini cm}
\begin{tikzpicture}
\filldraw (0,0) circle [radius=0.08];
\filldraw (\1,0) circle [radius=0.08];
\draw [thick] (0,\1) to [out=-25,in=115] (\1,0);
\draw [thick] (0,\1) to [out=-65,in=155] (\1,0);
\filldraw [white] (0,\1) circle [radius=0.115];
\draw [thick] (0,\1) circle [radius=0.115];
\node at (0,\1) {$\times$};
\end{tikzpicture}
\end{minipage}
= \rho \scalebox{1.3}{\Bigg{(}} \hspace{-1mm}
\begin{minipage}{\mini cm}
\begin{tikzpicture}
\filldraw (0,0) circle [radius=0.08];
\filldraw (\1,0) circle [radius=0.08];
\draw [thick] (0, 0) -- (0,\1);
\draw [thick] (\1,\1) -- (0,\1);
\filldraw [white] (0,\1) circle [radius=0.115];
\draw [thick] (0,\1) circle [radius=0.115];
\node at (0,\1) {$\times$};
\node at (\1,\1) {\scalebox{0.8}{$\bigstar$}};
\end{tikzpicture}
\end{minipage}+
\hspace{0.03cm} \hspace{0.03cm}
\begin{minipage}{\mini cm}
\begin{tikzpicture}
\filldraw (0,0) circle [radius=0.08];
\filldraw (\1,0) circle [radius=0.08];
\draw [thick] (0,\1) -- (\1,\1);
\draw [thick] (0,\1) -- (\1,0);
\filldraw [white] (0,\1) circle [radius=0.115];
\draw [thick] (0,\1) circle [radius=0.115];
\node at (0,\1) {$\times$};
\node at (\1,\1) {\scalebox{0.8}{$\bigstar$}};
\end{tikzpicture}
\end{minipage}\scalebox{1.3}{\Bigg{)}}
\end{align*}
\caption{Two different linear combinations of $\star$-graphs result in an identical $\times$-relation for $\pd = 2$.  In particular, the $\times$-graph with a coefficient 2 is associated with all three $\star$-graphs in the figure.}
\label{fig: times relation}
\end{figure}

\subsection{Building Blocks for Consistency Equations} \label{sec: graphical basis} 

In this section, we introduce the building blocks with which we will construct the consistency equation \eqref{eq: consistency}, $\delta_\pd ( L ) = \rho ( L^{ \star } )$.  Any polynomial shift-invariant can be generated using these building blocks.  We show that we can constrain $L$ to contain only loopless plain-graphs, and all other invariants are equal to these ones up to total derivatives.  Consequently, $\delta_\pd (L)$, and thus $\rho(L^*)$, contains only loopless $\times$-graphs.  Therefore, $\rho (L^\star) = \rho^{(0)} (L^\star)$, where $\rho^{(0)}$ acts in the same way as $\rho$ but omits any looped graphs (Definition \ref{def: loopless r}).  In fact, we can restrict $L^{\star}$ to be a linear combination $L_M$ of a particular type of $\star$-graph, such that $\rho^{(0)} (L^\star) = \rho^{(0)} (L_M)$.  These $\star$-graphs will be called Medusas (Definition \ref{def: Medusa}).  In Appendix \ref{sec: lower bound} we determine a lower bound on the degree of a vertex in any graph that appears in the consistency equation.

\subsubsection{The Loopless Realization of \texorpdfstring{$\mathcal{L}/\mathcal{R}$}{L/R}} \label{sec: loopless basis}
	
	There are usually many alternative expressions for a single $\pd$-invariant algebraic term, which are equal to each other up to total derivatives.  In the graphical language, the graphs representing these equivalent expressions are related by plain-relations.  Therefore, we are interested in the space of linear combinations of plain-graphs modding out plain-relations, i.e., the quotient space $\mathcal{L}/\mathcal{R}$.  We need to find subset of graphs, $\mathcal{B} \subset \mathcal{G}$, whose span is isomorphic to $\mathcal{L}/\mathcal{R}$.  In other words, every element of $\mathcal{L}$ can be written as a linear combination of graphs in $\mathcal{B}$ and plain-relations.  Furthermore, this means that there are no plain-relations between elements in the set $\mathcal{B}$.  The following proposition shows that the set of loopless plain-graphs realizes the set $\mathcal{B}$.
\begin{Prop}[Loopless Basis] \label{prop: loopless_basis}
    The span of loopless plain-graphs is isomorphic to $\mathcal{L}/\mathcal{R}$.
\end{Prop}

\begin{proof}
	Denote the span of all loopless plain-graphs by $\mathcal{L}_\text{loopless}$.  If $\partial_i \partial_i$ acts on a single $\phi$, then one can always integrate by parts to move one of the $\partial_i$'s to act on the remaining $\phi$'s.  In the graphical language, this means that any graph with loops can always be written as a linear combination of loopless graphs up to a plain-relation.  This proves $\mathcal{L}/\mathcal{R} \subset \mathcal{L}_\text{loopless}$.    
    
	Now, we show that there are no plain-relations between the loopless plain-graphs.  Suppose there exists a linear combination of loopless plain-graphs $L$ that is a plain-relation.  That is, there exists a linear combination $L'$ of $\star$-ed plain-graphs such that $L = \rho ( L' )$.  Let $\Gamma'$ be a $\star$-ed plain-graph in $L'$.  Then, $\rho ( \Gamma' )$ is a linear combination of plain-graphs each of which has a number of loops no greater than the number of $\star$-vertices in $\Gamma'$.  There will be exactly one graph, $\Gamma_{\text{f.l.}}$, which is fully-looped (with the number of loops equal to the number of $\star$-vertices in $\Gamma'$), produced when all the original $\star$-vertices (and the edges incident to them) are replaced with loops.  Furthermore, $\Gamma_{\text{f.l.}}$ uniquely determines $\Gamma'$ by replacing each loop in $\Gamma_{\text{f.l.}}$ with an edge incident to an extra $\star$-vertex.  Choose the $\star$-ed plain-graph appearing in $L'$ with the largest number of $\star$-vertices (the maximally $\star$-ed plain-graphs).  This maximum exists since the number of the edges incident to the $\star$-vertices is bounded above by the number of edges $\Delta$.  The fully-looped graphs formed from these maximally $\star$-ed plain-graphs cannot cancel each other (by uniqueness) and cannot be canceled by any other graphs that are not fully-looped (by maximality).  Therefore $\rho ( L' )$ is a linear combination containing looped graphs, which contradicts the initial assumption that $L$ only consists of loopless graphs.  This proves $\mathcal{L}_\text{loopless} \subset \mathcal{L}/\mathcal{R}$.
    
    Therefore, $\mathcal{L}_\text{loopless} \cong \mathcal{L}/\mathcal{R}$.
\end{proof}

\noindent Henceforth, we can restrict our search for $\pd$-invariants to $\mathcal{L}_\text{loopless}$.  Note that if $L\in\mathcal{L}_\text{loopless}$, then all the graphs in $\delta ( L )$ are also loopless, so it is sufficient to consider only loopless $\times$-graphs.  We can also restrict to loopless $\times$-relations, $\mathcal{R}^{ \times }_\text{loopless} \subset \mathcal{R}^{ \times }$, which is the vector space consisting of $\times$-relations that are linear combinations of loopless graphs.  We summarize this discussion in the following corollary:

\begin{Cor}\label{Cor: kernel}
The $\pd$-invariants that are independent up to total derivatives are represented by $L \in \mathcal{L}_\text{loopless}$ with $\delta ( L ) \in \mathcal{R}^{ \times }_\text{loopless}$.  Equivalently, they span the kernel of the map:
\begin{equation*}
q\circ\delta: \mathcal{L}_\text{loopless} \rightarrow \mathcal{L}^{ \times }_\text{loopless}/\mathcal{R}^{ \times }_\text{loopless},
\end{equation*}
\noindent where $q$ is the quotient map $q: \mathcal{L}^{ \times }_\text{loopless} \rightarrow \mathcal{L}^{ \times }_\text{loopless}/\mathcal{R}^{ \times }_\text{loopless}$.  
\end{Cor}

\subsubsection{Medusas and Spiders} \label{sec: Medusa spider}

	Corollary \ref{Cor: kernel} motivates us to look for a basis of $\mathcal{R}^{ \times }_\text{loopless}$, the $\times$-relations that are linear combinations of loopless graphs.  To classify all such loopless $\times$-relations, we are led to study the linear combinations of $\star$-graphs that give rise to these relations under the derivative map $\rho$.  
	We will realize a convenient choice for the basis of $\mathcal{R}^{ \times }_\text{loopless}$, which will tremendously simplify our calculations: It turns out that the basis of $\mathcal{R}^{ \times }_\text{loopless}$ is in one-to-one correspondence with a particular subset of loopless $\star$-graphs, which we now define.
\begin{Def}[Medusa] \label{def: Medusa}
	A \emph{Medusa} is a loopless $\star$-graph with all $\star$-vertices adjacent to the $\times$-vertex, such that the degree of the $\times$-vertex $\text{\emph{deg}}(\times)$ and the number of $\star$-vertices $\nsv$ satisfy $\text{\emph{deg}}(\times) = \pd + 1 - \nsv$.  We denote the set of Medusas by $\mathcal{M}_{\nv , \Delta}$.
\end{Def}
\noindent We should point out that applying $\rho$ to a Medusa does not necessarily generate a $\times$-relation in $\mathcal{R}^{ \times }_{\text{loopless}}$; it will sometimes produce $\times$-graphs with loops.  In order to form a loopless $\times$-relation, these looped $\times$-graphs must be canceled by contributions from other $\star$-graphs.  

	In the proof of the one-to-one correspondence between the basis of $\mathcal{R}^\times_\text{loopless}$ and the subset of Medusas, we will frequently refer to the following definitions.

\begin{Def} [Primary $\star$-Graphs] 
	A \emph{primary $\star$-graph} is a $\star$-graph that contains exactly one $\star$-vertex.  
\end{Def}

\begin{Def}[Spider]
    A \emph{spider} is a primary $\star$-graph with the $\star$-vertex adjacent to the $\times$-vertex and $\text{\emph{deg}}(\times) = \pd$.
\end{Def}

\noindent Since $\text{deg}(\times) \geq \nsv$ for a Medusa, we have $\frac{1}{2}( \pd +1) \leq \text{deg}(\times) \leq \pd$.  In particular, if $\pd = 1$ or $2$, then $\text{deg}(\times) = \pd$ and $\nsv =1$ (i.e., a Medusa is a loopless spider for $\pd =1$ or $2$).  Spiders will play an important role in sorting out all independent loopless $\times$-relations in Appendix \ref{sec: Basis of Relations}, and loopless spiders will lead us to the classification of 1-invariants in Appendix \ref{sec: linear}.

	In \ref{sec: pyramid} and \ref{sec: Basis of Relations} we will construct $\mathcal{R}^{ \times }_\text{loopless}$ by spiders and Medusas.  We will need to keep track of graphs containing specific numbers of loops, and the following refinement of the derivative map $\rho$ (Definition \ref{def: der map}) will allow us to formulate the graphical operations algebraically.

\begin{Def}  \label{def: loopless r} 
	The map $\rho^{(\ell)}: \mathcal{L}^{ \star } \rightarrow \mathcal{L}^{ \times }$ is defined such that, for $L^{ \star } \in \mathcal{L}^{ \star }$, $\rho^{(\ell)}(L^{ \star })$ is equal to $\rho (L^{ \star })$, with the coefficient of any graph that does not contain $\ell$ loops set to zero.  
\end{Def}

\noindent Note that $\rho = \sum_{i=0}^\infty \rho^{(i)}$.  If $\Gamma^{ \star } \in \mathcal{G}^{ \star }$ contains a loop, then $\rho^{(0)} (\Gamma^{ \star })$ is identically zero.  In Theorem \ref{thm: loopless times basis} we will show that the map $\rho^{(0)}$ defines the one-to-one correspondence between $\mathcal{M}$ and the preferred basis of $\mathcal{R}^{ \times }_\text{loopless}$.  
	
	It is also useful to introduce the operation of ``undoing" loops.
	
\begin{Def} \label{def: undo loop}
	Given a $\Gamma \in \mathcal{G}^{ \times } \cup \mathcal{G}^{ \star }$ which contains $\ell$ loops, labeled from 1 to $\ell$, the map $\theta_i: \mathcal{G}^{ \times } \cup \mathcal{G}^{ \star } \rightarrow \mathcal{G}^{ \star }$ is defined such that $\theta_i ( \Gamma )$ is a $\star$-graph constructed by deleting the $i^{\mathrm{th}}$ loop from $\Gamma$, adding an extra $\star$-vertex $v^{ \star }$ and adding an edge joining $v^{ \star }$ to the vertex at which the $i^{\mathrm{th}}$ deleted loop ended.  Define $\theta: \mathcal{G}^{ \times } \cup \mathcal{G}^{ \star } \rightarrow \mathcal{G}^{ \star }$ by $\theta(\Gamma) \equiv \theta_1 \circ \cdots \circ \theta_\ell ( \Gamma )$.  This map extends to $\mathcal{L}^{ \times } \cup \mathcal{L}^{ \star } \rightarrow \mathcal{L}^{ \star }$ by distributing $\theta$ over the formal sum.  The map $\theta$ can be similarly defined on $\star$-ed plain-graphs.
\end{Def} 

\noindent We will need to distinguish different types of loops:

\begin{Def}
	A loop at the $\times$-vertex is called a \emph{$\times$-loop}, and a loop at a $\bullet$-vertex is a \emph{$\bullet$-loop}.  A graph that contains $\ell$ loops is called \emph{$\ell$-looped}.
\end{Def}

	We can now prove a key formula:

\begin{Prop} \label{prop: order}
	If an $\ell$-looped $\star$-graph $\Gamma^{ \star }$ contains a loop at vertex $v_A$, then
\begin{equation}\label{ordereqn}
	\rho^{(\ell - 1)} \circ \theta_A (\Gamma^{ \star }) = \rho^{(\ell-1)}  \circ \theta_A \circ \rho^{(\ell)} (\Gamma^{ \star }) + \rho^{(\ell-1)} \left ( \text{\emph{L}}^{(\ell-1)}_{\text{\emph{spider}}} \right ),
\end{equation} 
\noindent where $\theta_A$ undoes a loop at $v_A$, and $\text{\emph{L}}^{(\ell-1)}_{\emph{Spider}}$ is a linear combination of $( \ell - 1 )$-looped spiders.  Every graph in a nonzero $\text{\emph{L}}^{(\ell-1)}_{\emph{Spider}}$ has one fewer $\times$-loop than $\Gamma^{ \star }$.  $\text{\emph{L}}^{(\ell-1)}_{\emph{Spider}}$ is nonzero if and only if the following three conditions are satisfied: 
\begin{enumerate}[$\emph{(}$a$\emph{)}$]

\item \label{Condition 1}
    
    $v_A$ is the $\times$-vertex;

\item \label{Condition 2}

    There is a $\star$-vertex that is not adjacent to the $\times$-vertex;

\item \label{Condition 3}

    $\text{\emph{deg}}(\times) \geq \pd + 1 - \nbsv$, where $\nbsv$ is the number of $\star$-vertices in $\Gamma^{ \star }$ that are adjacent to $\bullet$-vertices.  

\end{enumerate} 
\end{Prop}

\begin{proof}
	Denote the vertices in $\Gamma^{ \star }$ adjacent to $\star$-vertices by $v_i$, $i = 1, \ldots, k$, and the $\times$-vertex by $v^{ \times }$.  Note that $v^\times$ and $v_A$ may coincide with each other and with some of the $v_i$'s.  Form an $(\ell-1)$-looped graph $\Gamma$ from $\Gamma^{ \star }$ by deleting all $\star$-vertices and deleting a loop at $v_A$.  We now show that all graphs in $\rho^{(\ell - 1)} \circ \theta_A (\Gamma^{ \star })$ and $\rho^{(\ell-1)}  \circ \theta_A \circ \rho^{(\ell)} (\Gamma^{ \star })$ in \eqref{ordereqn} contain $\Gamma$ as a subgraph: 
	
 \begin{itemize}

\item 

	$\rho^{(\ell - 1)} \circ \theta_A (\Gamma^{ \star })$: Define $\Gamma_{v_{\beta_1} \ldots v_{\beta_k}}$ to be a graph formed from $\Gamma$ by adding $k$ edges joining $v_i$ and $v_{\beta_i}$, respectively.  Then 
\begin{equation} \label{eq: no tilde}
    \rho^{(\ell - 1)} \circ \theta_A (\Gamma^{ \star }) = \sum_{v_\alpha \neq v_A} \sum_{v_{\beta_1} \neq v_1} \ldots \sum_{v_{\beta_k} \neq v_k} \Gamma^{v_\alpha}_{v_{\beta_1} \ldots v_{\beta_k}},
\end{equation}
\noindent where $\Gamma^{v_\alpha}_{v_{\beta_1} \ldots v_{\beta_k}}$ is the graph $\Gamma_{v_{\beta_1} \ldots v_{\beta_k}}$ with an extra edge joining $v_A$ and $v_\alpha$.  Note that the graphs formed from $\Gamma$ by adding edges joining $v_A$ (or $v_i$) to itself are not included in this sum, because these graphs are not $(\ell-1)$-looped.

\item

	$\rho^{(\ell-1)}  \circ \theta_A \circ \rho^{(\ell)} (\Gamma^{ \star })$: Define $\Gamma_A$ to be the graph $\Gamma$ with a loop added at $v_A$.  Then $\rho^{(\ell)} ( \Gamma^{ \star } ) = \sum_{v_{\beta_i} \neq v_i} \tilde{\Gamma}_{v_{\beta_1} \ldots v_{\beta_k}}$, where $\tilde{\Gamma}_{v_{\beta_1} \ldots v_{\beta_k}}$ is the graph formed from $\Gamma_A$ by adding $k$ edges joining $v_i$ and $v_{\beta_i}$, respectively.  $v_{\beta_i} \neq v_i$ in the sum because only $\ell$-looped graphs are included.  Applying $\rho^{(\ell-1)} \circ \theta_A$ gives:
\begin{equation} \label{eq: tilde}
    \rho^{(\ell - 1)}\circ \theta_A \circ \rho^{(\ell)} ( \Gamma^{ \star } ) = \sum_{v_\alpha \neq v_A} \sum_{v_{\beta_1} \neq v_1} \ldots \sum_{v_{\beta_k} \neq v_k} \tilde{\Gamma}^{v_\alpha}_{v_{\beta_1} \ldots v_{\beta_k}},
\end{equation}
\noindent where $\tilde{\Gamma}^{v_\alpha}_{v_{\beta_1} \ldots v_{\beta_k}}$ is formed from $\tilde{\Gamma}_{v_{\beta_1} \ldots v_{\beta_k}}$ by deleting a loop at $v_A$ and adding an edge joining $v_A$ and $v_\alpha$.  $v_\alpha \neq v_A$ in the sum because only $(\ell-1)$-looped graphs are included.

\end{itemize}

    We now want to compare $\Gamma^{v_\alpha}_{v_{\beta_1} \ldots v_{\beta_k}}$ and $\tilde{\Gamma}^{v_\alpha}_{v_{\beta_1} \ldots v_{\beta_k}}$ for $v_\alpha \neq v_A$ and $v_{\beta_i} \neq v_i$.  At first it might seem that $\Gamma^{v_\alpha}_{v_{\beta_1} \ldots v_{\beta_k}} = \tilde{\Gamma}^{v_\alpha}_{v_{\beta_1} \ldots v_{\beta_k}}$, since both ultimately involve taking $\Gamma$ and adding an edge joining $v_A$ and $v_\alpha$, and edges joining $v_i$ and $v_{\beta_i}$.  However, there is a subtlety involved: Recall that graphs containing a $\times$-vertex of degree larger than $P$ are identified with the null graph.  Therefore, $\Gamma^{v_\alpha}_{v_{\beta_1} \ldots v_{\beta_k}} = \tilde{\Gamma}^{v_\alpha}_{v_{\beta_1} \ldots v_{\beta_k}}$, provided neither side is null or both sides are null; when one side of this equation represents the null graph and the other does not, then this equation will not hold.  This violation happens only if there is a difference in deg$(\times)$ of the graphs formed during the construction of $\Gamma^{v_\alpha}_{v_{\beta_1} \ldots v_{\beta_k}}$ and $\tilde{\Gamma}^{v_\alpha}_{v_{\beta_1} \ldots v_{\beta_k}}$.  Note that we add the same $k$ edges (joining $v_i$ and $v_{\beta_i}$) to both $\Gamma$ and $\Gamma_A$ to form the intermediate graphs, $\Gamma_{v_{\beta_1} \ldots v_{\beta_k}}$ and $\tilde{\Gamma}_{v_{\beta_1} \ldots v_{\beta_k}}$, respectively.  Thus the difference between the latter two graphs is the same as the difference between $\Gamma$ and $\Gamma_A$: There is an extra loop at $v_A$ in $\Gamma_A$ compared to $\Gamma$, which will only be deleted after the edges have been added.  Hence, the violation of the equality, $\Gamma^{v_\alpha}_{v_{\beta_1} \ldots v_{\beta_k}} = \tilde{\Gamma}^{v_\alpha}_{v_{\beta_1} \ldots v_{\beta_k}}$, happens only if $v_A = v^{ \times }$.  This is condition (\ref{Condition 1}).  
    
    From now on, we assume $v_A = v^{ \times }$.  In addition, there must exist at least one $v_{\beta_i} = v^{ \times }$ in order for the violation to occur, since otherwise deg$(\times)$ in any graph we are considering never exceeds the one in $\Gamma^{ \star }$, and thus no graph in \eqref{eq: no tilde} and \eqref{eq: tilde} is null.  Hence for at least one $v_{\beta_i}$, $v^{ \times } = v_{\beta_i} \neq v_i$, which implies condition (\ref{Condition 2}).  From now on we assume $v^{ \times } = v_{\beta_i} \neq v_i$, for at least one $v_{\beta_i}$, and denote the number of $v_{\beta_i}$ equal to $v^{ \times }$ by $b$, where $1\leq b\leq \nbsv$.
	
	We know that deg$(\times)$ in $\tilde{\Gamma}_{v_{\beta_1} \ldots v_{\beta_k}}$ is 2 higher than deg$(\times)$ in $\Gamma_{v_{\beta_1} \ldots v_{\beta_k}}$, due to the one extra loop at $v_A$ in $\Gamma_A$.   
	
	Therefore, if $\Gamma^{v_\alpha}_{v_{\beta_1} \ldots v_{\beta_k}}$ vanishes, then $\tilde{\Gamma}^{v_\alpha}_{v_{\beta_1} \ldots v_{\beta_k}}$ should also vanish, since $\tilde{\Gamma}_{v_{\beta_1} \ldots v_{\beta_k}}$ already has a higher deg$(\times)$.  Furthermore, deg$(\times)$ in $\tilde{\Gamma}_{v_{\beta_1} \ldots v_{\beta_k}}$ is 1 higher than deg$(\times)$ in $\Gamma^{v_{\alpha}}_{v_{\beta_1} \ldots v_{\beta_k}}$, and thus $\Gamma^{v_\alpha}_{v_{\beta_1} \ldots v_{\beta_k}} \neq \tilde{\Gamma}^{v_\alpha}_{v_{\beta_1} \ldots v_{\beta_k}}$ if and only if deg$(\times) = \pd$ in $\Gamma^{v_\alpha}_{v_{\beta_a} \ldots v_{\beta_k}}$.  In this case, $\Gamma^{v_\alpha}_{v_{\beta_1} \ldots v_{\beta_k}}$ does not vanish, but $\tilde{\Gamma}_{v_{\beta_1} \ldots v_{\beta_k}}$ contains a $\times$-vertex of degree $\pd + 1$ and is identified with the null graph, which means $\tilde{\Gamma}^{v_\alpha}_{v_{\beta_1} \ldots v_{\beta_k}}$ is also null.  So the equality is violated if and only if deg$(\times)= \pd +1$ in $\tilde{\Gamma}_{v_{\beta_1} \ldots v_{\beta_k}}$.  We want to write this condition in terms of deg$(\times)$ in $\Gamma^{ \star }$.  Note that deg$(\times)= \pd +1$ in $\tilde{\Gamma}_{v_{\beta_1} \ldots v_{\beta_k}}$ if and only if deg$(\times)= \pd +1-b-(k- \nbsv )$ in $\Gamma_A$, and deg$(\times)= \pd -1-b-(k- \nbsv )$ in $\Gamma$.  Finally this implies deg$(\times)= \pd +1-b$ in $\Gamma^{ \star }$.  Since $1\leq b\leq \nbsv$, we have that $\pd \geq \text{deg} (\times) \geq \pd +1 - \nbsv$, which is condition (\ref{Condition 3}).  Moreover,\begin{equation*}
	\sum_{v_\alpha \neq v_A} \Gamma^{v_\alpha}_{v_{\beta_1} \ldots v_{\beta_k}} = \rho \left ( \Gamma_\text{spider} \right ),
\end{equation*}
\noindent where $\Gamma_\text{spider}$ is an $(\ell - 1)$-looped spider formed from $\Gamma_{v_{\beta_1} \ldots v_{\beta_k}}$ by adding a $\star$-vertex and then adding an edge that joins this $\star$-vertex and the $\times$-vertex.  By construction, $\Gamma_\text{spider}$ has one fewer $\times$-loop than $\Gamma^{ \star }$.  Such spiders form the desired $\text{L}^{(\ell-1)}_{\text{spider}}$ in \eqref{ordereqn}.  
\end{proof}

\subsubsection{Constructing Loopless \texorpdfstring{$\times$}{X}-Relations} \label{sec: pyramid}

	Constructing a basis for $\mathcal{R}^{ \times }_\text{loopless}$ requires a thorough examination of $\star$-graphs.  The next lemma shows that any $\times$-relation can be written as a derivative map acting on a linear combination of primary $\star$-graphs, which allows us to restrict to primary $\star$-graphs in classifying all loopless $\times \text{-relations}$.
	
\begin{Lem} \label{lem: primary generators}
	For any $L^{ \times } \in \mathcal{R}^{ \times }$, there exists $L^{ \star } \in \mathcal{L}^{ \star }$ that contains only primary $\star$-graphs, satisfying $L^{ \times } = \rho ( L^{ \star } )$.
\end{Lem} 
		
\begin{proof}
	Since $L^{ \times }$ is a $\times$-relation, there exists $\tilde{L}^{ \star } = \sum_i b_i \hspace{0.6mm} \Gamma^{ \star }_i \in \mathcal{L}^{ \star }$ such that $\rho ( \tilde{L}^{ \star } ) = L^{ \times }$.  Starting with $\Gamma^{ \star }_i$, one can follow steps \ref{Step 1 Def der map} and \ref{Step 2 Def der map} in Definition \ref{def: der map} to construct a series of $\star$-graphs, $(\Gamma^{ \star }_i)_{j_1}, (\Gamma^{ \star }_i)_{j_1 j_2}, \ldots (\Gamma^{ \star }_i)_{j_1 \ldots j_{k-1}}$, 	with $j_\alpha = 0, \ldots, n - 1$ and $\alpha = 1, \ldots, k$.  By construction, $(\Gamma^{ \star }_i)_{j_1 \ldots j_{k-1}}$ contains exactly one $\star$-vertex (which makes it a primary $\star$-graph), and $\rho ( \Gamma^{ \star }_i ) = \rho \left ( \sum_{j_1, \ldots, j_{k-1}=0}^{n-1} (\Gamma^{ \star }_i)_{j_1 \ldots j_{k-1}} \right )$.  Therefore, 
\begin{equation*}
	L^{ \times } = \sum_i b_i \cdot \rho ( \Gamma^{ \star }_i ) = \rho \left ( \sum_i \sum_{j_1, \ldots, j_{k-1}=0}^{\nv -1} b_i \left ( \Gamma^{ \star }_i \right )_{j_1 \ldots j_{k-1}} \right ).  
\end{equation*}
\noindent This linear combination of primary $\star$-graphs $(\Gamma^{ \star }_i)_{j_1 \ldots j_{k-1}}$ defines the desired $L^{ \star }$.  Operationally, such $L^{ \star }$ is constructed from $\tilde{L}^{ \star }$ by removing the $\star$-vertices one by one, as per the steps in the definition of $\rho$, until only one $\star$-vertex remains.
\end{proof}
	
\noindent The following proposition presents a general construction for loopless $\times$-relations.  In the next section we will prove that this procedure generates all elements in $\mathcal{R}^{ \times }_\text{loopless}$.
	
\begin{Prop} \label{prop: pyramid}
	Let $L^{ \star } \in \mathcal{L}^{ \star }$ be a linear combination of $\ell$-looped primary $\star$-graphs such that all $\times$-graphs in $\rho ( L^{ \star } )$ are $\ell$-looped.  There exists an $L^{ \star }_\ell \in \mathcal{L}^{ \star }$, such that
\begin{enumerate}[\emph{(}a\emph{)}]
	
\item \label{Statement 1}

	$L^{ \star }_\ell - L^{ \star }$ contains no graph with more than $\ell - 1$ loops;
	
\item \label{Statement 2}

	$\rho( L^{ \star }_\ell ) = (-1)^\ell \rho^{(0)} \circ \theta ( L^{ \star } + \text{\emph{L}}_\text{\emph{spider}} ) \in\mathcal{R}^{ \times }_\text{loopless}$.  $\text{\emph{L}}_\text{\emph{spider}}$ is a linear combination of spiders.
	
\end{enumerate}
\end{Prop}
	
\begin{proof}
	
	Since each graph in $\rho (L^{ \star })$ is $\ell$-looped, $\rho ( L^{ \star } ) = \rho^{(\ell)} ( L^{ \star } )$.  Let $\rho^{(\ell)} ( L^{ \star } ) \equiv \sum_{i=1}^{k} b_i \hspace{0.6mm} \Gamma^{ \times }_{i}$, where each $\Gamma^{ \times }_{i}$ contains $\ell$ loops.  For each $\Gamma^{ \times }_{i}$, label the loops from $1$ to $\ell$.   By Definition \ref{def: undo loop}, $\theta_1$ undoes the $1^{\mathrm{st}}$ loop, and $\theta_1 (\Gamma^{ \times }_{i})$ defines a primary $\star$-graph with $\ell - 1$ loops that is associated with $\Gamma^{ \times }_i$.  By Lemma \ref{lem: star star association}, $\Gamma^{ \times }_{i}$ drops out of $\rho ( L^{ \star } - b_i \cdot \theta_1 (\Gamma^{ \times }_{i} ) )$.  Moreover, all $\times$-graphs in $\rho \circ \theta_1 (\Gamma^{ \times }_{i})$, except for $\Gamma^{ \times }_{i}$, are $(\ell - 1)$-looped.  Then
\begin{equation*}
	L^{ \star }_1 \equiv L^{ \star } - \sum_{i=1}^{k} b_i \cdot \theta_1 ( \Gamma^{ \times }_{i} ) = L^{ \star } - \theta_1 \circ \rho^{(\ell)} (L^{ \star })
\end{equation*} 
\noindent defines an $L^{ \star }_1$ that satisfies $\rho ( L^{ \star }_1 ) = \rho^{(\ell-1)} (L^{ \star }_1)$.  Repeat this procedure for $L^{ \star }_1$ and the $2^{\mathrm{nd}}$ loop, in place of $L^{ \star }$ and the $1^{\mathrm{st}}$ loop, obtaining
\begin{equation*}
	L^{ \star }_2 \equiv L^{ \star }_1 - \theta_2 \circ \rho^{(\ell - 1)} (L^{ \star }_1) = L^{ \star } - \theta_1 \circ \rho^{(\ell)} (L^{ \star }) + \theta_2 \circ \rho^{(\ell - 1)} \circ \theta_1 \circ \rho^{(\ell)} (L^{ \star }).
\end{equation*}
\noindent The second equality holds because $\rho^{(\ell - 1)} ( L^{ \star } ) = 0$.  Furthermore, $\rho ( L^{ \star }_2 )=\rho^{(\ell-2)}( L^{ \star }_2 )$.  Iterating this $\ell$ times, we will reach a linear combination of primary $\star$-graphs
\begin{align} \label{eq: L ell star}
	& L^{ \star }_{\ell} \equiv L^{ \star } + \sum_{i=1}^\ell (-1)^i X^{ \star }_{i}, & X^{ \star }_i \equiv \theta_i \circ \rho^{(\ell - i + 1)} \circ \cdots \circ \theta_2  \circ \rho^{(\ell-1)} \circ \theta_1 \circ \rho^{(\ell)} (L^{ \star }).  
\end{align}
\noindent Here $\rho(L^{ \star }_\ell)=\rho^{(0)}(L^{ \star }_\ell) \in \mathcal{R}^{ \times }_\text{loopless}$.  Moreover, $X^{ \star }_i$ only contains $( \ell - i )$-looped graphs.  This means that graphs in $L^{ \star }_{\ell} - L^{ \star } = \sum_{i=1}^\ell (-1)^i X^{ \star }_i$ contain at most $\ell - 1$ loops.  Therefore, $L^{ \star }_\ell$ satisfies condition (\ref{Statement 1}) of the proposition.  

	To prove that $L^{ \star }_\ell$ also satisfies condition (\ref{Statement 2}), take $\text{L}^{(i)}_{\text{spider}}$ to stand for ``any linear combination of $i$-looped spiders" and, for $0< k\leq \ell$, define
\begin{align*} 
	Z^{ \star }_{\ell - k} \equiv \theta_{k} \circ \ldots \circ \theta_{1} \left ( L^{ \star } \right ) - \sum_{\alpha = 2}^{k} \theta_{k} \circ \ldots \circ \theta_{\alpha} \left ( L_\text{spider}^{(\ell - \alpha + 1)} \right ) - \text{L}^{(\ell - k)}_\text{spider}
\end{align*} 

\noindent which contains only $(\ell - k )$-looped graphs.  Define $Z^{ \star }_{\ell} \equiv L^{ \star }$.  Therefore, for $0\leq k\leq \ell$, applying Proposition \ref{prop: order},
\begin{align*}
	& \rho^{(\ell - k - 1)} \circ \theta_{k+1} \circ \rho^{(\ell - k)} \left ( Z^{ \star }_{\ell - k} \right )
	= \rho^{(\ell - k - 1)} \left [ \theta_{k+1} \left ( Z^{ \star }_{\ell - k} \right ) - \text{L}^{(\ell - k - 1)}_\text{spider} \right ] \notag \\
	= & \rho^{(\ell - k - 1)} \left [ \theta_{k+1} \circ \cdots \circ \theta_{1} \left ( L^{ \star } \right ) - \sum_{\alpha = 2}^{k+1} \theta_{k+1} \circ \ldots \circ \theta_{\alpha} \left ( L_\text{spider}^{(\ell - \alpha + 1)} \right ) - \text{L}^{(\ell - k - 1)}_\text{spider} \right ],
\end{align*}

\noindent i.e.,
\begin{align} \label{eq: induction k}
	\rho^{(\ell - k - 1)} \circ \theta_{k+1} \circ \rho^{(\ell - k)} \left ( Z^{ \star }_{\ell - k} \right ) = \rho^{(\ell - k - 1)} \left ( Z^{ \star }_{\ell - k - 1} \right ).
\end{align}

\noindent Note that $\rho^{(0)} ( L^{ \star } ) = 0$ and $\rho^{(0)} ( X^{ \star }_i ) = 0$ for $i = 1, \ldots, \ell - 1$.  Then, by \eqref{eq: L ell star} and  \eqref{eq: induction k},
\begin{eqnarray*} 
	\rho^{(0)} ( L^{ \star }_\ell ) & = & (-1)^\ell \rho^{(0)} (X^{ \star }_\ell) \notag \\
	& = & (-1)^\ell \rho^{(0)} \circ \theta_\ell \circ \rho^{(1)} \circ \cdots \circ \theta_2 \circ \rho^{(\ell - 1)} \circ \theta_1 \circ \rho^{(\ell)} (Z^{ \star }_\ell) \notag \\
	& = & (-1)^\ell \rho^{(0)} \circ \theta_\ell \circ \rho^{(1)} \circ \cdots \circ \theta_2 \circ \rho^{(\ell - 1)} (Z^{ \star }_{\ell - 1}) = \ldots = (-1)^\ell \rho^{(0)} \left ( Z^{ \star }_0 \right ) \notag \\
	& = & \notag (-1)^\ell \rho^{(0)} \left ( \theta_{\ell} \circ \ldots \circ \theta_{1} \left ( L^{ \star } \right ) - \sum_{\alpha = 2}^{\ell} \theta_{\ell} \circ \ldots \circ \theta_{\alpha} \left ( L_\text{spider}^{(\ell - \alpha + 1)} \right ) - \text{L}^{(0)}_\text{spider} \right ) \notag \\
	& = & (-1)^\ell \rho^{(0)} \circ \theta \left ( L^{ \star } - \sum_{\alpha=2}^{\ell + 1} \text{L}^{(\ell - \alpha + 1)}_{\text{spider}} \right ).
\end{eqnarray*}

\noindent Hence,
\begin{equation*} \label{eq: final}
	\rho ( L^{ \star }_\ell ) = \rho^{(0)} ( L^{ \star }_\ell ) = (-1)^\ell \rho^{(0)} \circ \theta ( L^{ \star } + \text{L}_\text{spider} ) \in \mathcal{R}^{ \times }_\text{loopless}.
\end{equation*}

\noindent This gives condition (\ref{Statement 2}), with $\text{L}_\text{spider}=-\sum_{\alpha = 2}^{\ell+1} \text{L}^{(\ell - \alpha + 1)}_{\text{spider}}$.
\end{proof}

\begin{Cor} \label{prop: spider pyramid}
	Given $L_s$, a linear combination of spiders, ${\rho^{(0)} \circ \theta ( L_s ) \in \mathcal{R}^{ \times }_\text{loopless}}$.  
\end{Cor}

\begin{proof}
	It is enough to show that this corollary is true for one spider $S$.  We claim that Proposition $\ref{prop: pyramid}$ holds for $L^{ \star } = S$ and where $\text{L}_{\text{spider}}$ is null if we order the loops such that loops 1 to $\ell^{ \times }$ are $\times$-loops and $\ell^{ \times }+1$ to $\ell$ are $\bullet$-loops and the loops are removed in this order.  Define $Z^{ \star }_{\ell-k} \equiv \theta_{k} \circ \ldots \circ \theta_{1} \left ( S \right )$ and $Z^{ \star }_{\ell} \equiv  S $.  As in the proof of Proposition $\ref{prop: pyramid}$, we are done if we can prove \eqref{eq: induction k}, but with this new definition of $Z^{ \star }_{\ell - k}$ (i.e., when $\text{L}_{\text{spider}}$ is always taken to be zero).
	
	For $k=0$, $Z^{ \star }_{\ell} = S$, a spider, which has no $\star$-vertex adjacent to a $\bullet$-vertex in violation of Proposition $\ref{prop: order}$($\ref{Condition 2}$).  Thus, $\rho^{(\ell - 1 )} \circ \theta_1 \circ \rho^{( \ell )} ( Z^{ \star }_{\ell} ) = \rho^{( \ell - 1 )} \circ \theta_1 ( Z^{ \star } _{\ell}) = \rho^{( \ell - 1 )} ( Z^{ \star }_{\ell - 1} )$.  This holds regardless of which loop is chosen to be undone first.  However, if the first loop is a $\times$-loop, then $Z^{ \star }_{\ell -1} = \theta_1 (S )$ will continue to violate Proposition $\ref{prop: order}$($\ref{Condition 2}$).  Therefore, if all of the $\times$-loops are undone first, then $\rho^{( \ell - k - 1)} \circ \theta_{k+1} \circ \rho^{\ell - k} ( Z^{ \star }_{\ell - k} ) = \rho^{( \ell - k -1 )} ( Z^{ \star }_{\ell - k - 1} )$ holds for $0 \leq k \leq \ell^{ \times }$.  
	
	Now, there are no longer any $\times$-loops.  Whenever there are $\bullet$-loops, $Z^{ \star }_{\ell - \ell^{ \times }}$ will violate Proposition $\ref{prop: order}$($\ref{Condition 1}$).  Thus, $\rho^{( \ell - k - 1)} \circ \theta_{k+1} \circ \rho^{\ell - k} ( Z^{ \star }_{\ell - k} ) = \rho^{( \ell - k -1 )} ( Z^{ \star }_{\ell - k - 1} )$ continues to hold all the way until $k = \ell$.
\end{proof}

\subsubsection{A Basis for \texorpdfstring{$\mathcal{R}^{ \times }_\text{loopless}$}{RX loopless}} \label{sec: Basis of Relations} 

\noindent To find a basis for $\mathcal{R}^{ \times }_\text{loopless}$, we first show that any loopless $\times$-relation can be written as $\rho^{(0)} \circ \theta$ acting on a linear combination of spiders.  We start with the following lemmas:

\begin{Lem} \label{lem: primary association}
	Let $L^{ \times }\in\mathcal{R}^{ \times }_\text{loopless}$ satisfy $L^{ \times } = \rho ( L^{ \star } )$, with $L^{ \star }$ a linear combination of primary $\star$-graphs.  For any $\Gamma^{ \star }_A$ in $L^{ \star }$ that is associated with a looped $\times$-graph $\Gamma^{ \times }$, there exists another $\star$-graph $\Gamma^{ \star }_B\neq\Gamma^{ \star }_A$ in $L^{ \star }$, such that $\Gamma^{ \times }$ is not contained in $\rho ( \Gamma^{ \star }_A - \Gamma^{ \star }_B )$.
\end{Lem}

\begin{proof}
    Suppose $\Gamma^{ \star }_A$ is the only $\star$-graph in $L^{ \star }$ that is associated with $\Gamma^{ \star }_A$.  Assume that the coefficient of $\Gamma^{ \star }_A$ in $L^{ \star }$ is $b_A \neq 0$.  Therefore, none of the $\star$-graphs in $L^{ \star } - b_A \Gamma^{ \star }_A$ is associated with $\Gamma^{ \times }$, and thus $\Gamma^{ \times }$ is not contained in $\rho ( L^{ \star } - b_A \Gamma^{ \star }_A)$.  Hence, the looped $\times$-graph $\Gamma^{ \times }$ appears in $\rho (L^{ \star }) = \rho (L^{ \star } - b_A \Gamma^{ \star }_A) + b_A \cdot \rho (\Gamma^{ \star }_A)$ with coefficient $b_A \neq 0$, which contradicts the fact that $\rho (L^{ \star }) \in \mathcal{R}^{ \times }_\text{loopless}$.
    
	The above argument shows that there exists a $\Gamma_{B}^{ \star } \neq \Gamma^{ \star }_A$ in $L^{ \star }$ that is associated with $\Gamma^{ \times }$.  By Lemma \ref{lem: star star association}, $\Gamma^{ \times }$ appears in both $\rho (\Gamma^{ \star }_A)$ and $\rho (\Gamma_{B}^{ \star })$ with coefficient 1.  Then, $\rho ( \Gamma^{ \star }_A - \Gamma^{ \star }_B )$ does not contain $\Gamma^{ \times }$.  
\end{proof}

\noindent Specifically, if $\Gamma^{ \star }_A$ is $\ell$-looped and $\Gamma^{\times}$ is $(\ell + 1)$-looped, then $\Gamma^{ \star }_B$ is also $\ell$-looped.

\begin{Lem} \label{cor: nonspider}
	If $\Gamma^{ \star }_A, \Gamma^{ \star }_B \in \mathcal{G}^{ \star }$ are $\ell$-looped primary $\star$-graphs that are associated with the same $(\ell+1)$-looped $\times$-graph, then $\theta ( \Gamma^{ \star }_A - \Gamma^{ \star }_B ) = 0$.
\end{Lem}

\begin{proof}
	Since $\Gamma^{ \star }_A$ and $\Gamma^{ \star }_B$ are both associated with the same $(\ell+1)$-looped $\times$-graph, $\Gamma^{ \times }$, $\theta ( \Gamma^{ \star }_A ) = \theta ( \Gamma^{ \star }_B ) = \theta ( \Gamma^{ \times } )$.  Therefore, $\theta ( \Gamma^{ \star }_A - \Gamma^{ \star }_B ) = 0$.
\end{proof}

\begin{Prop} \label{prop: spider basis}
	For any loopless $\times$-relation $L^\times \in \mathcal{R}^\times_\text{loopless}$, there exists a linear combination of spiders $L_s$, such that $L^{ \times } = \rho^{(0)} \circ \theta ( L_s )$.
\end{Prop}

\begin{proof}
	In the following, we take $\text{L}_{\text{spider}}$ to stand for ``any linear combination of spiders".  Since $L^{ \times }$ is a $\times$-relation, there exists $L^{ \star } \in \mathcal{L}^{ \star }$, consisting of primary $\star$-graphs, such that $L^{ \times } = \rho ( L^{ \star } )$.  Take the set, $\mathcal{H}_\ell = \{ \Gamma^{ \star }_1, \ldots, \Gamma^{ \star }_{H_\ell} \}$, of $\star$-graphs in $L^{ \star }$ that contain the highest number, $\ell$, of loops.  Let $b^i_\ell$ be the coefficient of $\Gamma^{ \star }_i$ in $L^{ \star }$.  Therefore, $L^{ \star } - \sum_{i=1}^{H_\ell} b^i_\ell \hspace{0.3mm} \Gamma^{ \star }_i$ contains no graphs with more than $\ell - 1$ loops.  We implement the following procedure for all graphs in $\mathcal{H}_\ell$ in order from $\Gamma^{ \star }_{1}$ to $\Gamma^{ \star }_{H_{\ell}}$:
\begin{enumerate}

\item \label{Step 1}

	Define $\beta_\ell^{(1)} \equiv b_\ell^{(1)}$.  Apply to $\Gamma^{ \star }_1$ the construction outlined in Proposition \ref{prop: pyramid}:
\begin{enumerate}

\item 

	If $\Gamma^{ \star }_1$ is a spider: By Corollary \ref{prop: spider pyramid}, Proposition \ref{prop: pyramid}(\ref{Statement 2}) becomes that there exists a linear combination of primary $\star$-graphs $L^{(1)}_\ell$, such that
\begin{equation*}
	\rho \bigl( L^{(1)}_\ell \bigr) = \rho^{(0)} \circ \theta \left ( \Gamma^{ \star }_1 \right ) = \rho^{(0)} \circ \theta \left ( \text{L}_\text{spider} \right ) \in \mathcal{R}^{ \times }_\text{loopless}.
\end{equation*} 
\noindent By Proposition \ref{prop: pyramid}(\ref{Statement 1}), graphs in $L^{(1)}_\ell - \Gamma^{ \star }_1$ contain at most $\ell - 1$ loops.  
	
\item
    
	If $\Gamma^{ \star }_1$ is not a spider: $\rho ( \Gamma^{ \star }_1 )$ contains an $(\ell + 1)$-looped $\times$-graph $\Gamma^{ \times }$.  By Lemma \ref{lem: primary association}, there exists an $\ell$-looped $\star$-graph $\Gamma^{ \star }_j \in \mathcal{H}$, $\Gamma^{ \star }_j \neq \Gamma^{ \star }_1$, that is associated with $\Gamma^{ \times }$ and $\Gamma^{ \times }$ is not in $\rho ( \Gamma^{ \star }_1 - \Gamma^{ \star }_2 )$.  By Lemma \ref{cor: nonspider}, $\theta ( \Gamma^{ \star }_1 - \Gamma^{ \star }_j ) = 0$.  By Proposition \ref{prop: pyramid},
\begin{equation*}
	\rho \bigl( L^{(1)}_\ell \bigr) = \rho^{(0)} \circ \theta \left ( \text{L}_\text{spider} \right ) \in \mathcal{R}^{ \times }_\text{loopless},
\end{equation*} 
\noindent where graphs in $L^{(1)}_\ell - ( \Gamma^{ \star }_1 - \Gamma^{ \star }_j )$ contain at most $\ell - 1$ loops.
    
\end{enumerate}    

\item

    Define $\beta_\ell^{(i)}$, for $i>1$, as the coefficient (which may be zero) of $\Gamma^{ \star }_i$ in $L^{ \star } - \beta_\ell^{(1)} L^{(1)}_\ell$.  
    
\item \label{Step 3}

	If $\beta_\ell^{(2)} = 0$, skip this step; if not, repeat step \ref{Step 1} for $\Gamma^{ \star }_2$, resulting in an $L^{(2)}_\ell$ with
\begin{equation*}
	\rho \big ( L_\ell^{(2)} \big ) = \rho^{(0)} \circ \theta (\text{L}_\text{spider}) \in \mathcal{R}^{ \times }_\text{loopless}.
\end{equation*}
	
 Redefine $\beta_\ell^{(i)}$, for $i>2$, to be the coefficient of $\Gamma^{ \star }_i$ in $\big ( L^{ \star } - \beta_{\ell}^{(1)} L^{(1)}_\ell \big ) - \beta_{\ell}^{(2)} L^{(2)}_\ell$.  
	
\item

	Repeat step \ref{Step 3} for $\Gamma_\ell^{(3)}, \ldots, \Gamma_\ell^{(H_\ell)}$ in sequence.  This will eventually generate a linear combination of $\star$-graphs, $L^{ \star } - \sum_{i=1}^{H_\ell} \beta^{(i)}_\ell \hspace{0.3mm} L^{(i)}_\ell$, where 
\begin{equation*}
	\rho \left ( \sum_{i=1}^{H_\ell} \beta^{(i)}_\ell \hspace{0.3mm} L^{(i)}_\ell \right ) = \rho^{(0)} \circ \theta (\text{L}_\text{spider}) \in \mathcal{R}^{ \times }_\text{loopless},
\end{equation*}
\noindent and all graphs in $L^{ \star } - \sum_{i=1}^{H_\ell} \beta^{(i)}_\ell L^{(i)}_\ell$ contain at most $\ell - 1$ loops.
\end{enumerate}

\noindent We can now repeat this procedure for $L^{ \star } - \sum_{i=1}^{H_\ell} \beta^{(i)}_\ell L^{(i)}_\ell$.  We will obtain
\begin{align*}
	& L^{ \star } - \sum_{i=1}^{H_\ell} \beta^{(i)}_\ell \hspace{0.3mm} L^{(i)}_\ell - \sum_{i=1}^{H_{\ell-1}} \beta^{(i)}_{\ell-1} \hspace{0.3mm} L^{(i)}_{\ell-1}, & \rho \left ( \sum_{i=1}^{H_{\ell-1}} \beta^i_\ell \hspace{0.3mm} L^{(i)}_{\ell-1} \right ) = \rho^{(0)} \circ \theta (\text{L}_\text{spider}) \in \mathcal{R}^{ \times }_\text{loopless}.
\end{align*}
\noindent In the first expression graphs contain at most $\ell - 2$ loops.  Iterate $\ell$ times to get
\begin{equation*}
	\tilde{L}^{ \star } \equiv L^{ \star } - \sum_{\alpha=1}^\ell \sum_{i=1}^{H_\alpha} \beta^i_\alpha \hspace{0.6mm} L^{(i)}_\alpha, 
\end{equation*}
\noindent which contains only loopless primary $\star$-graphs.  In addition,
\begin{equation*}
	\rho \left ( \sum_{\alpha=1}^\ell \sum_{i=1}^{H_\alpha} \beta^i_\alpha \hspace{0.6mm} L^{(i)}_\alpha \right ) = \rho^{(0)} \circ \theta ( \text{L}_\text{spider} ) \in \mathcal{R}^{ \times }_\text{loopless}.
\end{equation*}
\noindent Therefore, 
\begin{equation*}
	 \rho ( \tilde{L}^{ \star } ) = L^{ \times } - \rho \left ( \sum_{\alpha=1}^\ell \sum_{i=1}^{H_\alpha} \beta^i_\alpha \hspace{0.6mm} L^{(i)}_\alpha \right ) \in \mathcal{R}^{ \times }_\text{loopless}.
\end{equation*}
    
	Next we show that $\tilde{L}^{ \star }$ is a linear combination of spiders.  Suppose there exists a $\star$-graph $\Gamma^{ \star }_A$ in $\tilde{L}^{ \star }$ that is not a spider.  Then, by Lemma \ref{lem: primary association}, there should exist another $\star$-graph $\Gamma^{ \star }_B \neq \Gamma^{ \star }_A$ in $\tilde{L}^{ \star }$ that is associated with the 1-looped $\times$-graph $\Gamma^{ \times }$ in $\rho (\Gamma^{ \star }_A)$.  But since $\Gamma^{ \times }$ is associated with a unique loopless $\star$-graph (resulting from undoing the loop), this is impossible.  Therefore, graphs in $\tilde{L}^{ \star }$ are loopless spiders, and thus $\rho (\tilde{L}^{ \star }) = \rho^{(0)} \circ \theta (\tilde{L}^{ \star })$.  Hence, 
\begin{equation*}
    L^{ \times } = \rho \left (\tilde{L}^{ \star } + \sum_{\alpha=1}^\ell \sum_{i=1}^{H_\alpha} \beta^i_\alpha \hspace{0.6mm} L^{(i)}_\alpha \right ) = \rho^{(0)} \circ \theta \left ( \text{L}_\text{spider} \right).
\end{equation*}
\noindent The $\text{L}_\text{spider}$ within the last pair of parentheses is the desired $L_s$.
\end{proof}

\noindent Thus we have shown that any loopless $\times$-relation can be written as $\rho^{(0)} \circ \theta ( L_s )$.  In fact, we can go further and show that it is equal to $\rho^{(0)} ( L_M )$, where $L_M$ is a linear combination of Medusas.  We start with the following Lemma:
\begin{Lem} \label{lem: lem}
	Given a spider $S$ that contains $\ell^{ \times }$ $\times$-loops, there exist two linear combinations of spiders $Y^{ \star }$, with graphs containing $\ell^{ \times }$ $\times$-loops but no $\bullet$-loop, and $W^{ \star }$, with graphs containing fewer than $\ell^{ \times }$ $\times$-loops, such that $\rho^{(0)} \circ \theta \left ( Y^{ \star } \right ) = \rho^{(0)} \circ \theta \bigl( S + W^{ \star } \bigr)$.
\end{Lem}

\begin{proof}
	Suppose $S$ contains $\ell^{ \bullet }$ $\bullet$-loops, labeled from $1$ to $\ell^{ \bullet }$.  Denote the total number of loops in $S$ to be $\ell = \ell^{ \times } + \ell^{ \bullet }$.  Label the $\times$-loops in $S$ from $\ell^{ \bullet } + 1$ to $\ell$.  We will follow the proof of Proposition \ref{prop: pyramid} to undo these $\bullet$-loops.  However, we want to keep the $\star$-vertex in $S$ untouched.  Therefore, define $\rho_s$ to be the usual derivative map except that it keeps the original $\star$-vertex in $S$ untouched; we can grade $\rho_s$ by number of loops in analogy with $\rho$.
	
    Apply $\theta_1$ to $S$ to undo the $1^{\mathrm{st}}$ $\bullet$-loop.  Define $Y^{ \star }_1 \equiv S - \rho_s \circ \theta_1 ( S )$ and note that $S$ drops out of $Y^{ \star }_1$.  Furthermore, since $S$ is the only $\ell$-looped graph in $\rho_s \circ \theta ( S )$, 
\begin{equation*}
	Y^{ \star }_1 = S - \rho_s \circ \theta_1 ( S ) = - \rho_s^{(\ell - 1)} \circ \theta_1 ( S ).
\end{equation*} 
\noindent Repeat this procedure for all graphs in $Y^{ \star }_1$ and the $2^{\mathrm{nd}}$ $\bullet$-loop, in place of $S$ and the $1^{\mathrm{st}}$ $\bullet$-loop, obtaining $	Y^{ \star }_2 \equiv Y^{ \star }_1 - \rho_s \circ \theta_2 ( Y^{ \star }_1 )$, which is a linear combination of $(\ell - 2)$-looped graphs.  However, since all graphs in $Y^{ \star }_1$ are $( \ell - 1 )$-looped, we have
\begin{equation*}
	Y^{ \star }_2 = (-1)^2 \rho_s^{(\ell - 2)} \circ \theta_2 \circ \rho_s^{(\ell - 1)} \circ \theta_1 (S).
\end{equation*}
\noindent Iterate this $\ell^{ \bullet }$ times, resulting in 
\begin{equation*}
	Y^{ \star }_{\ell^{ \bullet }} = (-1)^{\ell^{ \bullet }} \rho_s^{(\ell-\ell^{ \bullet })} \circ \theta_{\ell^{ \bullet }} \circ  \cdots \rho_s^{(\ell - 2)} \circ \theta_2 \circ \rho_s^{(\ell - 1)} \circ \theta_1 ( S ).
\end{equation*}
\noindent Graphs in $Y^{ \star }_{\ell^{ \bullet }}$ contain no $\bullet$-loops.  As in the derivation of \eqref{eq: final} in Proposition \ref{prop: pyramid},
\begin{align} \label{eq: star pyramid}
	Y^{ \star }_{\ell^{ \bullet }} = & (-1)^{\ell^{ \bullet }} \rho_s^{(\ell-\ell^{ \bullet })} \left ( \theta_{\ell^{ \bullet }} \circ \cdots \circ \theta_1 ( S ) - \sum^{\ell^{ \bullet }}_{\alpha = 2} \theta_{\ell^{ \bullet }} \circ \cdots \circ \theta_\alpha \left ( \text{L}^{(\ell - \alpha + 1)}_{\text{spider}} \right ) - \text{L}^{(\ell-\ell^{ \bullet })}_\text{spider} \right )
\end{align} 
\noindent By Proposition \ref{prop: order}, each graph in $\text{L}^{(\ell - \alpha + 1)}_{\text{spider}}$ and $\text{L}^{(\ell - \ell^{\bullet})}_\text{spider}$ contains fewer than $\ell^{ \times }$ $\times$-loops.  

    Take $Y^{ \star } = (-1)^{\ell^{ \bullet }} Y^{ \star }_{\ell^{{}^{\bullet}}}$.  Then, by Proposition \ref{prop: order}, 
\begin{align*}
    \rho^{(0)} \circ \theta ( Y^{ \star } ) &= \rho^{(0)} \circ \theta_{\ell} \circ \cdots \circ \theta_{\ell^{\bullet} + 1} ( Y^{ \star } ) = \rho^{(0)} \circ \theta ( S + W^{ \star } ),
\end{align*}
\noindent where $W^{ \star }$ is a linear combination of spiders containing fewer than $\ell^{ \times }$ $\times$-loops.  
\end{proof}

\noindent The next proposition finally allows us to relate spiders and Medusas.

\begin{Prop} \label{prop: star pyramid}
	For any linear combination of spiders $L_s \in \mathcal{L}^{ \star }$, there exists a linear combination of Medusas $L_M$, such that $\rho^{(0)} \circ \theta ( L_s ) = \rho^{(0)} ( L_M )$.  
\end{Prop}

\begin{proof}
	Take the set, 
\begin{equation*}
	\mathcal{H_{\ell^{ \times }}} = \left \{ S^{(\ell^{ \times })}_{1}, \ldots, S^{(\ell^{ \times })}_{H_{\ell^{ \times }}} \right \},
\end{equation*}
\noindent of spiders in $L_s$ that contain the highest number, $\ell^{ \times }$, of $\times$-loops.  Denote the coefficient of $S^{(\ell^{ \times })}_i$ in $L_s$ as $b_i^{(\ell^{ \times })}$.  Therefore, graphs in $L_s - \sum_{i = 1}^{H_{\ell^{ \times }}} b_i^{(\ell^{ \times })} S^{(\ell^{ \times })}_i$ contain at most $\ell^{ \times } - 1$ $\times$-loops.  Applying Lemma \ref{lem: lem} to each $S^{(\ell^{ \times })}_i$ yields two linear combinations of spiders $Y^{(\ell^{ \times })}$, comprised of graphs containing $\ell^{ \times }$ $\times$-loops but no $\bullet$-loop, and $W^{(\ell^{ \times })}$, comprised of graphs containing fewer than $\ell^{ \times }$ $\times$-loops, such that
\begin{equation} \label{eq: no proof}
	\sum_{i = 1}^{H_{\ell^{ \times }}} \rho^{(0)} \circ \theta \left ( b_i^{(\ell^{ \times })} S^{(\ell^{ \times })}_i \right ) = \rho^{(0)} \circ \theta \left ( Y^{(\ell^{ \times })} - W^{(\ell^{ \times })} \right ).
\end{equation} 
\noindent Furthermore,
\begin{equation} \label{eq: ell - 1}
	L_s - \left ( \sum_{i = 1}^{H_{\ell^{ \times }}} b_i^{(\ell^{ \times })} S^{(\ell^{ \times })}_i + W^{(\ell^{ \times })} \right )
\end{equation}
\noindent is a linear combination of spiders that contain at most $\ell^{ \times } - 1$ $\times$-loops.  Replace $L_s$ with \eqref{eq: ell - 1}, and the above procedure applies, resulting in a linear combination of spiders that contain at most $\ell^{ \times } - 2$ $\times$-loops.  Iterating $\ell^{ \times }$ times, we will obtain
\begin{equation} \label{eq: L s (0)}
	L_s^{(0)} \equiv L_s - \sum_{\alpha = 1}^{\ell^{ \times }} \left ( \sum_{i = 1}^{H_\alpha} b_i^{(\alpha)} \hspace{0.3mm} S^{(\alpha)}_i + W^{(\alpha)} \right ),
\end{equation}
\noindent By construction, $L_s^{(0)}$ is a linear combination of spiders containing no $\times$-loop.  Moreover,
\begin{equation} \label{eq: alpha}
	\rho^{(0)} \circ \theta \left ( Y^{(\alpha)} \right ) = \rho^{(0)} \circ \theta \left ( \sum_{i = 1}^{H_{\alpha}} b_i^{(\alpha)} S^{(\alpha)}_i + W^{(\alpha)} \right ),
\end{equation}
\noindent in analogy with \eqref{eq: no proof}.  Finally, repeat the same procedure one last time with $L^{(0)}_s$ in \eqref{eq: L s (0)} in place of $L_s$.  Since there are no longer any $\times$-loops in $L_{s}^{(0)}$, no $W^{(\alpha)}$'s will arise.  We obtain
\begin{align} \label{eq: 0}
	& L^{(0)}_s = \sum_{i = 1}^{H_{0}} b_i^{(0)} S^{(0)}_i, & \rho^{(0)} \circ \theta \left ( Y^{(0)} \right ) = \rho^{(0)} \circ \theta \left ( \sum_{i = 1}^{H_{0}} b_i^{(0)} S^{(0)}_i \right ).
\end{align}
\noindent Combining \eqref{eq: L s (0)}, \eqref{eq: alpha} and \eqref{eq: 0}, we obtain
\begin{equation*}
	\rho^{(0)} \circ \theta \left ( L_s \right ) = \rho^{(0)} \circ \theta \left ( \sum_{\alpha = 0}^{\ell^{ \times }} Y^{(\alpha)} \right ).
\end{equation*}
\noindent Since the $Y^{(\alpha)}$'s are linear combinations of spiders with no $\bullet$-loops and $\alpha$ $\times$-loops, $\theta(Y^{(\alpha)})$, is a linear combination of loopless $\star$-graphs with deg$(\times) = \pd - \alpha$ and $\nsv = \alpha + 1$.  This means deg$(\times) + \nsv = \pd + 1$ in these loopless $\star$-graphs.  By Definition \ref{def: Medusa}, such graphs are Medusas.  Therefore, 
\begin{equation*}
	L_M = \theta \left ( \sum_{\alpha = 0}^{\ell^{ \times }} Y^{(\alpha)} \right )
\end{equation*}
\noindent gives the desired linear combination of Medusas in $\rho^{(0)} \circ \theta \left ( L_s \right ) = \rho^{(0)} \left ( L_M \right )$.
\end{proof}

\begin{Thm} \label{thm: loopless times basis}
	$\rho^{(0)} \left ( \mathcal{M} \right )$ forms a basis of $\mathcal{R}^{ \times }_\text{\emph{loopless}}$.  
\end{Thm}

\begin{proof}
    Proposition \ref{prop: spider basis} states that any $L^{ \times } \in \mathcal{R}^{ \times }_{\text{loopless}}$ can be written as $\rho^{(0)} \circ \theta (L_s)$, with $L_s$ a linear combination of spiders.  Proposition \ref{prop: star pyramid} states that there exists a linear combination of Medusas $L_M$, such that $\rho^{(0)} \circ \theta (L_s) = \rho^{(0)} ( L_M )$.  Therefore, $L^{ \times } = \rho^{(0)} ( L_M )$.  This proves the completeness.

   Suppose that there exists a linear combination of Medusas $L_M = \sum_{i = 1}^k \alpha_i M_i \neq 0$, such that $\rho^{(0)} ( L_M ) = 0$.  Note that if two Medusas $M_A$ and $M_B$ have $\times$-vertices of different degrees, then $\rho^{(0)}(M_A)$ and $\rho^{(0)}(M_B)$ do not have $\times$-graphs in common.  Therefore, without loss of generality, we can assume that the $\times$-vertices in all $M_i$ have the same degree.  By the definition of Medusas, $\pd = \text{deg}(\times) + \nsv - 1$, and this implies that they also have the same number of $\star$-vertices.  Furthermore, we can assume that, after deleting the $\times$-vertex, all $M_i$'s are identical since this must be the case in order for the $\rho^{(0)} ( M_i )$ to cancel each other.  Therefore, the only differences between the $M_i$ are in the edges incident to the $\times$-vertex.
   
For any $\bullet$-vertex, $v$, take the set $\mathcal{H} = \{ \tilde{M}_1, \ldots, \tilde{M}_{H} \}$ of distinct Medusas in $L_{M}$, such that each $\tilde{M}_i$ contains the highest number, $E$, of edges that join $v$ and the $\times$-vertex among all Medusas in $L_M$.  For each $\tilde{M}_i$, form a specific graph, $\Gamma_i^{ \times }$, in $\rho^{(0)} ( \tilde{M}_i )$ by deleting all $\nsv$ $\star$-vertices and adding $\nsv$ edges joining the $\times$-vertex and $v$.  Now, $\Gamma^{ \times }_i$ contains $E + \nsv$ edges joining $v$ and the $\times$-vertex.  By construction, the $\Gamma^{ \times }_i$ contain the highest number of edges joining $v$ and the $\times$-vertex, among all $\times$-graphs in $\rho^{(0)} ( L_M )$.  Therefore, the $\Gamma_i^{ \times }$ can only be canceled among $\rho^{(0)} (\tilde{M}_i)$'s.  Thus in order for $\Gamma^{ \times }_i$ to be canceled, we need $\Gamma^{ \times }_i=\Gamma^{ \times }_j$ for some $i\neq j$ (note that this $\Gamma^{ \times }_{i}$ appears in $\rho^{(0)} ( \tilde{M}_i )$ with unit coefficient).  But this means $\tilde{M}_i=\tilde{M}_j$ which is a contradiction since the $\tilde{M}_i$ are distinct.  This proves independence.
\end{proof}

We can use this result to systematically find any $\pd$-invariant.  When $\pd =0$ or $\pd =1$, a simple classification is possible for any $\nv$.  The case $\pd =1$ is studied in detail in Appendix \ref{sec: linear} and $\pd =0$ is dealt with in the following corollary.
\begin{Cor}
Any loopless $0$-invariant is an exact $0$-invariant.
\end{Cor}
\begin{proof}
There are no Medusas in $\pd =0$ (since Medusas require $\nsv \geq 1$ and $\nsv \leq \text{deg}(\times)= \pd +1- \nsv \leq \pd$).  Therefore $\mathcal{R}^{ \times }_\text{loopless}=\{0\}$.
\end{proof}
\noindent Since we have already identified all exact invariants in Corollary \ref{lem: exact}, this classifies all loopless $0$-invariants.

\subsubsection{Lower Bound on Vertex Degree} \label{sec: lower bound}

	Within the loopless basis, we will set a lower bound on the degree of any $\bullet$- or $\times$-vertex in any graph appearing in a consistency equation for a $\pd$-invariant.  
	
\begin{Prop} \label{prop: min plain}
	The vertices of a plain-graph $\Gamma$ in a $\pd$-invariant $L \in \mathcal{L}_\text{loopless}$ are of degree no less than $\frac{1}{2} ( \pd +1 )$.  
\end{Prop}

\begin{proof}
	Let $v$ be a vertex in $\Gamma$ of degree less than $\frac{1}{2} ( \pd +1)$.  Let $\Gamma^{ \times }$ in $\delta (\Gamma)$ be the term given by replacing $v$ with a $\times$-vertex.  By Corollary \ref{cor: a}, $\Gamma^{ \times }$ is in $\delta (L)$.  Since $L$ is $\pd$-invariant, by Theorem $\ref{thm: loopless times basis}$, $\delta ( L ) = \rho^{(0)} \left( \sum_i \alpha_i M_{i} \right)$ for $M_{i} \in \mathcal{M}$ and thus, $\Gamma^{ \times }$ is in $\rho^{(0)} ( M_{j} )$ for some $j$.  Since $M_{j}$ is a Medusa, the $\times$-vertex in $M_{j}$ is of degree no less than $\frac{1}{2} ( \pd +1)$.  This contradicts the fact that the $\times$-vertex in $\Gamma^{ \times }$ is of degree less than $\frac{1}{2} ( \pd +1)$.
\end{proof}

\begin{Prop} \label{prop: min star}
	If $L \in \mathcal{L}_\text{loopless}$ is a $\pd$-invariant and $L_M$ is a linear combination of Medusas, such that $\delta L = \rho^{(0)} ( L_M )$, then the vertices of any Medusa in $L_M$ are of degree no less than $\frac{1}{2} ( \pd +1)$.  
\end{Prop}

\begin{proof} Suppose that a Medusa $M_0$ in $L_M$ has a vertex, $v$, of degree lower than $\frac{1}{2} ( \pd +1)$.  For an interaction term, there exists at least one other vertex $v_0$ in $M_0$ to which the $\star$-vertices can be joined such that $v_0$ is neither $v$ nor the $\times$-vertex.  The resulting graph in $\rho^{(0)} ( M_0)$ is a $\times$-graph $\Gamma_0^{ \times }$ with vertex $v$ of degree lower than $\frac{1}{2} ( \pd +1)$, and thus, by Proposition $\ref{prop: min plain}$, does not appear in $\delta ( L )$.  Therefore, there must exist another Medusa $M_1$ in $L_M$ such that $M_0 \neq M_1$ and $\Gamma_0^{ \times }$ is absent from $\rho^{(0)} ( M_0 - M_1 )$.  Therefore, $M_1$ must produce $\Gamma_0^{ \times }$ after deleting the $\star$-vertices and adding the same number, $\nsv$, of edges joining the $\times$-vertex and the other vertices.  Since $M_0 \neq M_1$ at least one of these other vertices is not $v_0$, so the degree of $v_0$ is larger in $M_1$ than in $M_0$.  Now, form another $\times$-graph, $\Gamma_1^{ \times }$, in $\rho^{(0)} ( M_1 )$ by deleting the $\star$-vertices in $M_1$ and adding the same number of edges joining the $\times$-vertex and $v_0$.  $\Gamma_1^{ \times }$ again contains $v$ with degree lower than $\frac{1}{2} ( \pd +1)$.  This $\times$-graph must be canceled by introducing a third Medusa $M_2$.  This procedure can be iterated to obtain an infinite sequence of Medusas $M_0, M_1, M_2, \dots$ in $L_M$, with the number of edges incident to $v_0$ in $M_i$ monotonically increasing with $i$.  But this is impossible since the Medusas have a fixed finite number of edges and thus we have a contradiction.  
\end{proof}

\subsection{Linear Shift Symmetry, Trees and Galileons} \label{sec: linear} 

	For $\pd =1$, Medusas have a very limited configuration: A Medusa $M$ has two subgraphs, which are disconnected from each other, one of which is
$
\begin{minipage}{1.2cm}
\begin{tikzpicture}
	\draw [thick] (0,0) -- (0.6,0);
        	\node at (0,0) {\scalebox{0.8}{$\bigstar$}};
        	\filldraw [white] (0.6,0) circle [radius=0.115];
	\draw [thick] (0.6,0) circle [radius=0.115];
	\node at (0.6,0) {$\times$};
\end{tikzpicture}
\end{minipage}
$, and the other of which is a loopless plain-graph.  This strongly restricts the possible associations between graphs.  In addition, for $\pd =1$, $\rho (M ) = \rho^{(0)} ( M )$ for any Medusa $M$.
\begin{Prop} \label{lem: 2}
	For $\pd =1$, a loopless $\times$-graph is associated with at most one Medusa.  
\end{Prop}

\begin{proof}
	If the $\times$-vertex in the loopless $\times$-graph has degree zero, then it cannot be associated with a Medusa.  If the $\times$-vertex in the loopless $\times$-graph has degree one, then the loopless $\times$-graph takes the form of a $\times$-vertex and an edge joining this $\times$-vertex to a vertex in a loopless plain-graph, $\Gamma$.  The unique Medusa associated with this $\times$-graph is given by deleting the edge incident to the $\times$-vertex and adding a $\star$-vertex together with an edge joining the $\star$-vertex and the $\times$-vertex.
\end{proof}

\noindent Following the same logic as in the proofs of Corollary \ref{prop: 1} and \ref{cor: a}, we obtain:

\begin{Cor} \label{prop: 3}
	For $\pd =1$, any two associated Medusas are identical to each other.  
\end{Cor}

\begin{Cor} \label{medusa inclusion under rho}
	For $\pd =1$, if $M$ is a Medusa in a sum of Medusas $L_M$, then $\rho ( L_M )$ contains all graphs in $\rho ( M )$.
\end{Cor}

\begin{Cor} \label{cor: c}
	For $\pd =1$, if a Medusa $M$ is associated with a plain-graph in a 1-invariant, $L \in \mathcal{L}_{loopless}$, then $\delta(L)$ contains all graphs in $\rho (M)$.
\end{Cor}

\begin{proof}
	If $M$ is associated with $\Gamma$ in the 1-invariant $L$, then there exists $\Gamma^{ \times }$ shared by $\delta ( \Gamma )$ and $\rho ( M)$.  Corollary $\ref{cor: a}$ implies $\Gamma^{ \times }$ is in $\delta (L)$.  Theorem $\ref{thm: loopless times basis}$ implies $\delta (L) = \rho ( L_M )$ where $L_M$ is a sum of Medusas.  Therefore, $\Gamma^{ \times }$ is in $\rho ( L_M )$.  Proposition $\ref{lem: 2}$ implies $M$ is in $L_M$.  Corollary $\ref{medusa inclusion under rho}$ implies $\rho ( L_M )$ contains $\rho ( M )$ and thus $\delta (L)$ contains $\rho ( M )$.
\end{proof}

\subsubsection{Minimal Invariants}

\begin{Def} \label{def: minimal}
	A nonzero $\pd$-invariant $L_{\nv , \Delta}$ is \emph{minimal} if there is no nonzero $\pd$-invariant $L'_{\nv , \Delta'}$ for any $\Delta' < \Delta$.  For a given $\pd$ and $\nv$, let $\Delta_{\text{min}}$ denote the minimum $\Delta$ for which a $\pd$-invariant exists.
\end{Def}

\noindent Now, we prove that a minimal 1-invariant is a sum of trees.  We start with the following lemma.

\begin{Lem} [Leaf Shuffling] \label{lem: ls}
	If a graph $\Gamma_A$ that contains a leaf $v$ appears in a 1-invariant $L$, then any graph $\Gamma_B$ that contains $v$ as a leaf and satisfies $\Gamma_B - v = \Gamma_A - v$ is also contained in $L$.
\end{Lem}

\begin{proof}
	Depicted below is a series of graphs, which are all associated with each other:
	
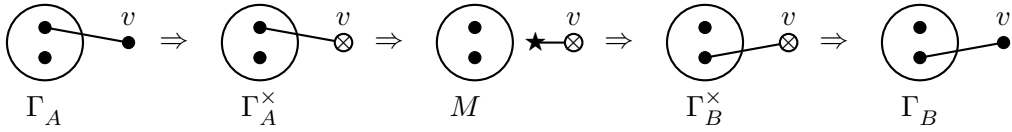
\begin{center}
\begin{minipage}{2.7cm}
\begin{tikzpicture}
	\draw [thick] (0.5,0.5) circle [radius=0.5];
	\filldraw (0.5,0.7) circle [radius=0.08];
	\filldraw (1.6,0.5) circle [radius=0.08];
	\filldraw (0.5,0.3) circle [radius=0.08];
	\draw [thick] (1.6,0.5) -- (0.5,0.7);
	\node [above] at (1.6,0.6) {$v$};
	\node [below] at (0.5,0) {$\Gamma^{\phantom{\times}}_A$};
	\node at (2.2,0.5) {$\Rightarrow$};
\end{tikzpicture}
\end{minipage}
\begin{minipage}{2.7cm}
\begin{tikzpicture}
	\draw [thick] (0.5,0.5) circle [radius=0.5];
	\filldraw (0.5,0.7) circle [radius=0.08];
	\filldraw (0.5,0.3) circle [radius=0.08];
	\draw [thick] (1.6,0.5) -- (0.5,0.7);
	\filldraw [white] (1.6,0.5) circle [radius=0.115];
	\draw [thick] (1.6,0.5) circle [radius=0.115];
	\node at (1.6,0.5) {$\times$};
	\node [above] at (1.6,0.6) {$v$};
	\node [below] at (0.5,0) {$\Gamma_A^\times$};
	\node at (2.2,0.5) {$\Rightarrow$};
\end{tikzpicture}
\end{minipage}
\begin{minipage}{2.9cm}
\begin{tikzpicture}
	\draw [thick] (0.5,0.5) circle [radius=0.5];
	\filldraw (0.5,0.7) circle [radius=0.08];
	\filldraw (0.5,0.3) circle [radius=0.08];
	\draw [thick] (1.8,0.5) -- (1.3,0.5);
	\node at (1.3,0.5) {\scalebox{0.8}{$\bigstar$}};
	\filldraw [white] (1.8,0.5) circle [radius=0.115];
	\draw [thick] (1.8,0.5) circle [radius=0.115];
	\node at (1.8,0.5) {$\times$};
	\node [above] at (1.8,0.6) {$v$};
	\node [below] at (0.5,0) {$M^{\phantom{\times}}_{\phantom{A}}$};
	\node at (2.4,0.5) {$\Rightarrow$};
\end{tikzpicture}
\end{minipage}
\begin{minipage}{2.7cm}
\begin{tikzpicture}
	\draw [thick] (0.5,0.5) circle [radius=0.5];
	\filldraw (0.5,0.7) circle [radius=0.08];
	\filldraw (0.5,0.3) circle [radius=0.08];
	\draw [thick] (1.6,0.5) -- (0.5,0.3);
	\filldraw [white] (1.6,0.5) circle [radius=0.115];
	\draw [thick] (1.6,0.5) circle [radius=0.115];
	\node at (1.6,0.5) {$\times$};
	\node [above] at (1.6,0.6) {$v$};
	\node [below] at (0.5,0) {$\Gamma_B^\times$};
	\node at (2.2,0.5) {$\Rightarrow$};
\end{tikzpicture}
\end{minipage}
\begin{minipage}{2.7cm}
\begin{tikzpicture}
	\draw [thick] (0.5,0.5) circle [radius=0.5];
	\filldraw (0.5,0.7) circle [radius=0.08];
	\filldraw (1.6,0.5) circle [radius=0.08];
	\filldraw (0.5,0.3) circle [radius=0.08];
	\draw [thick] (1.6,0.5) -- (0.5,0.3);
	\node [above] at (1.6,0.6) {$v$};
	\node [below] at (0.5,0) {$\Gamma^{\phantom{\times}}_B$};
\end{tikzpicture}
\end{minipage}
\end{center}

\noindent where the circles denote subgraphs.  In particular, both $\Gamma_A$ and $\Gamma_B$ are associated with the Medusa $M$.  By Corollary \ref{cor: c}, $\delta L$ contains all graphs in $\rho (M)$.  Furthermore, by Corollary \ref{cor: d}, $L$ contains both $\Gamma_A$ and $\Gamma_B$.
\end{proof}

\noindent We call the procedure that relates $\Gamma_B$ to $\Gamma_A$ described in the above lemma \emph{leaf shuffling}.  The corollaries below follow immediately.  

\begin{Cor} \label{cor: rp}
	If a plain-graph $\Gamma_A$ is contained in a 1-invariant $L$, and $\Gamma_B$ is formed from $\Gamma_A$ by shuffling leaves, then $\Gamma_B$ is also contained in $L$.
\end{Cor}

\begin{Cor} \label{cor: ist}
	If a plain-graph $\Gamma$ in a loopless 1-invariant $L$ contains more than one connected component, then none of these connected components is a tree.
\end{Cor}

\begin{proof}
	Note that there are at least two leaves in a tree, if the tree is not an empty vertex.  If in $\Gamma$ there is a connected component $T$ that is a tree, then we can shuffle all leaves in $T$ to be joined to other connected components, while turning $T$ into another tree with at least one fewer vertex.  This procedure can be iterated until all but one vertex in $T$ are shuffled to be joined to other connected components, which turns $\Gamma$ into a graph $\tilde{\Gamma}$ containing an empty vertex.  By Corollary \ref{cor: rp}, since $L$ contains $\Gamma$, it also contains $\tilde{\Gamma}$, which violates the lower bound on vertex degree.  Therefore, such $\Gamma$ cannot appear in loopless 1-invariants.
\end{proof}

\begin{Prop} \label{prop: 4}
	If $L$ is a nonzero minimal $N$-point loopless 1-invariant, then it contains all trees with $N$ vertices.  
\end{Prop}

\begin{proof}
	If $L$ contains no plain-graphs with leaves, then the vertices in $L$ have degree at least $2$ (note that an empty vertex is disallowed by Proposition \ref{prop: min plain}).  Then the number of edges $\Delta$ satisfies 
\begin{equation} \label{eq: ine1}
	\Delta \geq N.  
\end{equation}

	Otherwise, consider a plain-graph $\Gamma$ contains a leaf $v$ and appears in $L$.  If there exist any other leaves in $\Gamma$, shuffle them to be adjacent to $v$.  Iterate for the resulting graphs until reaching a plain-graph $\Gamma_0$ in which all leaves are adjacent to $v$.  By Corollary \ref{cor: rp}, $\Gamma_0$ is also contained in $L$.  Denote the subgraph of $\Gamma_0$ consisting of all leaves in $\Gamma_0$ and $v$ by $T$, which is a particular type of tree usually called a \emph{star}.  Define $\Gamma'$ as a subgraph of $\Gamma_0$ that is formed by deleting from $\Gamma_0$ all vertices in $T$.
	
	If $\Gamma'$ is not null, by Corollary \ref{cor: ist}, $T$ cannot be disconnected from $\Gamma'$, since otherwise $T$ would be a tree disconnected from at least one other connected component in $\Gamma_0$.  Moreover, since $v$ is a leaf in $\Gamma$, $\Gamma'$ is joined to $T$ by exactly one edge incident to $v$.  Define $N_T$ as the number of vertices in $T$ and $N'$ as the number of vertices in $\Gamma'$, then $N = N_T + N'$.  Note that there is no leaf in $\Gamma'$, and thus the vertices in $\Gamma'$ have degree at least 2.  Therefore,
\begin{equation} \label{eq: ine2}
	\Delta \geq (N_T - 1) + N' + 1 = N.
\end{equation}
	
	If $\Gamma'$ is null, then $\Gamma_0 = T$ and the original graph $\Gamma$ is a tree.  In this case, $\Delta = N - 1$, which is lower than the bounds (\eqref{eq: ine1} and \eqref{eq: ine2}) for the other cases.  By Corollary \ref{cor: rp}, $L$ contains $T$. The Galileon invariants presented in \S\ref{sec: Galileons} realize this bound $\Delta_\text{min} = N - 1$.  In fact, any tree with $N$ vertices can be turned into a star by shuffling leaves.  Since $L$ contains the star with $N$ vertices, it contains all trees with $N$ vertices.	
\end{proof}

\noindent Next, we prove the uniqueness of the minimal term.

\begin{Prop} \label{thm: 3}
    The minimal $\nv$-point loopless 1-invariant with $\Delta = \nv -1$ is unique up to proportionality.
\end{Prop}

\begin{proof}
	Let $L_1$ and $L_2$ be minimal $\nv$-point loopless 1-invariants with $\Delta = \nv -1$.  Let $T$ be a tree in $L_1$ and $L_2$, which exists by Proposition $\ref{prop: 4}$.  Rescale $L_1$ and $L_2$ so that $T$ appears in each with unit coefficient.  Then, $T$ is not in $L_1 - L_2$.  However, $L_1 - L_2$ is a minimal $\nv$-point loopless 1-invariant.  Therefore, Proposition $\ref{prop: 4}$ implies that $L_1 - L_2$ vanishes.
\end{proof}

\noindent Finally, we prove the existence of the minimal $\nv$-point loopless 1-invariant with $\Delta = \nv -1$:

\begin{Thm} \label{thm: 4}
	Any minimal $\nv$-point loopless $1$-invariant is proportional to the sum with unit coefficients of all trees with $\nv$ vertices.
\end{Thm}

\begin{proof}
	Let $T_\nv$ be a general tree with $\nv$ $\bullet$-vertices.  Let $T^{ \times }_{\nv}$ be a general tree with $( \nv -1)$ $\bullet$-vertices and one $\times$ vertex.  Let $T^{ \star }_{\nv}$ have two connected components, one of which is 
$
\begin{minipage}{1.2cm}
\begin{tikzpicture}
	\draw [thick] (0,0) -- (0.6,0);
        	\node at (0,0) {\scalebox{0.8}{$\bigstar$}};
        	\filldraw [white] (0.6,0) circle [radius=0.115];
	\draw [thick] (0.6,0) circle [radius=0.115];
	\node at (0.6,0) {$\times$};
\end{tikzpicture}
\end{minipage}
$
and the other is a tree $T_{\nv -1}$.  Let $L$, $L^{ \times }$ and $L^{ \star }$ be the sum with unit coefficients of all $T_\nv$, $T^{ \times }_{\nv}$ and $T^{ \star }_{\nv}$, respectively.  Replacing a $\bullet$-vertex in $T_\nv$ with a $\times$ produces a unique $T^{ \times }_{\nv}$, since vertices are labeled.  Therefore, $\delta (L)$ is simply $L^{ \times }$.  Similarly, $\rho ( L^{ \star } ) = L^{ \times }$.  Therefore, $\delta (L ) = \rho (L^{ \star } )$ and $L$ is 1-invariant, which is unique by virtue of Proposition $\ref{thm: 3}$.
\end{proof}

\noindent It is shown in \cite{gal} that the Galileon-like term in spacetime dimension $d = \sdim + 1 \geq \nv$:
\begin{equation}\label{eq: Galileon}
	\epsilon^{i_1 \cdots i_{\nv -1} \, k_\nv \cdots k_\sdim}_{\phantom{k_\sdim}} \epsilon^{j_1 \cdots j_{\nv -1}}_{\hspace{1.2 cm} k_\nv \cdots k_\sdim} \partial_{i_1} \phi \, \partial_{j_1} \phi \, \partial_{i_2} \partial_{j_2} \phi \cdots \partial_{i_{\nv -1}} \partial_{j_{\nv -1}} \phi,
\end{equation} 
\noindent is invariant up to a total derivative.  By Theorem \ref{thm: 4}, the sum with unit coefficients of all trees with $\nv$ vertices is proportional to \eqref{eq: Galileon}, up to a total derivative.  In fact, the constant of proportionality is 1.

\subsubsection{Non-minimal Invariants} 

	So far we have found the unique minimal 1-invariant for each $\nv$, that is with $\Delta= \nv -1$.  In what follows we argue that any non-minimal 1-invariant ($\Delta > \nv -1$) is equal to an exact 1-invariant up to a plain-relation (a total derivative).  We begin with the following definition:

\begin{Def} [Frame]
	For $\Gamma \in \mathcal{G}_{\nv , \Delta}^{}$, delete all edges incident to leaves;  iterate this procedure until reaching a graph $\Gamma_f$ that contains no leaves.  We call $\Gamma_f$ the \emph{frame} of $\Gamma$.
\end{Def}

 \noindent Note that, by definition, after deleting from $\Gamma$ all edges that appear in $\Gamma_f$, each of the connected components in the resulting graph is a tree.  We also define:

\begin{Def} [Frame Invariant]
	A loopless 1-invariant $L = \sum_{i=1}^k \alpha_i \Gamma_i$ is a \emph{frame invariant} if the frames of $\Gamma_{i}$ for all $i = 1, \ldots, k$ are identical.  
\end{Def}

\noindent For convenience, we define a map $f$ on frame invariants, such that $f(L)$ is the frame which is common to all of the graphs contained in $L$.

	By definition, the nonempty vertices in any frame have degree greater than 1.  No such vertex can be turned into a $\times$-vertex or be adjacent to a $\star$-vertex in a $\star$-graph associated with $\Gamma$.  This means that terms with different frames cannot lead to cancellations in the consistency equation.  Therefore, a consistency equation naturally splits into multiple consistency equations, each one a frame invariant.  This is summarized in the lemma below:

\begin{Lem} \label{lem: framinvlincomb}
	Any loopless 1-invariant is a linear combination of frame invariants.
\end{Lem}

\noindent The proofs presented in Proposition \ref{prop: 4} and \ref{thm: 3} and Theorem \ref{thm: 4} are directly applicable to frame invariants, from which we conclude:

\begin{Prop} \label{prop: frame}
	Up to proportionality, any loopless frame invariant $L$ is equal to the sum with unit coefficients of all plain-graphs that satisfy the following conditions:
\begin{enumerate}[$\emph{(}$a$\emph{)}$]

\item

	The plain-graph has a frame $f ( L )$.

\item \label{propframeb}

	If there is more than one connected component in the plain graph, then none of them is a tree.
	
\end{enumerate}
\end{Prop}

\noindent Note that condition (\ref{propframeb}) follows directly from Corollary \ref{cor: ist}. Proposition \ref{prop: frame} effectively provides an equivalent definition for frame invariants.  Classifying all non-minimal 1-invariants is thus reduced to classifying all non-minimal frame invariants.  Furthermore, by Proposition \ref{prop: loopless_basis}, we can restrict our search to loopless frame invariants.  In the following we will show that any loopless frame invariant is equal to an exact invariant up to total derivatives.  We start with a useful lemma:

\begin{Lem} \label{lem: l loop exact}
	Let $\Gamma^{(\ell)}\in \mathcal{G}$ be an $\ell$-looped exact 1-invariant, such that all vertices with loops have degree 2.  Then $\rho^{(0)} \circ \theta ( \Gamma^{(\ell)} )$ is a loopless 1-invariant and is equal to $\Gamma^{(\ell)}$ up to a plain-relation.  
\end{Lem}

\begin{proof}
	Label the loops in $\Gamma^{(\ell)}$ from $1$ to $\ell$.  Note that $\rho \circ \theta_1 ( \Gamma^{(\ell)} ) = \Gamma^{(\ell)} + \rho^{(\ell - 1)} \circ \theta_1 ( \Gamma^{(\ell)} )$.  Since $\theta_1 ( \Gamma^{(\ell)} )$ is a $\star$-ed plain graph, $\rho \circ \theta_1 ( \Gamma^{(\ell)} )$ is a plain-relation.  Hence $\rho^{(\ell - 1)} \circ \theta_1 ( \Gamma^{(\ell)} )$ is 1-invariant.  Define \begin{equation*}
	L_1 \equiv \rho^{(\ell - 1)} \circ \theta_1 ( \Gamma^{(\ell)} ) = \rho \circ \theta_1 ( \Gamma^{(\ell)} ) - \Gamma^{(\ell)},
\end{equation*}
\noindent $L_1$ is by definition an exact 1-invariant up to a plain-relation $\rho \circ \theta_1 ( \Gamma^{\ell} )$.  Furthermore, all graphs in $L_1$ are $(\ell - 1)$-looped.  Therefore, $\rho \circ \theta_2 ( L_1 ) = L_1 + \rho^{(\ell - 2)} \circ \theta_2 ( L_1 )$.  Define 
\begin{equation*}
	L_2 \equiv \rho^{(\ell - 2)} \circ \theta_2 ( L_1 ) = \rho \circ \theta_2 (L_1) - L_1,
\end{equation*}
\noindent which is also an exact 1-invariant up to plain-relations and consists of graphs that are $(\ell - 2)$-looped.  Iterating $\ell$ times, we obtain
\begin{equation*}
	L_\ell = \rho^{(0)} \circ \theta_{\ell - 1} \circ \cdots \circ \rho^{(\ell - 2)} \circ \theta_{2} \circ \rho^{(\ell - 1)} \circ \theta_1 (\Gamma^{(\ell)}) = \rho^{(0)} \circ \theta (\Gamma^{(\ell)}),
\end{equation*}
\noindent which is an exact 1-invariant up to plain-relations.  
\end{proof}

\begin{Thm}
	A non-minimal frame invariant is an exact invariant up to a plain-relation.
\end{Thm}

\begin{proof}
	It is sufficient to consider any loopless non-minimal frame invariant $L^{(k)}\in\mathcal{L}_{\nv , \Delta}^{}$, with $k$ the number of empty vertices in $f ( L^{(k)} )$.  Note that since $L^{(k)}$ is non-minimal, it is not a tree and therefore $k < \nv$.  We prove the theorem by induction on $k$.  
\begin{enumerate}

\item

	$k = 0$: In this case, $f ( L^{(0)} ) = L^{(0)}$.  By construction, a vertex in a graph in $L^{(0)}$ is of degree no less than 2, and thus $L^{(0)}$ is already exactly invariant.

\item

	If any $L^{(k)}$ with $k < \alpha$ is an exact invariant plus a plain-relation: Consider any loopless non-minimal frame invariant $L^{(\alpha)}$.  Form an $\ell$-looped exact 1-invariant $\Gamma^{(\alpha)}$ from $f (L^{(\alpha)})$ by adding a loop to each empty vertex.  By Lemma \ref{lem: l loop exact}, $\rho^{(0)} \circ \theta (\Gamma^{(\alpha)}) $ is a loopless 1-invariant, equal to $\Gamma^{(\alpha)}$ up to plain-relations.
	Using Lemma \ref{lem: framinvlincomb}, $\rho^{(0)} \circ \theta ( \Gamma^{(\alpha)} ) =  \sum_i F_i$, where each $F_i$ is a frame invariant with a distinct frame.  Provided $\alpha<n$, there is exactly one $F_i$, say $\tilde{F}$, with $f ( \tilde{F} )=f (L^{(\alpha)})$, and all other $F_i$ have fewer than $\alpha$ empty vertices in $f ( F_i )$.  Since all other $F_i$'s have fewer than $\alpha$ empty vertices in $f(F_i)$, they are exact 1-invariants up to a plain-relation, by the induction hypothesis.  But this means that $\tilde{F}$ is also an exact 1-invariant up to plain-relations.  By Proposition \ref{prop: frame}, since $\tilde{F}$ and $L^{(\alpha)}$ share the same frame, $\tilde{F} $ is proportional to $L^{(\alpha)}$ and thus $L^{(\alpha)}$ is also an exact 1-invariant up to plain-relations.
\end{enumerate}
\noindent By induction, $L^{(k)}$ is exactly 1-invariant up to plain-relations for any $0\leq k < \nv$.
\end{proof}

\noindent From the above discussion, we can conclude: The set of all exact 1-invariants with all looped vertices of degree 2 generates all non-minimal 1-invariants, up to plain-relations.  An example of these graphs for $\nv =3$ and $\Delta=4$ is shown in Figure \ref{fig: nonminimal}.
\begin{figure}[t!]
\begin{align*}
\begin{minipage}{\mini cm}
\begin{tikzpicture}
\draw [thick] (0.4,0.69) ..  controls (-0.168,0.115) and (0.968,0.115) ..  (0.4,0.69);
\filldraw (0,0) circle [radius=0.08];
\filldraw (\1,0) circle [radius=0.08];
\filldraw (0.4,0.69) circle [radius=0.08];
\draw [thick] (0,0) -- (\1,0);
\draw [thick] (0,0) to [out=30,in=150] (\1,0);
\draw [thick] (0,0) to [out=-30,in=-150] (\1,0);
\end{tikzpicture}
\end{minipage} 
\hspace{1cm}
\begin{minipage}{\mini cm}
\begin{tikzpicture}
\filldraw (0,0) circle [radius=0.08];
\filldraw (\1,0) circle [radius=0.08];
\filldraw (0.4,0.69) circle [radius=0.08];
\draw [thick] (0,0) -- (0.4,0.69) -- (\1,0);
\draw [thick] (0,0) to [out=30,in=150] (\1,0);
\draw [thick] (0,0) to [out=-30,in=-150] (\1,0);
\end{tikzpicture}
\end{minipage} 
\hspace{1cm}
\begin{minipage}{\mini cm}
\begin{tikzpicture}
\filldraw (0,0) circle [radius=0.08];
\filldraw (\1,0) circle [radius=0.08];
\filldraw (0.4,0.69) circle [radius=0.08];
\draw [thick] (0,0) to [out=30,in=150] (\1,0);
\draw [thick] (0,0) to [out=-30,in=-150] (\1,0);
\draw [thick] (0,0) to [out=90,in=-150] (0.4,0.69);
\draw [thick] (0,0) to [out=30,in=-90] (0.4,0.69);
\end{tikzpicture}
\end{minipage}
\end{align*}
\caption{All 1-invariant terms up to total derivatives, with $\nv =3$, $\Delta=4$ and $\pd =1$.}
\label{fig: nonminimal}
\end{figure}
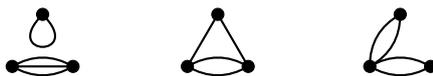

\vspace{0.3cm}

We end our discussion of the linear shift symmetry with a summary of the full classification of 1-invariants: 

\begin{Thm} [Classification of 1-invariants] \label{thm: classification of 1-invariants}
	The sum with unit coefficients of all trees with $n$ vertices is the unique 1-invariant with $\Delta = \nv -1$ \emph{(}up to proportionality and plain-relations\emph{)}.  The set of all graphs consisting of $\nv$ vertices, with all vertices of degree higher than 1 and any looped vertex of degree 2, generate all 1-invariants with $\Delta > \nv - 1$ \emph{(}up to plain-relations\emph{)}.  There are no invariants with $\Delta < \nv - 1$.
\end{Thm}

\subsection{Invariants from Superpositions} \label{sec: superposition} 

In this section, we describe a method of combining invariants to form other invariants.  Therefore, we will need to keep track of the degree of the polynomial shifts under which the variation of various terms are taken.  It is important to recall at this point that the variation map $\delta_\pd$ depends crucially on $\pd$.  Therefore, all the different types of graphs depend on $\pd$ as well.  Until now, this dependence on $\pd$ has been kept implicit. We will now make it explicit by referring to graphs as \emph{$\pd$-graphs}. 

The method of combining invariants involves the notion of superposition, defined below, which combines graphs with different values of $\pd$.

\begin{Def} [Superposition of Graphs]
	Given a $\pd_A$-graph $\Gamma_A$ and a $\pd_B$-graph $\Gamma_B$, which each have the same value of $\nv$, the \emph{superposition} of $\Gamma_A$ and $\Gamma_B$ is a $\pd$-graph formed by applying the the following procedure: 
\begin{enumerate}

\item
	
	If there is a $\times$-vertex in $\Gamma_B$, replace the vertex in $\Gamma_A$ that has the same label as the $\times$-vertex in $\Gamma_B$ with a $\times$-vertex.
	
\item
	
	Add any $\star$-vertices in $\Gamma_B$ to $\Gamma_A$.   
	
\item
	
	Take all edges in $\Gamma_B$ and add them to $\Gamma_A$, joining the same vertices as they do in $\Gamma_B$.

\item
	
	Identify the resulting graph as a null graph if deg$(\times)$ is higher than $\pd$ or there are two $\times$-vertices.
	
\end{enumerate}

\noindent The resulting graph is denoted by $\Gamma_A \cup \Gamma_B$.  
\end{Def}

\noindent Note that $\Gamma_A \cup \Gamma_B = \Gamma_B \cup \Gamma_A$. Note that the above definition of superposition depends on $\pd$.  We will refer to such a superposition as a \emph{$\pd$-superposition}.

\begin{Def} [Superposition of Linear Combinations]
	Given the linear combinations $L_A = \sum_{i=1}^{k_A} a_i \hspace{0.6mm} \Gamma^A_i$ and $L_B = \sum_{i=1}^{k_B} b_i \hspace{0.6mm} \Gamma^B_i$, where $\Gamma^A_i, \Gamma^B_j$ are graphs with the same $n$, the superposition of $L_A$ and $L_B$ is defined as 
\begin{equation*}
	L_A \cup L_B \equiv \sum_{i=1}^{k_A} \sum_{j=1}^{k_B} a_i \hspace{0.6mm} b_j \hspace{0.6mm} \Gamma^A_i \cup \Gamma^B_j.
\end{equation*}
\end{Def}

\subsubsection{Superposition of a \texorpdfstring{$\pd$}{\pd}-invariant and an Exact Invariant}

This section involves the construction of new invariants by taking the superposition of a $\pd$-invariant with an exact invariant.  
\begin{Lem}\label{lem: exactsuper}
\quad
\begin{enumerate}
\item     

    \label{lem: exactsuper s1}Given a $\pd$-graph $\Gamma \in \mathcal{L}_{\nv , \Delta}^{}$ and a $\pd_E$-graph $\Gamma_E \in \mathcal{L}_{\nv , \Delta_E}$, where $\Gamma_E$ is an exact $\pd_E$-invariant,
\begin{equation} \label{eq: supdel=delsup}
    \Gamma_E \cup \delta_\pd ( \Gamma ) = \delta_{\pd + \pd_E + 1}(\Gamma_E \cup \Gamma),
\end{equation}
where $\cup$ denotes $( \pd + \pd_E + 1)$-superposition.
\item

    \label{lem: exactsuper s2}Given a $\pd$-graph $\Gamma^{ \star }\in \mathcal{L}^{ \star }_{\nv , \Delta}$ and a $\pd_E$-graph $\Gamma_E \in \mathcal{L}_{\nv , \Delta_E}$, where $\Gamma_E$ is an exact $\pd_E$-invariant,
\begin{equation} \label{eq: suprho=rhosup}
    \Gamma_E \cup \rho(\Gamma^{ \star }) = \rho(\Gamma_E \cup \Gamma^{ \star }),
\end{equation}
where $\cup$ denotes $(\pd + \pd_E + 1)$-superposition.
\end{enumerate}
\end{Lem}
\begin{proof}
\quad
\begin{enumerate}

\item 

Operationally, a graph in the LHS of \eqref{eq: supdel=delsup} is given by substituting one $\bullet$-vertex, $v$, in $\Gamma$ with a $\times$-vertex and then adding the edges in $\Gamma_E$ to the result.  Meanwhile, the RHS is given by adding the edges in $\Gamma_E$ to $\Gamma$ first before substituting $v$ by a $\times$-vertex.  Thus, \eqref{eq: supdel=delsup} is violated only when a graph vanishes from one side and not the other.  A graph vanishes from the RHS if and only if the degree of the vertex $v$ in $\Gamma_E \cup \Gamma$ is greater than $\pd + \pd_E + 1$.  If this condition holds, then the graph also vanishes from the LHS by the rules of $( \pd + \pd_E + 1)$-superposition.  A graph could also possibly vanish from the LHS if deg$(v) > \pd$ in $\Gamma$.  However, if deg$(v) > \pd$ in $\Gamma$, then deg$(v) > \pd + \pd_E + 1$ in $\Gamma_E \cup \Gamma$ since the degree of a vertex in $\Gamma_E$ is at least $\pd_E +1$ (by Corollary \ref{lem: exact}).  Therefore, the conditions for the vanishing of a graph from either side of \eqref{eq: supdel=delsup} are identical and thus the equation holds.

\item 

Once again, \eqref{eq: suprho=rhosup} is violated only when a graph vanishes from one side and not the other.  A graph vanishes from the RHS if and only if the degree of the $\times$-vertex in $\rho(\Gamma_E \cup \Gamma^{ \star })$ is greater than $\pd + \pd_E + 1$.  If this condition holds, then the graph also vanishes from the LHS by the rules of $( \pd + \pd_E + 1)$-superposition.  A graph on the LHS could also possibly vanish if deg$(\times) > \pd$ for a graph $\Gamma^{ \times }$ in $\rho(\Gamma^{ \star })$.  However, then deg$(\times) > \pd + \pd_E + 1$ in $\Gamma_E \cup \Gamma^{ \times }$, since the degree of a vertex in $\Gamma_E$ is at least $\pd_E +1$ (by Corollary \ref{lem: exact}).  Therefore, the conditions for the vanishing of a graph from either side of \eqref{eq: suprho=rhosup} are identical and thus the equation holds.
\end{enumerate}
\vspace{-0.5cm}
\end{proof} 
 
\noindent Now we apply Lemma \ref{lem: exactsuper} to prove the main result:
 
\begin{Thm}
For fixed $\nv$, the superposition of a $\pd$-invariant and an exact $\pd_E$-invariant is a $( \pd + \pd_E + 1)$-invariant.
\end{Thm}

\begin{proof}
Denote the $P$-invariant by $L=\sum_{i=1}^{k} b_i\Gamma_i$ and the exact $\pd_E$-invariant by $L_E=\sum_{i=1}^{k_E} a_i \Gamma^E_i$.  By Corollary \ref{lem: exact}, all vertices in $\Gamma^E_i$ have degree greater than $P_E$.  Since $L$ is a $P$-invariant, there exists a linear combination of $\star$-graphs, $L^{ \star }=\sum_{i=1}^{k^{ \star }} c_i \Gamma^{ \star }_i$, such that the folowing consistency equation holds:
\begin{equation} \label{eq: consistency superposition}
    \delta_{P} (L)=\rho(L^{ \star }).  
\end{equation}
\noindent Define $\tilde{L} \equiv L_E \cup L = \sum_{i,j} a_i  b_j   \Gamma^E_i\cup\Gamma_j $.  Then, using Statement \ref{lem: exactsuper s1} of Lemma \ref{lem: exactsuper}:
\begin{equation} \label{eq: superexacteqn1}
\begin{aligned}
    \delta_{\pd + \pd_E+1}(\tilde{L}) & = \sum_{i,j} a_i  b_j  \, \delta_{\pd + \pd_E+1}(\Gamma^E_i\cup\Gamma_j) \\
    & = \sum_{i,j} a_i  b_j \, \Gamma^E_i\cup\delta_{\pd}(\Gamma_j) = \sum_{i} a_i \, \Gamma^E_i\cup\delta_{\pd}(L)
\end{aligned}
\end{equation}

\noindent Furthermore, define $\tilde{L}^{ \star }\equiv\sum_{i,j} a_i c_j  \Gamma^E_i\cup \Gamma^{ \star }_j$ using $( \pd + \pd_E + 1)$-superposition.  Using Statement \ref{lem: exactsuper s2} of Lemma \ref{lem: exactsuper}:
\begin{align}\label{eq: superexacteqn2}
&\rho(\tilde{L}^{ \star })= \sum_{i,j} a_i c_j  \rho(\Gamma^E_i\cup \Gamma^{ \star }_j)=\sum_{i,j} a_i c_j  \Gamma^E_i\cup\rho( \Gamma^{ \star }_j)=\sum_{i} a_i   \Gamma^E_i\cup\rho(L^{ \star })
\end{align}
Combining \eqref{eq: consistency superposition}, \eqref{eq: superexacteqn1} and \eqref{eq: superexacteqn2} we have that $\delta_{\pd + \pd_E + 1} (\tilde{L})=\rho(\tilde{L}^{ \star })$ and therefore $\tilde{L}\equiv L_E \cup L$ is a $(\pd + \pd_E + 1)$-invariant.
\end{proof}

\subsubsection{Superposition of Minimal Loopless 1-invariants}

    In this section we show that the superposition of $\nts$ minimal loopless 1-invariants results in a $(2\nts-1)$-invariant.  To prove this statement, we need to construct a linear combination of Medusas in order to write down a valid consistency equation.  This construction requires intermediate $\star$-graphs called ``hyper-Medusas", which we now define:

\begin{Def} [Hyper-Medusa]
   A \emph{hyper-Medusa} is a loopless $\star$-graph with all $\star$-vertices adjacent to the $\times$-vertex, such that the degree of the $\times$-vertex $\text{\emph{deg}}(\times)$ and the number of $\star$-vertices $N(\star)$ satisfy $\text{\emph{deg}}(\times) \geq \pd + 1 - \nsv$.  
\end{Def}

\begin{Lem}\label{lem: rhohm}
    Given $M_h$ a hyper-Medusa, there exists a linear combination of Medusas $L_M$ that satisfies $\rho^{(0)} ( M_h ) = \rho^{(0)} ( L_M )$.  
\end{Lem}

\begin{proof}
    Within the action of $\rho^{(0)}$, we can delete deg$(\times) + \nsv - \pd + 1 \geq 0$ $\star$-vertices and then add the same number of edges in $M_h$ , yielding a linear combination of $\star$-graphs with exactly $\pd + 1 - \text{deg}(\times)$ $\star$-vertices.  These resulting graphs are Medusas.  
\end{proof}

\noindent The following definition will allow us to construct the desired hyper-Medusas:

\begin{Def} \label{def: hm}
	Take $\Gamma_1$ to be any $\times$-graph or $\star$-graph and for each $i=2,...,\nts$, take $T_i$ to be any tree, such that $\Gamma$ and $T_i$ have the same value of $n$.  Label the $\times$-vertex in $\Gamma_1$ by $v^\times$ and define $T^{ \times }_i$ to be the graph formed from $T_i$ by replacing the vertex that is labeled by $v^\times$ with a $\times$-vertex.  If $v^\times$ is a leaf in $T_i$, then define $\tilde{T}_i$ to be the unique $\pd = 1$ Medusa that is associated with $T^{ \times }_i$, otherwise $\tilde{T}_i  = T^{ \times }_i$.  Then we define:
\begin{equation*}	
	\chi ( \Gamma_1 \cup T_2\cup\cdots\cup T^{}_\nts ) \equiv \Gamma_1 \cup \tilde{T}_1 \cup \cdots \cup  \tilde{T}^{}_\nts.
\end{equation*}
\end{Def}

\begin{Thm}
	For fixed $\nv$, the superposition of $\nts$ minimal loopless 1-invariants is a $(2 \nts - 1)$-invariant.
\end{Thm}

\begin{proof}
	By Theorem \ref{thm: 4}, any minimal $\nv$-point loopless $1$-invariant is equal to the sum with unit coefficients of all trees with $\nv$ vertices, up to proportionality.  Therefore, denote the $\nts$ copies of the minimal loopless 1-invariants by $L^{(c)}_n=\sum_{\alpha^{}_c=1}^{\nv^{\nv -2}}T^{(c)}_{\alpha^{}_{c}}$, for $c=1,\ldots,\nts$.  We add an additional structure to all graphs in this proof: We color all edges in all graphs in $L_\nv^{(c)}$ by a distinct color $(c)$.  Throughout this proof, two graphs are equal if and only if they are the same graph and, in addition, their edges are the same colors.  Taking into account this coloring, all of the plain-graphs in $L\equiv \bigcup_{c = 1}^{\, \nts} L^{(c)}_\nv$ now have unit coefficients, and the number of these plain-graphs is $\left ( \nv^{\nv -2} \right )^{\nts}$.  Moreover, $L$ is the sum over $\alpha^{}_1, \ldots, \alpha^{}_\nts$ of all such $T^{(1)}_{\alpha^{}_1} \cup T^{(2)}_{\alpha^{}_2} \cup \cdots \cup T^{(\nts)}_{\alpha^{}_{\nts}}$'s with unit coefficients.  By Theorem \ref{thm: 4}, there is a unique linear combination of Medusas $\sum_{\beta_c} M^{(c)}_{\beta_c}$ satisfying
\begin{equation*}
	\delta_1 \left ( L^{(c)}_\nv \right ) = \delta_1 \left ( \sum_{\alpha^{}_c=1}^{\nv^{\nv -2}} T^{(c)}_{\alpha^{}_c} \right ) = \rho^{(0)} \left ( \sum_{\beta_c=1}^{\nv (\nv - 1)^{\nv -3}} M^{(c)}_{\beta_c} \right ),
\end{equation*} 
\noindent where each $T_{\alpha^{}_c}^{(c)}$ is a distinct tree and each $M^{(c)}_{\beta_c}$ is a distinct $\pd = 1$ Medusa, consisting of a subgraph tree and a disconnected subgraph
$
\begin{minipage}{1.2cm}
\begin{tikzpicture}
	\draw [thick] (0,0) -- (0.6,0);
        	\node at (0,0) {\scalebox{0.8}{$\bigstar$}};
        	\filldraw [white] (0.6,0) circle [radius=0.115];
	\draw [thick] (0.6,0) circle [radius=0.115];
	\node at (0.6,0) {$\times$};
\end{tikzpicture}
\end{minipage}
$.  In the following, we take the limit $\pd \rightarrow \infty$, so that no graph vanishes.  Note that we have
\begin{align} \label{eq: start}
	& \sum_{\alpha^{}_1} \delta_\infty \left ( T^{(1)}_{\alpha^{}_1} \cup T^{(2)}_{\alpha^{}_2} \cup \cdots \cup T^{(\nts)}_{\alpha^{}_\nts}  \right ) = \sum_{\alpha^{}_1} \left ( \delta_\infty ( T^{(1)}_{\alpha^{}_1}) \right ) \cup T^{(2)}_{\alpha^{}_2} \cup \cdots \cup T^{(\nts)}_{\alpha^{}_\nts} \notag \\
	= & \sum_{\beta_1} \rho^{(0)} \left ( M^{(1)}_{\beta_1} \right ) \cup T^{(2)}_{\alpha^{}_2} \cup \cdots \cup T^{(\nts)}_{\alpha^{}_\nts} + \sum_{\alpha^{}_1} ( \delta_\infty - \delta_1 ) (T^{(1)}_{\alpha^{}_1})  \cup T^{(2)}_{\alpha^{}_2} \cup \cdots \cup T^{(\nts)}_{\alpha^{}_\nts}.
\end{align}

    Define $X_L \equiv \sum_{\beta_1, \alpha^{}_2, \ldots , \alpha^{}_\nts} M_{\beta_1}^{(1)} \cup T_{\alpha^{}_2}^{(2)} \cup \cdots \cup T_{\alpha^{}_\nts}^{(\nts)}$ and $X_R$ to be the sum with unit coefficients of all distinct graphs contained in $\sum_{\beta_1, \alpha^{}_2, \ldots , \alpha^{}_\nts} \chi ( M_{\beta_1}^{(1)} \cup T_{\alpha^{}_2}^{(2)} \cup \cdots \cup T_{\alpha^{}_\nts}^{(\nts)} )$.  In the following, we show that
\begin{equation} \label{eq: equal hyper}
	\rho^{(0)} (X_L) = \rho^{(0)} (X_R).
\end{equation}

\noindent Since $T_{\alpha^{}_c}^{(c)}$, $\alpha^{}_c = 1 , \ldots , \nv^{\nv -2}$, and $M_{\beta_c}^{(c)}$, $\beta_c = 1 , \ldots , \nv ( \nv -1)^{\nv -3}$, are all distinct from each other, all elements in $X_L$ and $X_R$ have unit coefficient.  Therefore, it will suffice to show that any graph in $\rho^{(0)} (X_R)$ is also in $\rho^{(0)} (X_L)$, and vice versa.

\vspace{3mm}

\noindent \textbf{RHS contains LHS:} Let $\Gamma^{ \times }$ be a $\times$-graph in $\rho^{(0)} (X_L)$.  Then, $\Gamma^{ \times }$ is contained in $\rho^{(0)} \bigl( M_{\beta_1}^{(1)} \cup T_{\alpha^{}_2}^{(2)} \cup \cdots \cup T_{\alpha^{}_\nts}^{(\nts)} \bigr)$ for some $\beta_1$, $\alpha^{}_2, \ldots , \alpha^{}_\nts$.  The $\star$-graph, $M_{\beta_1}^{(1)} \cup T_{\alpha^{}_2}^{(2)} \cup \cdots \cup T_{\alpha^{}_\nts}^{(\nts)}$ induces a unique $\chi \bigl( M_{\beta_1}^{(1)} \cup T_{\alpha^{}_2}^{(2)} \cup \cdots \cup T_{\alpha^{}_\nts}^{(\nts)} \bigr)$, such that all $\times$-graphs in $\rho^{(0)} \bigl( M_{\beta_1}^{(1)} \cup T_{\alpha^{}_2}^{(2)} \cup \cdots \cup T_{\alpha^{}_\nts}^{(\nts)} \bigr)$ (including $\Gamma^{ \times }$) are contained in $\rho^{(0)} \circ \chi \bigl( M_{\beta_1}^{(1)} \cup T_{\alpha^{}_2}^{(2)} \cup \cdots \cup T_{\alpha^{}_\nts}^{(\nts)} \bigr)$.  Therefore, $\Gamma^{ \times }$ is in $\rho^{(0)} (X_R)$ and $\rho^{(0)} (X_R)$ contains $\rho^{(0)} (X_L)$.
	
\vspace{3mm}
	
\noindent \textbf{LHS contains RHS:} Let $\Gamma^{ \times }$ be a $\times$-graph in $\rho^{(0)} (X_R)$.  Then, $\Gamma^{ \times }$ is contained in $\rho^{(0)} \circ \chi \bigl( M_{\beta_1}^{(1)} \cup T_{\alpha^{}_2}^{(2)} \cup \cdots \cup T_{\alpha^{}_\nts}^{(\nts)} \bigr)$ for some $\beta_1$, $\alpha^{}_2, \ldots , \alpha^{}_\nts$.  In particular, $\Gamma^{ \times }$ is contained in some $\rho^{(0)} \bigl ( M_{\beta'_1}^{(1)} \cup T_{\alpha'_2}^{(2)} \cup \cdots \cup T_{\alpha'_\nts}^{(\nts)} \bigr )$ with $T_{\alpha'_2}^{(2)} \cup \cdots \cup T_{\alpha'_\nts}^{(\nts)}$ in $\rho^{(0)} \circ \chi \bigl( T_{\alpha^{}_2}^{(2)} \cup \cdots \cup T_{\alpha^{}_\nts}^{(\nts)} \bigr)$.  Since $M_{\beta'_1}^{(1)} \cup T_{\alpha'_2}^{(2)} \cup \cdots \cup T_{\alpha'_\nts}^{(\nts)}$ is in $X_L$, $\Gamma^{ \times }$ is in $\rho^{(0)} (X_L)$ and $\rho^{(0)} (X_L)$ contains $\rho^{(0)} (X_R)$.

\vspace{3mm}

	From \eqref{eq: equal  hyper} we obtain
\begin{equation} \label{eq: equal hyperp}
	\sum_{\beta^{}_1, \ldots, \alpha^{}_\nts} \rho^{(0)} \left ( M^{(1)}_{\beta_1} \right ) \cup T^{(2)}_{\alpha^{}_2} \cup \cdots \cup T^{(\nts)}_{\alpha^{}_\nts} =  \rho^{(0)} \left ( X_R \right ).
\end{equation}
\noindent Similarly,  
\begin{equation} \label{eq: inf minus 1}
	\sum_{\alpha^{}_1, \ldots, \alpha^{}_\nts} \left ( \delta_\infty - \delta_1 \right ) T^{(1)}_{\alpha^{}_1} \cup T^{(2)}_{\alpha^{}_2} \cup \cdots \cup T^{(\nts)}_{\alpha^{}_\nts} = \rho^{(0)} \left ( \tilde{X}_R \right ),
\end{equation}
\noindent with $\tilde{X}_R$ given by the sum with unit coefficients of all graphs contained in 
\begin{equation*}
	\sum_{\alpha^{}_1, \ldots, \alpha^{}_\nts} \chi \left (( \delta_\infty - \delta_1 ) T^{(1)}_{\alpha^{}_1} \cup  T^{(2)}_{\alpha^{}_2} \cup \cdots \cup T^{(\nts)}_{\alpha^{}_\nts} \right ).
\end{equation*} 
\noindent Therefore, by \eqref{eq: start}, \eqref{eq: equal hyperp} and \eqref{eq: inf minus 1}, we conclude that
\begin{align} \label{eq: hyper-Medusa}
	& \delta_\infty ( L ) = \sum_{\alpha^{}_1, \ldots, \alpha^{}_\nts} \delta_\infty \left ( T^{(1)}_{\alpha^{}_1}  \cup T^{(2)}_{\alpha^{}_2} \cup \cdots \cup T^{(\nts)}_{\alpha^{}_\nts}\right ) = \rho^{(0)} \left ( X_R+\tilde{X}_R \right ).
\end{align}

	Finally, switch back to $P = 2 \nts - 1$.  Then graphs with deg$(\times) > 2 \nts - 1$ will vanish simultaneously on both sides of \eqref{eq: hyper-Medusa}, and thus, in $P=2\nts-1$:
\begin{equation} \label{eq: 2p-1}
	\delta_{2\nts - 1} ( L ) = \rho^{(0)} \left ( X_R+\tilde{X}_R \right ).
\end{equation}
\noindent Next we show that any graph, $\Gamma$, in $X_R+\tilde{X}_R$ is a hyper-Medusa.  By construction, $\Gamma$ results from the superposition of graphs with either a $\times$-vertex of degree 1 joined to a $\star$-vertex or a $\times$-vertex of degree larger than 1.  Therefore, deg$(\times)$ in $\Gamma$ satisfies deg$(\times) \geq \nsv + 2 \left ( \nts - \nsv \right )$.  So, with $P = 2 \nts - 1$, deg$(\times) \geq \pd + 1 - \nsv$ and $\Gamma$ is a hyper-Medusa.  Hence, by Lemma \ref{lem: rhohm}, there exists a linear combination, $L_M$, of Medusas, such that $\rho^{(0)} \left ( X_R + \tilde{X}_R \right ) = \rho^{(0)} \left ( L_M \right )$.  Therefore, combined with \eqref{eq: 2p-1}, we obtain $\delta_{2 \nts - 1} ( L ) = \rho^{(0)} ( L_M )$, which proves that $L$ is a $( 2 \nts - 1 )$-invariant.
\end{proof}

	We end our search for $\pd$-invariants with a summary of all invariants that we found:

\begin{Thm} \label{thm: superposition summary}
	For fixed $\nv$, the superposition of any exact $\pd_E$-invariant with the superposition of $\nts$ minimal loopless 1-invariants results in a $\pd$-invariant, provided $\pd_E + 2 \nts \geq \pd$.\footnote{This theorem applies even for $\pd_E < 0$.  Recall that, by Corollary \ref{lem: Pnegative}, an exact $\pd$-invariant for $\pd < 0$ is just any possible linear combination of plain-graphs.  For example, the plain-graph consisting only of empty vertices is an exact $\pd$-invariant for any $\pd < 0$.  } 
\end{Thm}

\noindent We conjecture that the above theorem captures all $\pd$-invariants, up to total derivatives.  Since we have classified all exact invariants and all 1-invariants, it is straightforward to construct the $\pd$-invariants in the above theorem for any specific case.

	Note that we have classified exact invariants and 1-invariants using the two parameters, $\nv$ (number of vertices) and $\Delta$ (number of edges).  Finite connected graphs can always be embedded on a Riemann surface of some genus, in which case Euler's theorem relates $\nv$, $\Delta$ and the number of faces $\nf$ of the embedding to the genus $g$ of the surface.  Therefore, one could also use the parameters $\nv$ and $\nf$ instead to classify invariants\footnote{In principle, $\Delta$ and $\nf$ could also be used, but this seems less natural.}.  Any finite graph that can be embedded into a 2-sphere can also be embedded into a plane, and is known as a planar graph.  In particular, this is true for any graph with $\nv \leq 4$ and also for any tree.  Furthermore, any superposition of planar graphs is again a planar graph.  The number of faces of these planar graphs is exactly the number of ``loops'' when the graph is interpreted as a Feynman diagram\footnote{Thanks to Kurt Hinterbichler for bringing this issue to our attention at the 2014 BCTP Tahoe Summit.}.  One can check this statement for all of the examples in \S\ref{sec: intro to superposition} since they were all generated by superposition of planar graphs.  Note that, in \S\ref{sec: intro to superposition}, except for $\pd = 0$ (which is a trivial case), all superpositions involve trees, so that all superposed graphs are connected.  For example, all of the graphs in Figure \ref{fig: (4,6)} have three faces when embedded into a plane.  Indeed, when interpreted as Feynman diagrams, these graphs have three ``loops".  In general, the superposition of $\nts$ minimal loopless 1-invariants yields graphs have $\nf$ ``loops" as Feynman diagrams, where $\nf$ is given by $\nf = (\nts-1)(\nv-1)$.  Theorem \ref{thm: superposition summary} says that these superpositions will be $\pd$-invariant for $\pd \leq 2\nts-1$.  For example, the superposition of three minimal loopless 1-invariants with $\nv = 4$ produces a 5-invariant with 6 faces.

\subsection{Unlabeled Invariants} \label{sec: unlabeled} 

	So far we have been dealing entirely with labeled graphs, which represent algebraic terms where each $\phi$ is given a distinct label.  But we are primarily interested in invariants where all $\phi$'s are the same.  These are represented by \emph{unlabeled} graphs, that is, where isomorphic graphs are identified with each other.  One may wonder whether or not the labeled $\pd$-invariants capture all of the unlabeled ones.   The following proposition addresses this question, and shows that our restriction to labeled $\pd$-invariants still allows us to find all unlabeled $\pd$-invariants.

\begin{Prop} \label{prop: unlab_lab}
	Given an unlabeled $\pd$-invariant $L_{\emph{unlab}}$, there exists a labeled $\pd$-invariant $L_{\emph{lab}}$, such that $L_{\emph{lab}}$ reduces to an integer multiple of $L_{\emph{unlab}}$ once the labels are removed.
\end{Prop}

\begin{proof}
	Define $L^{ \times }_{\text{unlab}} = \delta ( L_{\text{unlab}} )$, where $\delta ( L_{\text{unlab}} )$ contains $\delta ( \Gamma_{\text{unlab}} )$ for all $\Gamma_{\text{unlab}}$ in $L_{\text{unlab}}$.  Label $L_{\text{unlab}}$ (i.e., label the vertices from 1 to $\nv$) to form $L_{\text{lab}}$.  This labeling is fiducial since we will eventually sum over all possible labelings.  Do the same for $L^{ \times }_{\text{unlab}}$ to form $L^{ \times }_{\text{lab}}$.  Define $L^{ \times }_{\text{lab}'} = \delta ( L_{\text{lab}} )$, which is a labeling of $L^{ \times }_{\text{unlab}}$, possibly distinct from $L^{ \times }_{\text{lab}}$.  However, if $\Gamma^{ \times }_{\text{lab}}$ in $L^{ \times }_{\text{lab}}$ and $\Gamma^{ \times }_{\text{lab}'}$ in $L^{ \times }_{\text{lab}'}$ reduce to the same $\Gamma^{ \times }_{\text{unlab}}$ in $L^{ \times }_{\text{unlab}}$ once the labels are removed, then $\Gamma^{ \times }_{\text{lab}}$ and $\Gamma^{ \times }_{\text{lab}'}$ are simply related by a permutation.  Therefore,
\begin{equation} \label{eq: lab and labprime}
	\sum_{\sigma \in S_\nv} \sigma \circ \delta ( L_{\text{lab}} ) = \sum_{\sigma \in S_\nv} \sigma ( L^{ \times }_{\text{lab}} ) = \sum_{\sigma \in S_\nv} \sigma ( L^{ \times }_{\text{lab}'} ),
\end{equation}

\noindent where $S_\nv$ is the group of permutations on the $\nv$ vertices.

Since $L_{\text{unlab}}$ is $\pd$-invariant, there exists $L^{ \star }_{\text{unlab}}$ such that $\delta ( L_{\text{unlab}} ) = \rho ( L^{ \star }_{\text{unlab}} )$.  Label $L^{ \star }_{\text{unlab}}$ to form $L^{ \star }_{\text{lab}}$ and define $L^{ \times }_{\text{lab}''} = \rho ( L^{ \star }_{\text{lab}} )$.  Generically, there can be cancellations between isomorphic graphs in $L^{ \times }_{\text{lab}''}$, once the labels are removed, since one $\times$-graph can be associated with more than one $\star$-graph.  Therefore, $L^{ \times }_{\text{lab}''}$ is not necessarily a labeling of $L^{ \times }_{\text{unlab}}$.  Nevertheless, if $\alpha \, \Gamma_{\text{unlab}}^{ \times }$ appears in $L_{\text{unlab}}^{ \times }$ with $\alpha \neq 0$, then all of the graphs, $\Gamma^{ \times }_{1} , \ldots , \Gamma^{ \times }_{k}$ in $L^{ \times }_{\text{lab}''}$ which are isomorphic to $\Gamma^{ \times }_{\text{unlab}}$ appear in $L^{ \times }_{\text{lab}''}$ as a linear combination $\sum_{i=1}^{k} \alpha^{}_i \, \Gamma^{ \times }_{i}$ with $\sum_{i=1}^{k} \alpha_i = \alpha$.  Conversely, if $\sum_{i=1}^{k} \alpha^{}_i \, \Gamma^{ \times }_{i}$ appears in $L^{ \times }_{\text{lab}''}$, but the graph, $\Gamma^{ \times }_{\text{unlab}}$, to which $\Gamma^{ \times }_{i}$ reduces once the labels are removed, does not appear in $L^{ \times }_{\text{unlab}}$, then $\sum_{i=1}^{k} \alpha_i = 0$.  Therefore,
\begin{equation*}
	\sum_{\sigma \in S_\nv} \sum_{i=1}^{k} \alpha_i \cdot \sigma ( \Gamma^{ \times }_{i} ) =
	\sum_{i=1}^{k} \alpha_k \sum_{\sigma \in S_\nv} \sigma ( \Gamma^{ \times }_{\text{lab}} ) = \alpha \sum_{\sigma \in S_\nv} \sigma ( \Gamma^{ \times }_{\text{lab}} ),
\end{equation*}

\noindent which implies
\begin{equation} \label{eq: lab and labpprime}
	\sum_{\sigma \in S_\nv} \sigma \circ \rho ( L^{ \star }_{\text{lab}} ) = \sum_{\sigma \in S_\nv} \sigma ( L^{ \times }_{\text{lab}''} ) = \sum_{\sigma \in S_\nv} \sigma ( L^{ \times }_{\text{lab}} ).
\end{equation}

\noindent Combining Eqs.  $\eqref{eq: lab and labprime}$ and $\eqref{eq: lab and labpprime}$ with the facts that $\sigma \circ \delta = \delta \circ \sigma$ and $\sigma \circ \rho = \rho \circ \sigma$ for each $\sigma \in S_\nv$, yields the desired labeled consistency equation:
\begin{equation*}
	\delta \biggl( \sum_{\sigma \in S_\nv} \sigma ( L_{\text{lab}} ) \biggr) = \rho \biggl( \sum_{\sigma \in S_\nv} \sigma ( L^{ \star }_{\text{lab}} ) \biggr).
\end{equation*}

\noindent $\sum_{\sigma \in S_\nv} \sigma ( L_{\text{lab}} )$ is $\pd$-invariant and reduces to $\nv ! \, L_{\text{unlab}}$ once the labels are removed.
\end{proof}


\section{Coset Construction} \label{AppC}
	
	The standard technique for finding terms which are invariant under a nonlinear realization of a symmetry is to use a coset construction \cite{pl1, pl2, spacetime_symmetry, nonlinear_realization}.  In this appendix we explore the connection between our invariant Lagrangians and this construction.  Although we find that the coset construction can reproduce some of the invariants that we have discovered, using this method is computationally difficult when compared to the graphical method introduced in this paper.  
	
	We first review the coset construction applied to the polynomial shift symmetry as presented in \cite{galileon_WZ, extended_shift}. In dimension $D$, consider the polynomial shift transformations of the Goldstone fields defined in (\ref{eq: polyshift}), accompanied with space-time translations and spatial rotations.  Denote the corresponding generators by $Z, Z^i, \ldots, Z^{i_1 \ldots i_\pd}$ for polynomial shifts, $\mathcal{P}_i$ for spacial translations, $\mathcal{P}_0$ for temporal translations, and $J^{ij}$ for spatial rotations.  Note that there are no boost symmetries.  For the nonlinear realization of space-time symmetry, the translation generators are treated as the broken generators \cite{spacetime_symmetry, nonlinear_realization}.  The Goldstone fields transform as
\begin{align*}
    \delta_{\mathcal{P}_i} \phi & = \partial_i \phi, &%
    \delta_Z \phi & = 1, &%
    \delta_{Z^{i_1 \ldots i_k}} \phi & = \frac{1}{k!} x^{i_1} \ldots x^{i_k},\hspace{2mm} k = 1, \ldots, \pd.
\end{align*}

\noindent The commutators between the operators can be readily calculated,
\begin{align*} 
    \left [ \mathcal{P}_i, Z \right ] & = 0, &%
    \left [ \mathcal{P}_i, Z^{i_1 \ldots i_k} \right ] & = - i \sum_{j} \frac{1}{k} \, \delta_{i}^{j} Z^{i_1 \ldots \hat{\textit{\j}} \ldots i_k}, &%
    \left [ Z^{i_1 \ldots i_k}, Z^{j_1 \ldots j_\ell} \right ] & = 0,
\end{align*}

\noindent where $\hat{\textit{\j}}$ means that the index $j$ is omitted.

    The commutators above given by the generators $\mathcal{P}_0$, $\mathcal{P}_i$, $J^{ij}$, $Z$ and $Z^{i_1 \ldots i_k}$, $k = 1, \ldots, \pd$ define the Lie algebra of a Lie group $G$, and $\mathcal{P}_0$ and $J^{ij}$ correspond to the unbroken normal subgroup $H$.  Take left invariant differential $\nv$-forms on $G/H$ to be $\nv$-cochains, and take the coboundary operator $d^{(k)}$ to be the exterior derivative of differential forms.  Denote the group of $\nv$-cocycles by $\mathcal{Z}^k = \text{Ker } d^{(k)}$, and the group of $k$-coboundaries by $\mathcal{B}^k = \text{Im } d^{(k-1)}$.  The Chevalley-Eilenberg cohomology group $\mathcal{E}^k (G/H)$ is defined to be
\begin{equation*}
	\mathcal{E}^k (G/H) = \mathcal{Z}^k / \mathcal{B}^k,
\end{equation*}

\noindent which is isomorphic to the Lie algebra cohomology $\mathcal{H}^k_0 (G/H; \mathbb{Z})$.  (See \cite{lie_cohomology} for details.)

    We associate the generator $Z$ with the Goldstone field $\phi$.  To each generator $Z^{i_1 \ldots i_k}$, $k = 1, \ldots, \pd$ we associate a symmetric $k$-tensor field $\phi_{i_1 \ldots i_k}$.  Indices can be lowered or raised by Kronecker delta symbols.  The coset space is parametrized by
\begin{equation*}
	g = \exp \left ( i \mathcal{P}_i \, x^i \right ) \exp \left (i Z \phi + i \sum_{k=1}^{\pd} Z^{i_1 \ldots i_k} \phi_{i_1 \ldots i_k} \right ).
\end{equation*}
\noindent The Maurer-Cartan form is
\begin{equation*}
	- i g^{-1} d g = \mathcal{P}_i dx^i + Z ( d \phi + \phi_i \, d x^i ) + \sum_{n=1}^{\pd -1} Z^{i_1 \cdots i_n} \bigl( d \phi_{i_1 \cdots i_n} + \phi_{i_1 \cdots i_n i} dx^i \bigr) + Z^{i_1 \ldots i_\pd} d \phi_{i_1 \ldots i_\pd}.
\end{equation*}
\noindent Therefore, the basis dual to the generators is
\begin{align} \label{eq: 1-forms} 
	\omega_{\mathcal{P}}^{i} &= dx^i, &%
	\omega_{i_1 \cdots i_\pd} &= d \phi_{i_1 \cdots i_\pd}, \notag \\%
	\omega_{\phantom{P}}^{\phantom{i}} &= d \phi + \phi_i \, d x^i, &%
	\omega_{i_1 \cdots i_k} &= d \phi_{i_1 \cdots i_k} + \phi_{i_1 \cdots i_k i} \, dx^i,\hspace{5mm} k = 1, \ldots, \pd -1.
\end{align}

\noindent Moreover,
\begin{align*}
	d \omega_{\mathcal{P}}^{i} &= 0, &%
	d \omega_{i_1 \cdots i_\pd} &= 0,\\
	d \omega_{\phantom{P}}^{\phantom{i}} &= d \phi_i \wedge dx^i, &%
	d \omega_{i_1 \cdots i_k} &= d \phi_{i_1 \cdots i_k i} \wedge dx^i,\hspace{5mm} k = 1, \ldots, \pd -1.
\end{align*}
\noindent The inverse Higgs constraints \cite{inverse_higgs_1975,inverse_higgs} imply the vanishing of $\omega$ and $\omega_{i_1 \cdots i_k}$ ($k < \pd$) in \eqref{eq: 1-forms}:
\begin{equation} \label{eq: inverse Higgs}
	\phi_{i_1 \cdots i_k} = (-1)^k \partial_{i_1} \cdots \partial_{i_k} \phi, \hspace{5mm} k = 1, \ldots, \pd.
\end{equation}

	Having reviewed the coset construction, we now present examples for $\pd =2$ and $3$ (the $\pd =1$ scenario is essentially the same as the Galileon case \cite{galileon_WZ}). 

\vspace{0.3cm}

\noindent \textbf{\textit{\pd}=2 Case}: For $\nv =3$, the cohomology group is trivial for $\sdim <2$.  Therefore, let us start with the simplest nontrivial case, $\sdim =2$.  We are looking for a closed form involving the wedge of three $\omega$'s, which are not $\omega_\mathcal{P}$.  There is one independent cohomology element, with the lowest number of indices on $\omega$'s:
\begin{equation*} 
    \Omega_3 = \epsilon_{ij} \, \omega_{ab} \wedge \omega_{ia} \wedge \omega_{jb} = d \bigl( \epsilon_{ij} \phi_{ab} \, d \phi_{ia} \wedge d \phi_{jb} \bigr) \equiv d \beta_2.
\end{equation*}

\noindent These expressions can be extended to $\sdim \geq 2$:
\begin{align*}
    \Omega_{\sdim +1} &= \epsilon_{ijs^{}_3 \ldots s^{}_\sdim} \, \omega_{ab} \wedge \omega_{ia} \wedge \omega_{jb} \wedge \omega_{\mathcal{P}}^{s^{}_3} \wedge \cdots \wedge \omega_{\mathcal{P}}^{s^{}_\sdim}, \\
    \beta_\sdim &= \epsilon_{ijs^{}_3 \ldots s^{}_\sdim} \, \phi_{ab} \, d \phi_{ia} \wedge d \phi_{jb} \wedge d x^{s^{}_3} \wedge \cdots \wedge d x^{s^{}_\sdim}.
\end{align*}

\noindent Taking the pullback of $\beta_\sdim$ to the spacetime manifold and then applying the inverse Higgs constraints in \eqref{eq: inverse Higgs} gives a term proportional to
\begin{equation*} 
    \epsilon_{ijs^{}_3 \ldots s^{}_\sdim} \epsilon_{k \ell s^{}_3 \ldots s^{}_\sdim} \, \partial_a \partial_b \phi \, \partial_i \partial_k \partial_a \phi \, \partial_j \partial_{\ell} \partial_b \phi,
\end{equation*}

\noindent which is already contained in \cite{extended_shift}.  One can verify that this is equivalent up to integration by parts and overall prefactor to the invariant in \eqref{eq: p=2 n=3 Delta=4}.  That term was found to be invariant for $\pd =3$, and is therefore also invariant for $\pd =2$.

	Next, consider $\nv =4$.  The simplest nontrivial case is $\sdim =3$.  We seek a closed 4-form given as the wedge of four $\omega$'s, which are not $\omega_\mathcal{P}$.  There is one independent cohomology element, with the lowest number of indices on $\omega$'s:
\begin{align*}
    \Omega_4^\prime = & \epsilon_{ijk} \, \omega_a \wedge \omega_{ia} \wedge \omega_{jb} \wedge \omega_{kb} \notag \\
    = & d \biggl[ \epsilon_{ijk} \phi_{ac} \biggl( \frac{1}{2} \phi_{ac} \, d \phi_{jb} + \phi_{bc} \, d \phi_{ja} \biggr) \wedge d \phi_{kb} \wedge dx_i \biggr]  \equiv d \beta_3^\prime
\end{align*}

\noindent These expressions can be extended to $\sdim \geq 3$:
\begin{align*}
    \Omega_{\sdim +1}^\prime &= \epsilon_{ijks^{}_4 \ldots s^{}_\sdim} \, \omega_a \wedge \omega_{ia} \wedge \omega_{jb} \wedge \omega_{kb} \wedge \omega^{s^{}_4}_\mathcal{P} \wedge \cdots \wedge \omega^{s^{}_\sdim}_\mathcal{P}, \\
    \beta_\sdim^\prime &= \epsilon_{ijks^{}_4 \ldots s^{}_\sdim} \phi_{ac} \biggl( \frac{1}{2} \phi_{ac} \, d \phi_{jb} + \phi_{bc} \, d \phi_{ja} \biggr) \wedge d \phi_{kb} \wedge dx^i \wedge dx^{s^{}_4} \wedge \cdots \wedge dx^{s^{}_\sdim}.
\end{align*}

\noindent Taking the pullback of $\beta_\sdim^\prime$ to the spacetime manifold and then applying the inverse Higgs constraints in \eqref{eq: inverse Higgs} gives a term proportional to
\begin{equation*}
    \epsilon_{ijs^{}_3 \ldots s^{}_\sdim} \epsilon_{k \ell s^{}_3 \ldots s^{}_\sdim} \, \partial_{(a} \partial_{\raisebox{-1pt}{$\scriptstyle{b}$}} \phi \, \partial_{c)}\partial_j \partial_{\ell} \phi \, \partial_a \partial_{b} \phi \, \partial_c \partial_i \partial_{k} \phi.
\end{equation*}

\noindent One can verify that this is equivalent up to integration by parts and an overall prefactor to the invariant in \eqref{eq: 245}.  

\vspace{0.3cm}

\noindent \textbf{\textit{\pd}=3 Case}: Let us focus on $\nv =3$.  Again, we start with $\sdim =2$.  There is one independent cohomology element:
\begin{align*}
    \Omega_3 = \epsilon_{ij} \bigl( & \omega_{ab} \wedge \omega_{ia} \wedge \omega_{jb} + 2 \omega_a \wedge \omega_{ib} \wedge \omega_{jab} - 2 \omega_a \wedge \omega_{ia} \wedge \omega_{jbb} \notag \\
    & - \omega \wedge \omega_{iab} \wedge \omega_{jab} + \omega \wedge \omega_{iaa} \wedge \omega_{jbb} \bigr).
\end{align*}
\noindent It is quite a challenge to determine the potential, $\beta_2$, for this $\Omega_3$.  One can appreciate the power of the graphical method at this point: We have already determined that there is one independent $3$-invariant with $\nv =3$.  Therefore, we can immediately conclude without calculation that the pullback of $\beta_2$ must be proportional to the invariant in \eqref{eq: n=3 graph} up to total derivatives.  Again, $\Omega_3$ can be generalized to $\sdim \geq 2$ by wedging the appropriate number of $\omega_\mathcal{P}$'s on the end:
\begin{align*}
    \Omega_{\sdim +1} = \epsilon_{ijs^{}_3 \ldots s^{}_\sdim} \bigl( & \omega_{ab} \wedge \omega_{ia} \wedge \omega_{jb} + 2 \omega_a \wedge \omega_{ib} \wedge \omega_{jab} - 2 \omega_a \wedge \omega_{ia} \wedge \omega_{jbb} \notag \\
    & - \omega \wedge \omega_{iab} \wedge \omega_{jab} + \omega \wedge \omega_{iaa} \wedge \omega_{jbb}  \bigr) \wedge \omega_{\mathcal{P}}^{s^{}_3} \wedge \ldots \wedge \omega_{\mathcal{P}}^{s^{}_\sdim}.
\end{align*}

\bibliographystyle{JHEP}
\bibliography{pol}
\end{document}